\documentclass[12pt,english]{article}
\usepackage{mathtools,float}

\allowdisplaybreaks[4]
\usepackage{setspace}

\usepackage[dvipsnames]{xcolor}
\usepackage{geometry}
\usepackage{color,tikz}
\usepackage{amsmath}
\usepackage{amsthm}
\usepackage{amssymb}
\usepackage[authoryear]{natbib}
\usepackage{booktabs}
\usepackage{graphicx}
\usepackage{tikz}
\usetikzlibrary{arrows,decorations.markings}
\usepackage{verbatim}

\tikzset{%
  ->-/.style={decoration={markings, mark=at position 0.7 with {\arrow{stealth},bend=right}},
              postaction={decorate}}
}
\usepackage{titletoc}

\makeatletter
 \theoremstyle{definition}
  
   \theoremstyle{plain}
  \newtheorem{assumption}{\protect\assumptionname}
\theoremstyle{plain}
\newtheorem{thm}{\protect\theoremname}
\theoremstyle{plain}

  \theoremstyle{plain}
  
    \theoremstyle{definition}
  \newtheorem{remark}{\protect\remarkname}
    \theoremstyle{plain}
  \newtheorem{lem}{\protect\lemmaname}

    \theoremstyle{plain}
  \newtheorem{algo}{\protect\algorithmname}
  
\definecolor{DarkBlue}{rgb}{0,0,0.6}
\usepackage{hyperref}
\hypersetup{
     colorlinks   = true,
     citecolor    = Maroon,
     urlcolor= DarkBlue,
     linkcolor=DarkBlue
}

\makeatother

\usepackage{babel}
  \providecommand{\assumptionname}{Assumption}
  \providecommand{\examplename}{Example}
  \providecommand{\propositionname}{Proposition}
\providecommand{\theoremname}{Theorem}
\providecommand{\remarkname}{Remark}
\providecommand{\corollaryname}{Corollary}
  \providecommand{\lemmaname}{Lemma}
    \providecommand{\algorithmname}{Algorithm}

\renewcommand{\baselinestretch}{1.1}\small\normalsize
\global\long\def\betamix{\beta_{{\rm mixing}}}

\global\long\def\Acal{\mathcal{A}}

\global\long\def\tF{\tilde{F}}

\global\long\def\hF{\hat{F}}

\global\long\def\bF{\bar{F}}

\global\long\def\cF{\check{F}}

\global\long\def\tH{\widetilde{H}}

\global\long\def\Mcal{\mathcal{M}}

\global\long\def\bZ{\bar{Z}}

\global\long\def\Zb{\mathbf{Z}}

\global\long\def\oneb{\mathbf{1}}

\global\long\def\hSigma{\hat{\Sigma}}

\global\long\def\hmu{\hat{\mu}}

\global\long\def\hQ{\hat{Q}}

\global\long\def\hbeta{\hat{\beta}}

\global\long\def\Wcal{\mathcal{W}}

\global\long\def\Hcal{\mathcal{H}}

\global\long\def\RR{\mathbb{R}}

\global\long\def\heps{\hat{\varepsilon}}

\global\long\def\Bcal{\mathcal{B}}

\global\long\def\bG{\bar{G}}

\global\long\def\Qcal{\mathcal{Q}}

\global\long\def\hu{\hat{u}}

\global\long\def\tbeta{\tilde{\beta}}

\global\long\def\Acal{\mathcal{A}}

\global\long\def\RR{\mathbb{R}}

\global\long\def\hOmega{\hat{\Omega}}

\global\long\def\tOmega{\tilde{\Omega}}

\global\long\def\tmu{\tilde{\mu}}

\global\long\def\trace{\text{trace}}

\global\long\def\tW{\tilde{W}}

\def\spacingset#1{\renewcommand{\baselinestretch}%
	{#1}\small\normalsize} \spacingset{1}

\newcommand{\blind}{1}

  \usepackage{xr}

\begin{document}

\if1\blind
{
  \title{\bf An Exact and Robust Conformal Inference Method for Counterfactual and Synthetic Controls}
  \author{Victor Chernozhukov\thanks{Massachusetts Institute of Technology; 50 Memorial Drive, E52-361B, Cambridge, MA 02142, USA; Email: vchern@mit.edu} \qquad Kaspar W\"uthrich\thanks{Department of Economics, University of California San Diego, 9500 Gilman Dr., La Jolla, CA 92093, USA; Email: kwuthrich@ucsd.edu} \qquad Yinchu Zhu\thanks{Brandeis University; 415 South Street, Waltham, MA 02453, USA; Email: yinchuzhu@brandeis.edu}}
 
} \fi

\if0\blind
{

  \title{\bf An Exact and Robust Conformal Inference Method for Counterfactual and Synthetic Controls}
  \date{ }
} \fi

\maketitle

\begin{abstract}
We introduce new inference procedures for counterfactual and synthetic control methods for policy evaluation. We recast the causal inference problem as a counterfactual prediction and a structural breaks testing problem. This allows us to exploit insights from conformal prediction and structural breaks testing to develop permutation inference procedures that accommodate modern high-dimensional estimators, are valid under weak and easy-to-verify conditions, and are provably robust against misspecification. Our methods work in conjunction with many different approaches for predicting counterfactual mean outcomes in the absence of the policy intervention. Examples include synthetic controls, difference-in-differences, factor and matrix completion models, and (fused) time series panel data models. Our approach demonstrates an excellent small-sample performance in simulations and is taken to a data application where we re-evaluate the consequences of decriminalizing indoor prostitution. Open-source software for implementing our conformal inference methods is available.

\bigskip


\vfill

\end{abstract}

\newpage
\spacingset{1.5} % DON'T change the spacing!
\pagenumbering{arabic}% Arabic page numbers (and reset to 1)

\section{Introduction}
We consider the problem of making inferences on the causal effect of a policy intervention in an aggregate time series setup with a single treated unit. The treated unit is observed for $T_0$ periods before and $T_\ast$ periods after the intervention occurs. Often, there is additional information in the form of possibly very many untreated units, which can serve as controls. Such settings are ubiquitous in applied research, and there are many different approaches for estimating the causal effect of the policy. Examples include difference-in-differences methods, synthetic control (SC) approaches, factor, matrix completion, and interactive fixed effects (FE) models, and times series models.\footnote{We refer to \citet{DI16}, \citet{gobillon2016regional}, and \citet{abadie2019jel} for excellent comparative overviews and reviews.} 
We refer to these methods as counterfactual and synthetic control (CSC) methods. 

This paper provides generic and robust procedures for making inferences on policy effects estimated by CSC methods. We propose a general counterfactual modeling framework that nests and generalizes many traditional and new methods for counterfactual analysis. We focus on methods that are able to generate mean-unbiased proxies, $P_t^N$, for the counterfactual outcomes of the treated unit in the absence of the policy intervention, $Y_{1t}^N$: 
\[
Y_{1t}^N=P_t^N+u_t, \quad E\left(u_t \right)=0, \quad t=1,\dots,T_0+T_\ast.
\] 
The policy effect in period $t$ is $\theta_t=Y_{1t}^I-Y_{1t}^N$, where $Y_{1t}^I$ is the counterfactual outcome of the treated unit with the policy intervention. We are interested in testing hypotheses about the trajectory of policy effects in the post-treatment period, $\theta=\{\theta_t\}_{t=T_0+1}^{T_0+T_\ast}$. Specifically, we postulate a trajectory $\theta^0= \{ \theta^0_t\}_{t=T_0+1}^{T_0+T_\ast}$ and test the sharp null hypothesis that $\theta=\theta^0$. We also consider the problem of testing hypotheses about per-period effects $\theta_t$ and propose a simple algorithm for constructing pointwise confidence intervals via test inversion. 

We recast the inference problem as a (counterfactual) prediction and a structural breaks testing problem. This allows us to build on the literature on conformal prediction \citep{vovk2005algorithmic} and end-of-sample stability testing \citep[][]{dufour1994generalized,andrews2003end} to construct inference procedures that are provably robust against misspecification and accommodate many classical and modern high-dimensional methods for estimating $P_t^N$. 

 The basic idea of our testing procedures is as follows. Suppose that there is only one post-treatment period and that $P_t^N$ is known. Under the sharp null that $\theta_{T_0+1}=\theta_{T_0+1}^0$, we can compute $Y_{1t}^N$ and $u_t=Y_{1t}^N-P_t^N$ for all time periods. If the stochastic shock process $\{u_t\}$ is stationary and weakly dependent, and its distribution is invariant under the intervention, the distribution of the error in the post-treatment period, $u_{T_0+1}$, should be the same as the distribution of the errors in the pre-treatment period, $\{u_t\}_{t=1}^{T_0}$. We operationalize this idea by proposing inference methods in which $p$-values are obtained by permuting blocks of estimated residuals across the time series dimension.

The proposed methods are valid under two different sets of conditions:
\begin{itemize} \setlength{\parskip}{0pt}
\item[\textbf{(i)}] \textbf{Estimator Consistency and Stationary Weakly Dependent Errors}

If the data exhibit dynamics, trends, and serial dependence but the stochastic shock sequence $\{u_t\}$ is stationary and weakly dependent, our inference procedures are approximately valid if the estimator of $P_t^N$ is consistent (pointwise and in prediction norm). Consistency can be verified for many different CSC methods. We provide concrete sufficient conditions for a representative selection of methods, including difference-in-differences, SC, factor models, matrix completion and interactive FE models, linear and nonlinear time series models, and fused time series panel data models.

\item[\textbf{(ii)}] \textbf{Estimator Stability and Stationary Weakly Dependent Data}

In practice, misspecification is an important concern. We show that even if the model for $P_t^N$ is misspecified and the estimator of $P_t^N$, $\hat{P}_t^N$, is inconsistent, our procedures are still valid, provided that the data are stationary and weakly dependent and $\hat{P}_t^N$ satisfies a \emph{stability} condition. This condition requires that $\hat{P}_t^N$ is stable under perturbations in a few observations. It is implied, for instance, if $\hat{P}_t^N$ is consistent for a pseudo-true parameter value but is shown to hold even in high-dimensional settings where consistency results under misspecification are often not available.

\end{itemize}

The main theoretical results in this paper are finite sample (non-asymptotic) bounds on the size accuracy of our methods; these bounds imply that our methods are exact as $T_0\rightarrow \infty$. Unlike traditional asymptotic results, which are only informative when the sample size is large enough, our non-asymptotic bounds show how different factors affect the finite sample performance. This feature is relevant in CSC applications where sample sizes are often small.

A key feature of our conformal inference methods is that $P_t^N$ is estimated under the null hypothesis based on data from all $T_0+T_\ast$ periods. Estimation under the null guarantees the exact finite sample validity of our procedures if the data are iid or exchangeable. Even when exchangeability fails, imposing the null for estimation is essential for a good performance in CSC applications where $T_0$ is often rather small. Figure \ref{fig:importance_H0} plots the empirical rejection probabilities for testing the null that $\theta_{T_0+1}=0$ when $T_0=19$, $J=50$ (as in our empirical application), $\{u_t\}$ is an AR(1) process, and $P_t^N$ is estimated using SC. The size properties of our method are excellent. By contrast, estimating $P_t^N$ based on the $T_0$ pre-treatment periods without imposing the null yields substantial size distortions. Figure \ref{fig:importance_H0} suggests that imposing the null continues to improve size accuracy even when exchangeability fails and that these improvements can be substantial in small samples.

\begin{center}
[Figure \ref{fig:importance_H0} around here.]
\end{center}

We make two additional contributions that may be of independent interest. First, we introduce the $\ell_1$-constrained least squares estimator or \emph{constrained Lasso} \citep[e.g.,][]{raskutti2011minimax} as an essentially tuning-free alternative to existing penalized regression estimators and study its theoretical properties. Constrained Lasso nests SC and difference-in-differences, providing a unifying approach for the regression-based estimation of the mean proxies $P_t^N$. Second, we obtain theoretical consistency results for SC estimators in settings with potentially very many control units.

We develop three extensions of our main results. First, we show that our method can be modified to test hypotheses about average effects over time. Second, we extend our method to settings with multiple treated units. Third, we propose easy-to-implement placebo tests for assessing the credibility of inferences based on our method. 

Monte Carlo simulations suggest that our procedures exhibit excellent size properties and are robust to misspecification. We find that imposing additional constraints (e.g., using SC instead of the more general constrained Lasso) does not improve power when these additional restrictions are correct but can cause power losses when they are not.

Finally, we re-analyze the causal effect of decriminalizing indoor prostitution on sexually transmitted infections. Following \citet{cunningham2018decriminalizing}, we exploit the unanticipated decriminalization of indoor prostitution in Rhode Island in 2003. We find that decriminalizing indoor prostitution significantly decreased the incidence of female gonorrhea.

\paragraph{Related Literature.} We contribute to the literature on inference procedures for CSC methods with few treated units.  
A popular method is the finite population permutation approach of \citet{abadie10sc}, see also \citet{firpo18synthetic} and \citet{abadie2019jel}. This approach permutes the policy assignment  and relies on permutation distributions for inference. It corresponds to conventional randomization inference \citep{fisher1935design} under random assignment of the policy \citep[e.g.,][]{abadie10sc,abadie2019jel}. However, random assignment is not plausible in typical CSC applications, and assignment mechanisms are difficult to model and estimate when there are only few treated units \citep[e.g.,][]{abadie2019jel}.  \citet{shaikh2019randomization} propose randomization tests for settings with staggered treatment adoption, which encompass the approach of \citet{abadie10sc}. The main assumption of \citet{shaikh2019randomization}'s approach is that policy adoption follows a Cox proportional hazards model. We do not model the assignment mechanism. Instead, we exploit stationarity and weak dependence of the errors across time in a repeated sampling framework. One advantage of exploiting the time series dimension is that we only require a suitable model for the potential outcome of the treated unit. By contrast, cross-sectional approaches often require estimating models for all units. That is, we only require a good ``local'' instead of a good ``global'' fit, which reduces the risk of model misspecification. On the other hand, our approach requires a large number of pre-treatment periods and relies on invariance of the error distribution under the intervention.

There is also an active literature on asymptotic inference methods for CSC models. Several papers focus on testing hypotheses about average or expected effects over time, requiring $T_0$ and $T_\ast$ to be large. \citet{li2017estimation}, \citet{carvalho2018arco}, \citet{chernozhukov2019ttest}, and \citet{li2020statistical} introduce inference methods based on penalized and constrained regression methods. \citet{arkhangelsky2018synthetic} propose inference methods for a version of SC with time and unit weights, which admits a weighted regression formulation. Asymptotic inference methods based on factor and interactive FE models are proposed by \citet{hsiao2012panel}, \citet{gobillon2016regional}, \citet{chan2016policy}, \citet{li2017estimation}, \citet{xu2017generalized}, and \citet{li2018inference}. 
Here we focus on sharp null hypotheses and permutation distributions and provide non-asymptotic performance guarantees. Our approach is generic and valid with many different methods, including constrained regression, factor models, and interactive FE estimators. \citet{conley2011inference} propose inference methods for difference-in-difference settings with few treated units. They exploit the cross-sectional dimension, relying on weak dependence and stationarity of the error terms across units, which may be hard to justify in typical CSC settings. By contrast, our procedures rely on stationarity and weak dependence of the errors  over time. On the other hand, exploiting the time series dimension, our approach requires $T_0$ to be large, whereas \citet{conley2011inference} allow $T_0$ to be fixed.
In related work, \citet{cattaneo2021prediction} provide prediction intervals for per-period effects estimated by SC methods. Their key observation is that there is both randomness from estimating the SC weights and from the prediction error. They propose a sampling-based inference method based on non-asymptotic probability bounds that accounts for both types of randomness and is valid with stationary and non-stationary data.

We recast the causal inference problem as a (counterfactual) prediction problem and build on the literature on conformal prediction \citep[e.g.,][]{vovk2005algorithmic,vovk2009online,lei2013distribution,lei2014distribution,lei2017distributionfree} and on the literature on permutation tests \citep[e.g.,][]{romano1990behavior,lehmann2005testing}, which was started by \citet{fisher1935design} in the context of randomization; see \citet{rubin1984bayesianly} for a Bayesian justification. On a more general level, our approach is also connected to transformation-based approaches to model-free prediction \citep[e.g.,][]{politis2015modelfree}. Let us discuss in more detail the relationship to \citet[][CWZ18 henceforth]{chernozhukov2018exact}, who extend classical conformal prediction to time series settings. Besides a different focus (prediction intervals for future outcome values vs.\ inference on policy effects), there are several important differences. First, we rely on permuting residuals, whereas CWZ18 permute blocks of data. Second, we theoretically analyze different types of permutations. In particular, we study the set of all permutations, which yields precise $p$-values in small samples. This set of permutations cannot be used in the framework of CWZ18 unless the data are iid. Third, we allow for non-stationary data, whereas the prediction methods in CWZ18 are strictly limited to stationary data. Forth, CWZ18 rely on abstract high-level conditions on the test statistics and do not provide any primitive conditions. By contrast, we develop transparent sufficient conditions that can be verified for many traditional and modern CSC methods, and we provide explicit primitive conditions for a large selection of popular approaches. Finally, we establish the validity of our methods with time series data under misspecification and stability, whereas the theoretical results for weakly dependent data in CWZ18 require correct specification and consistency.

Finally, we show that the problem of making inferences on policy effects can be recast as a structural breaks testing problem with a known break date. Therefore, we build on and contribute to the literature on testing for structural breaks and, in particular, to the literature on structural breaks testing using permutation approaches \citep[e.g.,][]{antoch2001permutation,zeileis2013toolbox}. Besides a different focus (inference on policy effects vs.\ testing for structural breaks), our paper differs from the existing literature in that we specifically focus on testing at the end of the sample, allow for a very general class of estimators, including modern high-dimensional methods, and provide non-asymptotic performance guarantees under correct specification and misspecification. Our paper is also related to tests for structural breaks at the end of the sample \citep[e.g.,][]{dufour1994generalized,andrews2003end}.\footnote{\citet{hahn2017synthetic} informally suggest applying a variant of \citet{andrews2003end}'s end-of-sample stability test in the context of SC, and \citet{ferman2019inference} use a version of this test in the context of difference-in-differences approaches with few treated groups.} Let us discuss the differences to \citet{andrews2003end}'s end-of-sample instability test based on subsampling in more detail. First, we focus on causal inference, whereas \citet{andrews2003end} is concerned with structural breaks testing. Second, our procedures are exactly valid under exchangeability, and we obtain finite sample bounds under weak conditions on the estimators, while the theoretical properties of \citet{andrews2003end}'s test rely on asymptotic analyses. Third, our methods are valid under misspecification, whereas \citet{andrews2003end} assumes correct specification. Forth, our results under correct specification only require stationarity and weak dependence of $\{u_t\}$, while \citet{andrews2003end}'s test assumes stationarity of the data.\footnote{ \citet{andrews2003end} briefly comments on page 1681 (comment 4) that his test can be shown to be asymptotically valid under stationary errors but does not provide a formal result.} Finally, our procedures work in conjunction with many modern high-dimensional estimators, whereas \citet{andrews2003end} focuses on low-dimensional GMM models. 
 
\paragraph{Notation.} For $q\geq 1$, the $\ell_q$-norm of a vector is denoted by $\| \cdot\|_{q}$. We use $\|\cdot\|_0$ to denote the number of nonzero entries of a vector; $\|\cdot\|_{\infty}$ is used to denote the maximal absolute value of entries of a vector. We use the
notation $a\lesssim b$ to denote $a \leq cb$ for some constant $c > 0$ that does not depend on the sample size. We use the notation $a \asymp b $ to denote $a\lesssim b$ and $b\lesssim a$. For a set $A$, $|A|$ denotes the cardinality of $A$.  For any $a\in\mathbb{R} $, we define $\lfloor a \rfloor =\max\{z\in\mathbb{Z}:z\leq a\} $  and $\lceil a \rceil =\min\{z\in\mathbb{Z}:z\geq a\} $, where $\mathbb{Z}$ is the set of integers. We use $\mathbb{N}$ to denote the set of natural numbers.

\section{A Conformal Inference Method}
\label{sec:conformal_inference}
\subsection{The Counterfactual Model}

We consider a time series of $T$ outcomes for a treated unit, labeled $j=1$.  During the first
$T_0$ periods, the unit is not treated by a policy and, during the remaining $T-T_0=T_\ast$ periods,  it is treated
by the policy. Extensions to more than one treated unit are discussed in the Appendix. Our typical setting is where $T_\ast$ is short compared to $T_0$. There may be other units that are not exposed to the policy, and they will be introduced below. We denote the observed outcome of the treated unit by $Y_{1t}$. We employ the potential (latent) outcomes framework \citep{neyman1923application,rubin1974estimating} and denote potential outcomes with and without the policy as $Y_{1t}^I$ and $Y_{1t}^N$. The effect of the policy intervention in period $t$ is $\theta_t=Y_{1t}^I-Y_{1t}^N$.

Our conformal inference method will rely on the following counterfactual modeling framework, which nests many traditional and new methods for counterfactual policy analysis; see Sections \ref{subsec:sc_panel}--\ref{subsec:ts_fused} for examples.

\begin{assumption}[Counterfactual Model]\label{ass:dgp} 
Let $\{P_t^N\}$ be a given sequence of mean-unbiased predictors or proxies for the counterfactual outcomes $\{Y_{1t}^N\}$ in the absence of the policy intervention, that is $\{E\left( P_t^N\right)\} = \{E \left( Y_{1t}^N\right)\}$. Let $\{\theta_t\}$ be a fixed policy effect sequence with $\theta_t = 0 $ for $ t \leq T_0$, so that potential outcomes under the intervention are given by $\{Y_{1t}^I\}  = \{Y_{1t}^N + \theta_t\}$.\footnote{In the Appendix, we consider an extension to random policy effects.}
In other words, potential outcomes can be written as
\begin{equation}\tag{CMF}
\begin{array}{l}
Y_{1t}^N = P_t^N + u_t\\
Y_{1t}^I = P_t^N + \theta_t + u_t  \\
\end{array}   \Bigg | \quad E( u_t) = 0,  \quad t=1,\dots,T , \\
\end{equation}
where $\{u_t\}$ is a centered stationary stochastic process. Observed outcomes are related to potential outcomes as $Y_{1t}=Y_{1t}^{N}+D_{t}\left(Y_{1t}^{I}-Y_{1t}^{N}\right)$, where $D_{t}=1\left(t>T_0 \right)$. 
\end{assumption}
Assumption \ref{ass:dgp} introduces the potential outcomes, but also postulates an identifying assumption in the form of the existence of mean-unbiased proxies $P^N_t$ such that $E \left( P_t^N\right) = E \left( Y_{1t}^N \right)$.  Assumption \ref{ass:dgp} allows $\{P_t^N\}$ to be fixed or random and does not impose any restrictions on the dependence between $\{P_t^N\}$ and $\{u_t\}$. In Sections \ref{subsec:sc_panel}--\ref{subsec:ts_fused}, we will discuss specific panel data and time series models that postulate (and identify) what $P_t^N$ is under a variety of conditions.  Additional assumptions on the stochastic shock process $\{u_t\}$ will be introduced later, in essence requiring $\{u_t\}$ to be either iid or, more generally, a stationary and weakly dependent process. 

Assumption \ref{ass:dgp} also postulates that the stochastic shock sequence is invariant under the intervention.  This is the fundamental identifying assumption. It requires that the timing of the policy intervention is independent of factors that change the distribution of $\{u_t\}$.\footnote{In principle, we can relax this assumption by specifying, for example, the scale and quantile shifts in the stochastic shocks that result from the policy, and then working with the resulting model; we leave this extension to future work.} If the policy changes the distribution of $\{u_t\}$, one can either interpret our method as a structural breaks test or view the policy effect as random in which case our method yields valid prediction sets; see the Appendix for details.

Often, there is additional information in the form of untreated units, which can serve as controls. Specifically, suppose that there are $J \geq 1$ control units, indexed by  $j=2,\dots,J+1$.  We assume that we observe all units for all $T$ periods, although this assumption can be relaxed. Let $Y_{jt}$ denote the observed outcome for these untreated units.  This observed outcome is equal to  the outcome in the absence of the policy intervention, i.e., $Y_{jt} = Y_{jt}^N$ for $2\le j\le  J+1$ and $1\le t\le T$. For each unit, we may also observe a vector of covariates $X_{jt}$. This motivates a variety of strategies for modeling and identifying $P_t^N$ as discussed below.

\subsection{Hypotheses of Interest, Test Statistics, and $p$-Values}
\label{subsec:hypotheses}
We are interested in testing hypotheses about the trajectory of policy effects in the post-treatment period, $\theta=\left(\theta_{T_0+1},\dots,\theta_{T}\right)'$. Our main hypothesis of interest is
\begin{eqnarray}\label{eq:H0}
H_0:~\theta=\theta^0,
\end{eqnarray}
where $\theta^0=\left(\theta^0_{T_0+1},\dots,\theta^0_{T}\right)'$ is a postulated policy effect trajectory. Hypothesis \eqref{eq:H0} is a sharp null hypothesis. It fully determines the value of the counterfactual outcome in the absence of the intervention in the post-treatment period since $Y_{1t}^N=Y^I_{1t}-\theta_t=Y_{1t}-\theta_t$. In the Appendix, we show that our method can also be used to test hypotheses about average effects.

To describe our procedure, we write the data under the null hypothesis as $\Zb:=\Zb(\theta^0)=(Z_1,\dots,Z_{T})'$, where
\[
Z_{t}=\begin{cases}
\left(Y^N_{1t},Y^N_{2t},\dots,Y^N_{J+1t},X'_{1t},\dots,X'_{J+1t}\right)', & t\le T_0\\
\left(Y_{1t}^I-\theta_t^{0},Y^N_{2t},\dots,Y^N_{J+1t},X'_{1t},\dots,X'_{J+1t}\right)', & t>T_0.
\end{cases}
\]

Using one of the methods described below, we will obtain a counterfactual proxy estimate, $\hat{P}^N_t$,
based on $\Zb$, and compute the residuals $\hat{u}=\left(\hat{u}_1,\dots,\hat{u}_T \right)'$, where $\hat{u}_t=Y_{1t}^N-\hat{P}^N_t$ for $1\le t\le  T$.
Since $P_t^N $ is computed using $\Zb=\Zb(\theta^0)$,  $P_t^N$ is estimated under the null hypothesis, which is essential for a good small sample performance. In ``ideal'' settings where the data are iid, imposing the null guarantees the model-free exact finite sample validity of our method; see the Appendix for details. By contrast, when $P_t^N$ is estimated based on the pre-treatment data $\{Z_t\}_{t=1}^{T_0}$ without imposing the null, permuting blocks of residuals does not yield procedures with exact finite sample validity, not even with iid data.

\smallskip

\noindent \textbf{Definition of Test Statistic $\mathbf{S}$.} We consider the following test statistic:
\[
S  (\hat u) = S_q(\hat u) = \left ( \frac{1}{\sqrt{T_{\ast}}} \sum_{t=T_0+1}^T|\hat{u}_t|^q \right)^{1/q}.
\]

\smallskip

Note that $S$ is constructed such that high values indicate rejection. Different choices of $q$ lead to power against different alternatives. For instance, if the intervention has a large but only temporary effect (i.e., if $|\theta_t|$ is large for few periods), choosing $q=\infty$ yields high power. On the other hand, if the intervention has a permanent effect (i.e., if $\theta_t$ is non-zero for many post-treatment periods), tests using $S_1$ or $S_2$ exhibit good power properties. In our application, we will be using $S_1$, which behaves well under heavy-tailed data. Throughout the paper, when the nature of the statistic is not essential, we write $S = S_q$. 

\begin{remark}[Choice of Test Statistic] While we focus on $S_q$, other test statistics can be used as well.  For example, when capturing deviations in the average effect $ T_{*}^{-1} \sum_{t=T_0+1}^T \theta_t$, it is useful to consider $S(\hat u) = T_\ast^{-1/2} \left |\sum_{t=T_0+1}^T \hat{u}_t \right |$. \qed
\end{remark}

We use (block) permutations to compute $p$-values. A permutation $\pi$ is a one-to-one mapping $\pi:\{1,\dots,T\}\mapsto\{1,\dots,T\}$. We denote the set of permutations under study as $\Pi$ and assume that $\Pi$ contains the identity map $\mathbb{I}$. We focus on two different sets of permutations: (i) the set of all permutations, which we call \emph{iid permutations},  $\Pi_{\text{all}}$, and (ii) the set of all (overlapping) \emph{moving block permutations}, $\Pi_{\to}$.\footnote{We can also consider other types of permutations; for example, the ``iid block'' permutations. Specifically, let $\{b_1, \dots, b_K\}$ be a partition of $\{1,\dots,T\}$, then we collect all the permutations $\pi$ of these blocks, forming the ``iid m-block'' permutations $\Pi_{mb}$. In our context, choosing $m=T_*$ is natural, though other choices should work as well, similarly to the choice of block size in the time series bootstrap. We refer to CWZ18 for more results on block permutations.} The elements of $\Pi_{\to}$ are indexed by $j \in \{0,1, \dots, T-1\}$, and the permutation $\pi_{j}$ is defined as
\[
\pi_{j}(i)=\begin{cases}
i+j & {\text{if}}\ i+j\leq T\\
i+j-T& {\text{otherwise}}.
\end{cases}
\]
Figure \ref{fig:illustration_permutation} presents a graphical illustration of $\Pi_{\text{all}}$ and $\Pi_{\to}$.

\begin{center}
[Figure  \ref{fig:illustration_permutation} around here.]
\end{center}

The choice of $\Pi$ does not matter for the exact finite sample validity of our procedures if the residuals are exchangeable. However, $\Pi_{\text{all}}$ has more elements than $\Pi_{\to}$, allowing us to compute more precise $p$-values and to test at lower significance levels. For the asymptotic validity under estimator consistency, the choice of $\Pi$ depends on the assumptions that we are willing to impose on the stochastic shock sequence $\{u_t\}$ (cf. Section \ref{subsec:small_error}).

For each $\pi \in \Pi$, let $\hat{u}_\pi=(\hat{u}_{\pi(1)},\dots,\hat{u}_{\pi(T)})'$ denote the vector of permuted residuals.\footnote{If the estimator of $P_t^N$ is invariant under permutations of the data $\{Z_t\}$ across the time series dimension (which is the case for many estimators in Section \ref{subsec:sc_panel}), permuting the residuals $\{\hat{u}_t\}$ is equivalent to permuting the data $\{Z_t\}$.} The permutation $p$-value is defined as follows.

\smallskip

\noindent \textbf{Definition of $p$-Value.} The $p$-value is 
\begin{equation}
\hat{p}=1- \hat{F}\left( S(\hat{u})\right), ~~\text{where}~ \hat{F}\left( x\right)=\frac{1}{|\Pi|}\sum_{\pi\in\Pi}\mathbf{1}\left\{ S\left(\hat{u}_\pi \right)<  x\right\}.
 \label{eq:p_value}
\end{equation}

We are often interested in testing pointwise hypotheses about $\theta_t$, $H_0:\theta_{t}=\theta_{t}^0$, and in constructing pointwise confidence intervals for $\theta_{t}$. Pointwise hypotheses can be tested by defining the data under the null as $\Zb=\left(Z_1,\dots,Z_{T_0},Z_{t}\right)'$, provided that $P_t^N$ can be estimated based on $\Zb$. Pointwise $(1-\alpha)$ confidence intervals for $\theta_{t}$ can be constructed via test inversion as described in Algorithm \ref{algo:pointwise_ci}.

\begin{algo}[Pointwise Confidence Intervals]\label{algo:pointwise_ci} (i) Choose a fine grid of $G$ candidate values $\tilde\Theta_t=\{\tilde\theta_{1t}^0,\dots,\tilde\theta_{Gt}^0\}$. (ii) For $\tilde\theta_{t}^0\in \tilde{\Theta}_t$, define $\Zb$ for the null hypothesis $H_0:\theta_{t}=\tilde\theta_{t}^0$ and compute the corresponding $p$-value, $\hat p (\tilde\theta_{t}^0)$, using \eqref{eq:p_value}. (iii) Return the $(1-\alpha)$ confidence set $\mathcal{C}_{1-\alpha}(t)=\left\{\tilde\theta_{t}^0\in \tilde\Theta_t: ~\hat p (\tilde\theta_{t}^0)> \alpha \right\}$. 
\end{algo}

\subsection{Models for Counterfactual Proxies via Synthetic Control and Panel Data}
\label{subsec:sc_panel}
The availability of control units motivates several strategies for modeling the counterfactual mean proxies $P_t^N$. We estimate $P_t^N$ based on the imputed data under the null hypothesis, $\Zb(\theta^0)$, and write $Y_{1t}^N$ instead of $Y^I_{1t}-\theta_t^0$ to alleviate the exposition. 

\subsubsection{Difference-in-Differences Methods}
\label{ex:did}
The difference-in-differences method postulates the following model for the counterfactual mean proxy \citep[e.g.,][Section 5.1]{DI16}:
$
P_t^N=\mu+\frac{1}{J}\sum_{j=2}^{J+1}Y^N_{jt}.
$
This model automatically embeds the identifying information. The counterfactual mean proxy can be estimated as 
$
\hat{P}^N_t=\frac{1}{T}\sum_{s=1}^{T}\left(Y^N_{1s}-\frac{1}{J}\sum_{j=2}^{J+1}Y^N_{js}\right)+\frac{1}{J}\sum_{j=2}^{J+1}Y^N_{jt}.
$

\subsubsection{Synthetic Control and Constrained Lasso}
\label{ex:sc}
The canonical SC method \citep[e.g.,][]{abadie2003economic,abadie10sc,abadie2015comparative} postulates the following model:
\begin{eqnarray}
P^N_t=\sum_{j=2}^{J+1}w_jY^N_{jt},  ~~\text{where}~~ w\ge 0 ~~\text{and}~~ \sum_{j=2}^{J+1} w_j=1. \label{eq:sc_model}
\end{eqnarray}
We need to impose an identification condition that allows us to identify the weights $w$, for example:\footnote{More generally, other exclusion restrictions and identifying assumptions could be used. See also \citet{abadie10sc}, \citet{ferman2019synthetic} and \citet{ferman2019properties}, who study the behavior of SC when the data are generated by a factor model.}
 
\begin{itemize}
\item[(SC)] Assume that the structural shocks $u_t$ for the treated unit are uncorrelated with 
contemporaneous values of the outcomes, namely: $E \left( u_t Y^N_{jt} \right)= 0~\text{for}~ 2\leq j\leq J+1$.
\end{itemize}

 The counterfactual is estimated as $\hat{P}^N_t=\sum_{j=2}^{J+1}\hat{w}_jY^N_{jt}$.
We focus on the following canonical SC estimator for $w$:\footnote{This formulation of canonical SC
without covariates is due to  \citet{DI16}, who refer to the estimator \eqref{eq:sc_problem} as ``constrained regression''. Note that unlike \citet{DI16}, we estimate $w$ under the null hypothesis based on all the data. We focus on the canonical problem \eqref{eq:sc_problem} for concreteness. \citet{abadie10sc,abadie2015comparative} consider a more general version that also includes covariates into the estimation of the weights. Our inference method also works in conjunction with more recently proposed modified versions of SC, such as the augmented SC estimator of \citet{benmichael2018augmented}.}
\begin{eqnarray}
\hat{w}=\arg\min_{w} \sum_{t=1}^{T}\left(Y^N_{1t}-\sum_{j=2}^{J+1}w_{j}Y^N_{jt}\right)^2~~ \text{s.t.}~~ w\ge 0 ~~\text{and}~~ \sum_{j=2}^{J+1}w_j=1\label{eq:sc_problem}.
\end{eqnarray}

As an alternative, we can consider the more flexible model\footnote{The idea to relax the non-negativity constraint on the weights is not new. It first appeared in \citet{hsiao2012panel}, who compared their factor model approach to SC, and also in \citet{valero2015synthetic}, who used the cross-validated Lasso to estimate the weights, and in \citet{DI16},  who used cross-validated Elastic Net for estimation of weights. They do not establish the formal properties of these estimators. Here we emphasize another version of relaxing SC, model \eqref{ex:cLasso}, which leads to constrained Lasso \eqref{eq:cLasso}. Constrained Lasso demonstrates an excellent theoretical and practical performance:  it is tuning-free, performs very well empirically and in simulations, and we prove that it is consistent for dependent data without any sparsity conditions on the weights and that it satisfies the estimator stability condition required for validity under misspecification. We emphasize that this estimator generally differs from the cross-validated Lasso estimator.}
\begin{equation}\label{ex:cLasso} 
P^N_t=\mu + \sum_{j=2}^{J+1}w_jY^N_{jt}, ~~ \text{where}~~  \| w\|_1 \leq 1,
\end{equation}
maintaining the same identifying assumption (SC). The counterfactual is estimated as $\hat{P}^N_t=\hat\mu+\sum_{j=2}^{J+1}\hat{w}_jY^N_{jt}$ 
by the $\ell_1$-constrained least squares estimator, or constrained Lasso \citep[e.g.,][]{raskutti2011minimax}:
\begin{eqnarray}
(\hat\mu,\hat{w})=\arg\min_{(\mu,w)} \sum_{t=1}^{T}\left(Y^N_{1t}-\mu-\sum_{j=2}^{J+1}w_{j}Y^N_{jt}\right)^2 ~~\text{s.t.}~~ \| w\|_1\le 1.\label{eq:cLasso}
\end{eqnarray}

The advantage over other penalized regression methods discussed next is that constrained Lasso is essentially tuning-free, does not rely on any sparsity conditions, and is valid for dependent data under weak assumptions. Moreover, constrained Lasso encompasses both difference-in-differences and canonical SC as special cases (by setting $w=(1/J,\dots,1/J)'$ and $\mu=0,w\ge 0$, respectively) and, thus, provides a unifying approach for the regression-based estimation of $P^N_t$. 

Section  \ref{sub: low level SC} provides primitive conditions that guarantee that the SC and the constrained Lasso estimators are valid in our framework in settings with potentially many control units (large $J$). Finally, we note that it is straightforward to incorporate (transformations of) covariates $X_{jt}$ into the estimation problems \eqref{eq:sc_problem} and \eqref{eq:cLasso}.

\subsubsection{Penalized Regression Methods}
\label{ex:pen_reg}
Consider a linear model for $P_t^N$:
$
P^N_t=\mu+\sum_{j=2}^{J+1}w_jY^N_{jt}.
$
We maintain the identifying assumption (SC). The counterfactual is estimated by $\hat{P}^N_t=\hat\mu+\sum_{j=2}^{J+1}\hat{w}_jY^N_{jt}$,
where 
\begin{eqnarray}
(\hat\mu,\hat{w})=\arg\min_{(\mu,w)} \sum_{t=1}^{T}\left(Y^N_{1t}-\mu-\sum_{j=2}^{J+1}w_{j}Y^N_{jt}\right)^2 + \mathcal{P}(w)\label{eq:restr_reg},
\end{eqnarray}
and $\mathcal{P}(w)$ is a penalty function that penalizes deviations away from zero.  If it is desired
to penalize deviations away from other focal points $w^0$, for example, $w^0 =(1/J, \dots , 1/J)$ used in the difference-in-differences approach, we may always use instead: $\mathcal{P}(w) \leftarrow \mathcal{P}(w- w^0)$. Note that it is straightforward to incorporate covariates $X_{jt}$ into the estimation problem \eqref{eq:restr_reg}.  

Different variants of $\mathcal{P}(w)$ can be considered. Examples include: Lasso \citep{tibshirani96regression}, where $\mathcal{P}(w)=\lambda \|w\|_1$ and $\lambda$ is a tuning parameter; Elastic Net \citep{Zou2005}, where $\mathcal{P}(w)=\lambda \left( (1-\alpha)\|w\|_2^2+\alpha \|w\|_1 \right)$ and $\lambda$ and $\alpha$ are tuning parameters; Lava \citep{chernozhukov2017lava}, where $\mathcal{P}(w)=\inf_{a + b = w}\lambda \left((1-\alpha) \|a\|_2^2+\alpha \|b\|_1 \right)$ and $\lambda$ and $\alpha$ are tuning parameters.

In the context of CSC methods, Lasso was used by \citet{valero2015synthetic}, \citet{li2017estimation}, and \citet{carvalho2018arco}, while \citet{DI16} proposed to use Elastic Net. We will impose only weak requirements on the performance of the estimators (pointwise consistency and consistency in prediction norm), which implies that these estimators are valid in our framework under any set of sufficient conditions that exists in the literature. 

\subsubsection{Interactive Fixed Effects, Factor, and Matrix Completion Models}
\label{ex:interactive_fe}
Consider the following interactive FE model for treated and untreated units:
\begin{eqnarray}
\begin{array}{l}
Y^N_{jt}=\lambda_{j}'F_{t}+ X_{jt}'\beta + u_{jt},\quad \text{for}\quad 1\leq j \leq J+1~\text{and}~ 1\leq t\leq T,
\end{array}\label{eq:interactive_feM}
\end{eqnarray}
where $F_t$ are unobserved factors, $\lambda_j$ are unit-specific factor loadings, and $\beta$ is a vector of common coefficients. Model \eqref{eq:interactive_feM} nests the classical factor model when $\beta=0$ and also covers the traditional linear FE model, in which $\lambda_j'F_t =  \lambda_j+F_t$. Consider the following assumption.

\begin{itemize}  
\item[(FE)]  Assume that $u_{jt}$ is uncorrelated with $(X_{jt},F_{t},\lambda_j)$,
as well as other identification conditions in \citet{bai2009panel}. 
\end{itemize}

The model leads to the following proxy:
\begin{eqnarray}
P_t^N=\lambda_1'F_t + X_{1t}'\beta.
\label{eq:interactive_fe}
\end{eqnarray}
Counterfactual proxies are estimated by $\hat{P}_t^N=\hat\lambda_1'\hat{F}_t + X_{1t}'\hat \beta$, where $\hat{\lambda}_1$ and $\hat{F}_t$, and $\hat \beta$ are obtained using the alternating least squares method  applied to the model \eqref{eq:interactive_feM}; see, for example, \citet{bai2009panel} and \citet{hansen2019factor} for a version with high-dimensional covariates.

\citet{hsiao2012panel} appears to the be first work that proposed the use of factor models for predicting the (missing) counterfactual responses specifically in SC settings. \citet{gobillon2016regional} and \citet{xu2017generalized} employ \citet{bai2009panel}'s estimator in this setting, albeit provide no formally justified inference methods. Formal inference results for interactive FE and factor models in SC designs are developed in \citet{chan2016policy} and \citet{li2018inference} among others.\footnote{Factor models are widely used in macroeconomics for causal inference and prediction; see, for example, \citet{stock2016factor} and the references therein. In microeconometrics, factor models are used for estimation of treatment/structural effects; see, for example, \citet{hansen2019factor} who use interactive FE models to estimate the effect of gun prevalence on crime.}

Other recent applications to predicting counterfactual responses include \citet{amjad2018robust} and \citet{athey2018matrix} (using, respectively, singular value thresholding and the nuclear norm penalization).\footnote{Note that  \citet{athey2018matrix}'s analysis applies to a broader collection of problems with general missing data patterns, nesting SC and difference-in-difference problems as special cases.} Our method delivers a way to perform valid inference for policy effects using any of the factor model estimators used in these proposals applied to the complete data under the null.\footnote{Note that in our case the sharp null allows us to impute the missing counterfactual response and apply any of the factor estimators to estimate the factor model
for the entire data, which is then used for conformal inference. Hence our inference approach does 
not provide inference for the counterfactual prediction methods given in those papers. Indeed, there, the missing data entries are being predicted using factor models, whereas in our case the missing data entries are known under the null, and we use any form of low-rank approximation or interactive FE model to estimate the model for the entire data under the null hypothesis.} We shall be focusing on \citet{bai2009panel}'s alternating least squares estimator\footnote{We choose to focus on PCA/SVD and the alternating least squares estimator for the following reasons: (1) they are by far the most widely used in practice, (2) the alternating least squares estimator is computationally attractive and easily accommodates unbalanced data.} and on matrix completion via nuclear norm penalization when verifying our conditions.

\subsection{Models for Counterfactual Proxies via Time Series and Fused Models}
\label{subsec:ts_fused}
\subsubsection{Simple Time Series Models} 
If no control units are available, one can use time series models for the single unit exposed to the intervention. For example, consider the following autoregressive model:\footnote{We can also add a moving average component for the errors, but we do not do so for simplicity.}
\begin{equation}
\begin{array}{l}
Y_{1t}^N - \mu =  \rho (Y_{1{(t-1)}}^N -\mu ) + u_{t}\\
Y_{1t}^I  - \mu = \rho (Y_{1{(t-1)}}^N - \mu ) +  \theta_t + u_t  
\end{array}   \Bigg | \quad E( u_t) = 0, \quad \{u_t\} ~\text{iid},  \quad t=1,\dots,T. \label{eq:model_fused}
\end{equation}
In model \eqref{eq:model_fused}, the mean unbiased proxy is given by
$
P_t^N =  \mu + \rho (Y_{1{(t-1)}}^N- \mu).
$
Note that the policy effect here is transitory, namely it does not feed-forward itself on the future values of $Y_{1t}^I$ beyond the current values.\footnote{We leave the model with persistent feed-forward effects, 
$Y_{1t}^I = \rho (Y_{1(t-1)}^I) + \theta_t + u_t $, to future work. }  Under the null hypothesis, we can impute the unobserved counterfactual  as $Y_{1t}^N = Y_{1t} - \theta_t$  and estimate the model using traditional time series methods, and we can conduct inference by permuting the residuals.

The simplest form of the autoregressive model is the AR($K$) process, where the $\rho(\cdot)$ take the form: $\rho ( \cdot )  =  \sum_{k=0}^K \rho_k \mathrm{L}^k (\cdot ),$ where $\mathrm{L}$ is the lag operator.  There are many identifying conditions
for these models, see, for example, \citet{Hamilton1994} or \citet{brockwell2013time}. More generally, we can use a nonlinear function of lag operators, $ \rho (\cdot ) = m( \cdot, \mathrm{L}^1 (\cdot), \dots, \mathrm{L^k} (\cdot ))$, as, for example, when applying neural networks to time series data \citep[e.g.,][]{chen1999improved,chen2001semiparametric}, and we refer to the latter for identifying conditions.

\subsubsection{Fused Time-Series/Panel Models}

A simple and generic way to combine the insights from the panel data and time series models is as follows. Consider the system
of equations:
\begin{equation}\label{eq: fused model AR error}
\begin{array}{l}
Y_{1t}^N = C_t^N + \varepsilon_t \\
Y_{1t}^I = C_t^N + \theta_t + \varepsilon_t \\
\end{array}  \Bigg |  
\begin{array}{l}  
\varepsilon_t = \rho (\varepsilon_{t-1}) + u_t, ~  \{u_t\} ~\text{iid},~  E( u_t) = 0,     \\ 
\{u_t\}  \textrm{ is independent of } \{C_t^N\},
 \end{array}
  \Bigg | \quad t=1,\dots,T, 
  \end{equation}
where  $C_t^N$ is a panel model proxy for $Y_{1t}^N$, identified by one of the panel data methods.  Note that the model has the autoregressive formulation:
$
Y_{1t}^N = C_t^N  + \rho (Y_{1(t-1)}^N - C_{t-1}^N) + u_t,
$
thereby generalizing the previous model. 

Here the mean unbiased proxy for $Y_{1t}^N$ is given by $P_t^N = C_t^N + \rho (\varepsilon_{t-1})$.
$P_t^N$ is a better proxy than $C_t^N$ because it provides an additional noise reduction through prediction of the stochastic shock by its lag. The model combines any favorite panel model $C_t^N$ for counterfactuals with a time series model for the stochastic shock model in a nice way:  we can identify $C_t^N$ under the null by ignoring the time series structure, and  then  identify the time series structure of  the residuals $Y_{1t}^N - C_t^N$.  Estimation can proceed analogously. This approach will often improve the size accuracy of our inferential procedures.

\section{Theory}
\label{sec:theory}

When the data are iid (or exchangeable), our procedure is exactly valid in finite samples as shown in the Appendix. In this section, we establish the validity of our inference methods with time series data. Our results are non-asymptotic in nature and, hence, hold in {\it finite samples}. Finite sample bounds are provided for the size properties of our procedure; these bounds imply that our approach is exact as $T_0\rightarrow \infty$. In Section \ref{subsec:small_error}, we establish the validity of our procedure when the estimator of $P_t^N$ satisfies weak and easy-to-verify small error conditions (pointwise consistency and consistency in the prediction norm). This result accommodates non-stationary data and only requires stationarity and weak dependence of the stochastic shock process $\{u_t\}$. In Section \ref{subsec:stability}, we consider a setting that accommodates misspecification and inconsistent estimators. We show that if the data are stationary and weakly dependent, our procedure is valid, provided that the estimators are stable.  

\subsection{Approximate Validity under Estimator Consistency}
\label{subsec:small_error}

The main condition underlying the results in this section is the following assumption on the stochastic shock process.

\begin{assumption}[Regularity of the Stochastic Shock Process]
\label{ass:u} Assume that the density function of $S(u)$ exists and is bounded, and that the stochastic process $\{u_t\}_{t=1}^T$ satisfies one of the following conditions.
\begin{enumerate}  \setlength{\itemsep}{0pt} \setlength{\parskip}{0pt}
\item $\{u_t\}_{t=1}^T$ are iid,  or  \label{ass:u_iid}
\item $\{u_t\}_{t=1}^T$ are stationary, strongly mixing, with sum of mixing coefficient bounded by $M$. \label{ass:u_weak_dependent} 
\end{enumerate}
\end{assumption}

Assumption \ref{ass:u} allows the data to be non-stationary and exhibit general dependence patterns. Assumption \ref{ass:u}.\ref{ass:u_iid} of iid shocks is our first sufficient condition. Under this condition, we will be able to use iid permutations, giving us a precise estimate of the $p$-value. The iid assumption can be replaced by Assumption \ref{ass:u}.\ref{ass:u_weak_dependent}, which holds for many commonly encountered stochastic processes such as ARMA and GARCH. It can be easily replaced by an even weaker ergodicity condition, as can be inspected in the proofs. Under this assumption, we will have to rely on the moving block permutations. 

\begin{remark}[Heteroscedasticity] Assumption \ref{ass:u} does not rule out conditional heteroscedasticity in the stochastic shock process $\{u_t\}$. Unconditional heteroscedasticity is allowed in $\{Z_t\}$ but not in $\{u_t\} $. When we suspect unconditional heteroscedasticity in $\{u_t\}$, we can apply another filter or model to obtain ``standardized residuals'' from $ \{\hat{u}_t\}$. This will generally require another layer of modeling assumptions, leading to an overall procedure that reduces the data to ``fundamental'' shocks that are assumed to be stationary under the null.  \qed
\end{remark}

We also impose the following condition on the estimation error under the null hypothesis. Let $P^N=(P^N_1,\dots,P^N_T)'$ and $\hat P^N=(\hat{P}^N_1,\dots,\hat{P}^N_T)'$.

\begin{assumption}[Consistency of the Counterfactual Estimators under the Null]\label{ass:high_level_est_error} Let there be sequences of constants $\delta_T$ and $\gamma_T$ converging to zero. Assume that with probability $1- \gamma_T$,
 \begin{enumerate} \setlength{\itemsep}{0pt} \setlength{\parskip}{0pt}
\item the mean squared estimation error is small, $\| \hat P^N - P^N \|^2_{2}/T \leq\delta^2_{T}$;
\item for $T_{0}+1\leq t\leq T$, the pointwise errors are small, $|\hat P^N_t- P^N_t| \leq\delta_{T}$.
\end{enumerate}
\end{assumption}

Assumption \ref{ass:high_level_est_error} imposes weak and easy-to-verify conditions on the performance of the estimators $\hat P_t^N$ of the counterfactual
mean proxies $P_t^N$.  These conditions are readily implied by the existing results for many estimators discussed in Section \ref{sec:conformal_inference}. In Section \ref{sec:primitive_conditions}, we provide explicit primitive conditions and references to primitive conditions implying Assumption \ref{ass:high_level_est_error}.\footnote{While our general results in this section are non-asymptotic, some of the analysis in Section \ref{sec:primitive_conditions} will not be non-asymptotic in nature.}

\begin{thm}[Approximate Validity under Consistent Estimation]  
\label{thm:approximate_validity}
Assume that $T_*$ is fixed. Suppose that Assumptions \ref{ass:dgp} and \ref{ass:high_level_est_error} hold. Impose  Assumption \ref{ass:u}.\ref{ass:u_iid} if $\Pi=\Pi_{\text{all}}$; impose Assumption \ref{ass:u}.\ref{ass:u_weak_dependent} if $\Pi=\Pi_{\to}$. Assume $S(u)$ has a density function bounded by $D$ under the null. Then, under the null hypothesis, the p-value is approximately unbiased in size:
\[
|P\left(\hat{p}\leq\alpha\right)- \alpha| \leq C ( \tilde \delta_T + \delta_T + \sqrt{\delta_T} + \gamma_T),
\]
where $\tilde \delta_T = (T_*/T_0)^{1/4}(\log T)$ and the constant $C$ depends on $T_*$, $M$ and $D$,  but not on $T$.
\end{thm}

The above bound is non-asymptotic, allowing us to claim uniform validity with respect to a rich variety of data generating processes.  Using simulations and empirical examples, we verify that our tests have good power and generate meaningful empirical results. There are other considerations  that also affect  power. For example, the better the model for $P_t^N$, the less variance the stochastic shocks will have, subject to assumed invariance to the policy. The smaller the variance of the shocks, the more powerful the testing procedure will be.

\subsection{Approximate Validity under Estimator Stability}
\label{subsec:stability}

Misspecification is an important practical concern, and consistency of the estimators of the counterfactual mean proxies $P_t^N$ may be questionable in certain settings. The classical analysis of misspecification focuses on convergence to pseudo-true values \citep[e.g.,][]{white1996estimation}. If it is possible to show that the estimator of the counterfactual mean proxy, $\hat{P}_t^N$, is consistent for some pseudo-true value $P_t^{N\ast}$ and that $\left\{Y_{1t}^N-P_t^{N\ast}\right\}_{t=1}^T$ is stationary and weakly dependent, the theoretical results in Section \ref{subsec:small_error} imply the validity of our procedure. Pseudo-true consistency can often be verified for low-dimensional models, but consistency results under misspecification remain elusive in high-dimensional settings. Therefore, we consider a notion of approximate exchangeability, which only requires the estimator to be \emph{stable} instead of consistent for a pseudo-true value. This stability condition does not require $\hat{P}_t^N$ to be consistent for anything, nor does it rely on correct specification of the counterfactual mean proxies. In the Appendix, we illustrate the difference between consistency and stability based on the analytically tractable example of Ridge regression. 

The basic idea underlying the theoretical analysis here is as follows. If the estimators are non-random or independent of the data, then stationarity and weak dependence of the data would mean that $\hat{p}$ based on moving block permutations approximately has a uniform distribution under the null. This result follows from uniform laws of large numbers for dependent data. However, in practice, the estimators are computed using the data and are thus not independent of the data. Our key insight is that stable estimators are approximately independent of individual observations. 

We now formalize the notion of stability of an estimator. To emphasize the dependence of $S(\hat{u}) $ on the estimator, with a slight abuse of notation, we write $S(\Zb,\beta)=\phi(Z_{T_{0}+1},\dots,Z_{T_{0}+T_{*}};\beta)$. 
Let $\{\tilde{Z}_{t}\}_{t=1}^{T}$ be iid from the distribution
of $Z_{1}$ and independent of $\Zb$. For any $H\subset\{1,\dots,T\}$,
let $Z_{t,H}=Z_{t}\oneb\{t\notin H\}+\tilde{Z}_{t}\oneb\{t\in H\}$,
and $\Zb_{H}=\{Z_{t,H}\}_{t=1}^{T}$. Hence, $\Zb_{H}$ is a perturbed
version of $\Zb$ under $H$, i.e., $\Zb$ with elements in $H$ replaced
by $\{\tilde{Z}_{t}\}_{t\in H}$. 

By stability, we mean that the estimator computed using $\Zb$ is similar to that computed using  $\Zb_H$ for $H\in \mathbb{H}$. Let $R\in\mathbb{N}$ and define $m=\left\lfloor T_{0}/R\right\rfloor $. The class $\mathbb{H}=\{\tH_{1},\dots,\tH_{R}\}$ contains $R$ members with $|\tH_j|\leq 3m$ elements. The plan is to require stability under $R\asymp T_0/\log(T_0)$ (so $|\tH_j|\asymp \log(T_0)$). Since $\log(T_0)\ll T_0$, swapping out $O(\log(T_0))$  out of $T_0+T_*$ data points should not cause a large change in the estimator for reasonable estimators. 

We now give precise definitions of sets in $\mathbb{H}$. For $j\in\{1,\dots,R\}$, let $H_{j}=\{(j-1)m+1,\dots,jm\}$. Since the test statistic depends on $T_*$ data points after obtaining the estimator,  defining $\mathbb{H}$ to be $\{H_1,\dots,H_R \}$ is not enough for technical arguments; we need a ``wedge'' to ensure that these $T_*$ data points do not cause a problem. To do so, we enlarge $H_j$ as follows.  Let $k\in\mathbb{N}$
satisfy $T_{*}<k<m$. We let $\tH_{j}$ denote the $k$-enlargement
of $H_{j}$, i.e., $\tH_{j}=\{s: \min_{t\in H_{j}}|s-t|\leq k\}$.
Note that $\tH_{j}=\{(j-1)m+1-k,\dots,jm+k\}$ for
$2\leq j\leq R-1$, $\tH_{1}=\{1,\dots,m+k\}$ and $\tH_{R}=\{(R-1)m+1-k,\min\{Rm+k,T\}\}$.

\begin{assumption}[Estimator Stability]
\label{assu: stability}Let $\Pi= \Pi_{\to}$. There exist non-decreasing functions $\varrho_{T}(\cdot)$
such that 
$
P\left(\max_{\pi\in\Pi}\left|S\left(\Zb^{\pi},\hat{\beta}(\Zb)\right)-S\left(\Zb^{\pi},\hat{\beta}(\Zb_{H})\right)\right|\leq\varrho_{T}(|H|)\right)\geq1-\gamma_{1,T}
$ 
and \\
$
P\left(\max_{\pi\in\Pi}\left|S\left((\dot{\Zb})^{\pi},\hat{\beta}(\Zb)\right)-S\left((\dot{\Zb})^{\pi},\hat{\beta}(\Zb_{H})\right)\right|\leq\varrho_{T}(|H|)\right)\geq1-\gamma_{1,T}
$
for any $H\in\{\tH_{1},\dots,\tH_{R}\}$, where $\dot{\Zb}\overset{d}{=}\Zb$
and $\dot{\Zb}$ is independent of $(\Zb,\{\tilde{Z}_{t}\}_{t=1}^{T})$. 
\end{assumption}

Assumption \ref{assu: stability} specifies the estimator stability condition. It strengthens the perturb-one sensitivity of \citet[][Assumption A.3]{lei2017distributionfree}. When the model is misspecified, Assumption \ref{assu: stability} holds whenever the estimator $\hat{\beta}(\Zb)$ is consistent to a pseudo-true parameter value. However, it is more general in that the estimator $\hat{\beta}(\Zb)$ need not converge to any non-random quantity as long as it is stable under perturbations in a few observations. This feature is crucial in our setting as it allows us to accommodate high-dimensional CSC methods for many of which consistency results under misspecification are not available. Primitive sufficient conditions for Assumption \ref{assu: stability} are provided in the Appendix.

Let $\Psi(x;\beta)=P(\phi(Z_{T_{0}+1},\dots,Z_{T_{0}+T_{*}};\beta)\leq x)$. Our strategy is to show that, under the null hypothesis,
$\hF(\phi(Z_{T_{0}+1},\dots,Z_{T_{0}+T_{*}};\hbeta(\Zb)))$ is approximately uniform on $(0,1)$. We exploit the stability condition in Assumption \ref{assu: stability} and show that $\hF(\phi(Z_{T_{0}+1},\dots,Z_{T_{0}+T_{*}};\hbeta(\Zb)))$ can be approximated by $\Psi\left(\phi(\bZ_{T_{0}+1},\dots,\bZ_{T_{0}+T_{*}};\hbeta(\Zb_{\tH_{R}}));\hbeta(\Zb_{\tH_{R}})\right) $, which has the uniform distribution on (0,1). Here $(\bZ_{T_{0}+1},\dots,\bZ_{T_{0}+T_{*}})$ has the same distribution as $(Z_{T_{0}+1},\dots,Z_{T_{0}+T_{*}})$ and is independent of $\Zb_{\tH_{R}}$. This essentially confirms the above intuition that for stable estimators, $\hbeta(\Zb)$ is almost independent of the last few observations $(Z_{T_{0}+1},\dots,Z_{T_{0}+T_{*}})$.

We impose the following regularity conditions on the data.

\begin{assumption}[Regularity of the Data]
\label{assu: regularity resid}The data under the null, $\{Z_{t}\}_{t=1}^{T}$, are stationary
and $\beta$-mixing with coefficient $\betamix(\cdot)$ satisfying
$\betamix(i)\leq D_{1}\exp(-D_{2}i^{D_{3}})$ for some constants $D_{1},D_{2},D_{3}>0$.
 For $1\leq j\leq R$, there exist  sequences $\xi_{T}>0$ and $\gamma_{2,T}=o(1)$ such that
$P\left(\sup_{x\in\RR}\left|\partial\Psi\left(x;\hbeta(\Zb_{\tH_{j}})\right)/\partial x\right|\leq\xi_{T}\right)\geq1-\gamma_{2,T}$.
\end{assumption}

Stationarity and $\beta$-mixing are commonly
imposed conditions on time series data.  For a large class of Markov chains,  GARCH and various
stochastic volatility models, $D_{3}=1$ \citep[cf.][]{carrasco2002mixing}.
Let $(\dot{Z}_{T_{0}+1},\ldots,\dot{Z}_{T_{0}+T_{*}})$ be an independent copy of  $ (Z_{T_{0}+1},\ldots,Z_{T_{0}+T_{*}})$ and also independent of $(\Zb,\{\tilde{Z}\}_{t=1}^{T})$.    The bounded derivative of $\Psi\left(x;\hbeta(\Zb_{\tH_{j}})\right)$ condition says that the density of $\phi(\dot{Z}_t,\ldots,\dot{Z}_{t+T_{*}-1};\hbeta(\Zb_{\tH_{j}}))$ conditional on $\hbeta(\Zb_{\tH_{j}})$ is bounded by $\xi_{T}$ with high probability. 
The bounded density condition states that the distribution of the residual does not collapse into a degenerate one or one with point mass. In many cases, $\xi_{T}=O(1)$ for continuous distributions. For example, if $(Y_t,X_t)$ is jointly Gaussian and the variance of $Y_t$ given $X_t$ is bounded below by a constant, then for any $w$ satisfying the SC restrictions, the density of $Y_t-X_t'w$ is bounded by a constant that does not depend on $w$.

The following result states the approximate validity of our testing
procedure. 
\begin{thm}[Approximate Validity under Estimator Stability] 
\label{thm: approx exchange} Let $\Pi=\Pi_{\to}$. Suppose that Assumptions \ref{assu: stability}
and \ref{assu: regularity resid} hold.  Then, under the null hypothesis, there exists a constant
$C_{1}>0$ depending only on $D_{1}$, $D_{2}$ and $D_{3}$ such
that for any $R$ with $k<\left\lfloor T_{0}/R\right\rfloor$ and $R<T_{0}/2$,
\begin{multline*}
\left|P\left(\hat{p}\leq\alpha\right)-\alpha\right|\leq C_{1}\sqrt{\xi_{T}\varrho_{T}(T_{0}/R+2k)}+C_{1}\left(T_{0}^{-1}R[\log(T_{0}/R)]^{1/D_{3}}\right)^{1/4}\\
+C_{1}\exp\left(-(k-T_{*}+1)^{1/D_{3}}\right)+C_{1}\sqrt{\gamma_{1,T}}+C_{1}\sqrt{\gamma_{2,T}}.
\end{multline*}
\end{thm}

In the theoretical arguments, we actually show a stronger result. The above bound holds for $E|P(\hat{p}\leq\alpha\mid \hbeta(\Zb_{\tH_{R}}))-\alpha|$. Since the stability condition states that $\hbeta(\Zb_{\tH_{R}})\approx\hbeta(\Zb)$, this means that $\hat{p} $ conditional on $\hbeta(\Zb)$ almost has a uniform distribution on $(0,1)$; with iid or exchangeable data, $\hat{p} $ conditional on $\hbeta(\Zb)$  has an exact uniform distribution. Therefore, we can view Theorem \ref{thm: approx exchange} as a result for approximate exchangeability. 

Due to the exponential decay of $\betamix(\cdot)$, the bound in Theorem
\ref{thm: approx exchange} tends to zero if  we  choose
$k$ to be a slowly growing sequence and $T_{0}/R$ to be of the same order. For example, we can choose $k$ and $R$ such that $k\asymp T_{0}/R \asymp \log T_0$. Since $|\tH_{j}|=\left\lfloor T_{0}/R\right\rfloor +2k $,
Assumption \ref{assu: regularity resid} only requires that the changes
to $S(\Zb^{\pi},\hat{\beta}(\Zb))$ are small if we replace only  $\log T_{0}$
observations in computing $\hat{\beta}(\Zb)$. Under finite dependence, it suffices to choose $k$ and $T_0/R$ to be large enough constants.  Note that $R$ is only needed in the theoretical arguments;  we do not need to choose $R$ when implementing the proposed procedure.

The theoretical analysis in this section suggests that allowing for both unrestricted patterns of non-stationarity and misspecification is not possible in general. To obtain valid inferences with non-stationary data, one has to either rely on correct specification and consistency or impose assumptions on the particular structure of the non-stationarity, which allow for pre-processing the data to make them stationary.

\section{Sufficient Conditions for Consistent Estimation}
\label{sec:primitive_conditions}

In this section, we revisit the representative models of counterfactual proxies introduced in Section \ref{sec:conformal_inference}. Primitive conditions are provided to guarantee that the estimation of the counterfactual mean proxies is accurate enough for the asymptotic validity of the proposed procedure. In particular, these conditions can be used to verify Assumption \ref{ass:high_level_est_error}. The regularity conditions (e.g., bounded moments, weak serial dependence) for different models are stated in the Appendix and are commonly imposed in the literature for these models.  The counterfactual mean proxies $P_t^N$ are estimated based on the imputed data under the null, $\Zb(\theta^0)$, and we write $Y_{1t}^N$ instead of $Y^I_{1t}-\theta_t^0$ to alleviate the exposition. All the results in this section assume that $T_0\rightarrow \infty$ and $J\rightarrow \infty$ (if $J$ is present in the model).

\subsection{Difference-in-Differences} \label{sub: low level DiD}
In Section \ref{ex:did}, we have seen that the counterfactual mean proxies implied by the canonical difference-in-differences model are:
$
P_t^N=\mu+J^{-1}\sum_{j=2}^{J+1}Y^N_{jt}.
$
We consider the following estimator:
$
\hat{P}^N_t=\hat{\mu}+\frac{1}{J}\sum_{j=2}^{J+1}Y_{jt},$ where $ \hat{\mu}=\frac{1}{T}\sum_{t=1}^{T}\left(Y^N_{1t}-\frac{1}{J}\sum_{j=2}^{J+1}Y^N_{jt}\right)=\mu+ \frac{1}{T}\sum_{t=1}^{T} u_t.
$
Since $\hat{P}_t^N-P_t^N=\hat \mu -\mu$, Assumption \ref{ass:high_level_est_error} holds for the simple difference-in-differences model provided that $ T^{-1}\sum_{t=1}^{T} u_t = o_P(1)$, which is true under very weak conditions.

\subsection{Synthetic Control and Constrained Lasso} \label{sub: low level SC}

Several models in Section \ref{sec:conformal_inference} (including SC and constrained Lasso) imply a structure in which  the counterfactual proxy is a linear function of observed outcomes of untreated units. 

To provide a unified framework for these models, we use $Y$ to denote a generic vector of outcomes and $X$ to denote the design matrix throughout this section. For example, in Section \ref{sec:conformal_inference}, we set $Y=Y^N_1$ and  $X=(Y^N_2,\ldots,Y^N_{J+1}) $, where $Y^N_j=(Y^N_{j1},\ldots,Y^N_{jT})' \in \mathbb{R}^T$ for $1\leq j\leq J+1 $. These models can be written as 
\begin{equation}\label{eq: linear model sec 5}
Y=Xw+u,
\end{equation}
where $u=(u_1,\dots,u_{T})'  \in \mathbb{R}^T$. Identification is achieved by requiring that $X$ and $u$ be uncorrelated (cf.\ Condition (SC)). 

Under the framework in \eqref{eq: linear model sec 5}, different models correspond to different specifications for the weight vector $w$. For the SC model in Section \ref{ex:sc}, $w $ is an unknown vector whose elements are nonnegative and sum up to one. More generally, one can simply restrict $w$ to be any vector with bounded $\ell_1 $-norm. This is the constrained Lasso estimator.

Since $P_t^N $ is the $t$-th element of the vector $Xw$, the natural estimator is $\hat{P}_t^N $ being the $t$-th element of $X\hat{w} $, where $\hat{w}$ is an estimator for $w$. The estimation of $w$ depends on the specification. Let $\mathcal{W} $ be the parameter space for $w$.  We consider the following version of the original SC estimator
\begin{equation}
\hat{w} =  \arg\min_w   \ \|Y-Xw\|_{2}\quad  \text{ s.t. }  w  \in \mathcal{W} = \{ v \geq0,  \| v\|_1 =1\} \label{eq:sc_estimator_compact}.
\end{equation}
The constrained Lasso estimator is 
\begin{equation}
\hat{w}= \arg\min_w \|Y-Xw\|_{2}\quad   \text{ s.t. }  w  \in \mathcal{W} =  \{ v:   \|v\|_{1}\leq K \}, \label{eq:classo_estimator_compact}
\end{equation}
where $K$ is bounded and $K>0$. In light of the estimator \eqref{eq:sc_estimator_compact}, a natural choice is $K=1$. 

In general, we choose the parameter space $\mathcal{W} $ to be an arbitrary subset of an $\ell_1 $-ball with bounded radius. The following result gives very mild conditions under which the constrained least squares estimators are consistent and satisfy Assumption \ref{ass:high_level_est_error}.\footnote{To simplify the exposition, we do not include an  intercept in Lemma \ref{lem: constrained LS}. Similar arguments could be used to prove an analogous result with an unconstrained intercept.} 

\begin{lem}[Constrained Least Squares Estimators]
\label{lem: constrained LS} Consider 
$\hat{w}=\arg\min_{v}  \ \|Y-Xv\|_{2}$ s.t. $v\in\mathcal{W}$,
where $\mathcal{W}$ is a subset of $\{v:\|v\|_{1}\leq K\}$
and $K$ is bounded.  Assume $w \in \mathcal{W}$, the data are $\beta$-mixing
with exponential speed, and other assumptions listed
at the beginning of the proof, including the identification condition (SC), then the estimator enjoys
the performance bounds  stated in the proof, in particular:
$
\frac{1}{T}\sum_{t=1}^{T} (\hat{P}_{t}^{N}-P_{t}^{N} )^{2}=o_{P}(1)$ and $
\hat{P}_{t}^{N}-P_{t}^{N}=o_{P}(1)$, for any $ T_{0}+1\leq t\leq T. 
$
\end{lem}

Lemma \ref{lem: constrained LS} provides several  features that are important for counterfactual inference in our setup. First, we allow $J$ to be large relative to $T$. To be precise, we only require $\log J=o(T^c) $, where $c>0$ is a constant depending only on the $\beta$-mixing coefficients; see the Appendix for details. This is particularly relevant for settings in which the number of (potential) control units and the number of time periods have a similar order of magnitude as in our empirical application in Section \ref{sec:application}. Second, Lemma \ref{lem: constrained LS} does not rely on any sparsity assumptions on $w$, allowing for dense vectors. Third, compared to typical high-dimensional estimators (e.g., Lasso or Dantzig selector), our estimator does rely on tuning parameters that can be difficult to choose in times series settings. Finally, Lemma \ref{lem: constrained LS} provides new theoretical consistency results for the canonical SC estimator in settings with time series data and potentially very many control units.

\subsection{Models with Factor Structures} \label{sub: low level factor models}
The models for counterfactual proxies introduced in Section \ref{ex:interactive_fe} have factor structures. We provide estimation results for pure factor models (without regressors), factor models with regressors (interactive FE models), and matrix completion models. In this subsection, following standard notation, we let  $N=J+1$.
\subsubsection{Pure Factor Models}

Recall from Section \ref{ex:interactive_fe} the standard factor model 
$
Y_{jt}^{N}=\lambda_{j}'F_{t}+u_{jt},
$
where $F=(F_1,\ldots,F_{T})'\in \mathbb{R}^{T\times k} $ and $\Lambda=(\lambda_{1},\ldots,\lambda_{N})'\in \mathbb{R}^{N\times k}$ represent the $k$-dimensional unobserved factors and their loadings, respectively. The counterfactual proxy for $Y_{1t}^N $ is $P_t^N=\lambda_{1}'F_{t} $. We identify $P_t^N $ by imposing the condition that the idiosyncratic terms  and the factor structure are uncorrelated (cf.\ Condition (FE)). 

We use the standard principal component analysis (PCA) for estimating $P_t^N $.\footnote{Note that PCA amounts to singular value decomposition, which can be computed using polynomial time algorithms, \citep[e.g., ][Lecture 31]{trefethen1997numerical}. } Let  $Y^{N}\in\mathbb{R}^{T\times N}$
be the matrix whose $(t,j)$ entry is $Y_{jt}^{N}$. We compute $\hat{F}=(\hat{F}_1,\ldots,\hat{F}_T)'\in\mathbb{R}^{T\times k}$
to be  the matrix containing the eigenvectors corresponding to the
largest $k$ eigenvalues of $Y^{N}(Y^{N})'$ with $\hat{F}'\hat{F}/T=I_{k}$.
Let $\hat{\lambda}_{j}'$
denote the $j$-th row of  $\hat{\Lambda}=(Y^{N})'\hat{F}/T$. Let $\hat{F}_{t}'$ denote the $t$-th row of $\hat{F}$.  Our estimate for $P_{t}^N $ is  $\hat{P}_t^N= \hat{\lambda}_{1}'\hat{F}_t $. The following lemma guarantees the validity of this estimator in our context under mild regularity conditions.

\begin{lem}[Pure Factor Model]\label{thm: low level pure factor} 
 Assume standard regularity conditions
given in \citet{bai2003inferential}, including the identification condition (FE). Consider the factor model and the principal component estimator. Then, for any $1\leq t\leq T$, as $N \to \infty$ and $T \to \infty$, we have $
\hat{P}^N_{t}-P^N_{t}=O_{P}(1/\sqrt{N} + 1/\sqrt{T} )
$
and $\frac{1}{T} \sum_{t=1}^T (\hat{P}^N_{t}-P^N_{t})^2=O_{P}(1/ N + 1/T)$.
\end{lem}

The only requirement on the sample size is that both $N$ and $T$ need to be large. Similar to Theorem 3 of  \citet{bai2003inferential}, we do not restrict the relationship between $N$ and $T$. This is flexible enough for a wide range of applications in practice as   the number of units is allowed to be much larger than, much smaller than, or similar to the number of time periods.

\subsubsection{Factor plus Regression Model: Interactive FE Model}

Now we study the general form of panel models with interactive FEs. Following Section \ref{ex:interactive_fe}, these models take the form 
$
Y_{jt}^{N}=\lambda_{j}'F_{t}+X_{jt}'\beta+u_{jt},
$
where $X_{jt}\in\mathbb{R}^{k_{x}}$ are observed covariates and $F=(F_1,\ldots,F_{T})'\in \mathbb{R}^{T\times k} $ and $\Lambda=(\lambda_{1},\ldots,\lambda_{N})'\in \mathbb{R}^{N\times k}$ represent the $k$-dimensional unobserved factors and their loadings, respectively. The counterfactual proxy for $Y_{1t}^N $ is $P_t^N=\lambda_{1}'F_{t}+X_{1t}'\beta $. In this model, we identify the counterfactual proxy through the condition that the idiosyncratic terms are  independent of the factor structure and the observed covariates (cf.\ Condition (FE)). 

 The two most popular estimators are the common correlated effects (CCE) estimator by \citet{pesaran2006estimation} and the iterative least squares estimator by \citet{bai2009panel}. We focus on the iterative least squares approach, but analogous results can be established for CCE estimators. 
 The notations for $F_t $,  $\lambda_{j} $, $\hat{F}_t $ and $\hat{\lambda}_{j} $ are the same as before. We compute 
 $$
(\hat{F},\hat{\Lambda},\hat{\beta})=\underset{F,\Lambda,\beta}{\arg \min } \sum_{t=1}^T \sum_{j=1}^{N} (Y^N_{jt}-X_{jt}'\beta-F_t'\lambda_{j} )^2\quad  \text{ s.t. }\quad  F'F/T=I_k \quad \Lambda'\Lambda = \text{Diagonal}_k.
 $$

The estimate for $P_t^N $ is $\hat{P}_t^N=\hat{\lambda}_{1}'\hat{F}_t+X_{1t}'\hat{\beta} $. The following result states the validity of applying this estimator in conjunction with our inference method.

\begin{lem}[Interactive FE Model]
\label{thm: low level interactive FE} Assume the standard conditions
in \citet{bai2009panel}, including the identification condition (FE). 
Then, for any $1\leq t\leq T$, 
$
\hat P^N_{t}-P^N_{t}=O_{P}(1/\sqrt{T} + 1/\sqrt{N})$ and $ \frac{1}{T}\sum_{t=1}^{T}(\hat{P}^N_{t}-P^N_{t})^{2}=O_{P}(1/T+ 1/N).
$
\end{lem}

Under the conditions in Theorem 3 of \citet{bai2009panel},  $N$ is of the same order as $T$ so that rate is really $T^{-1/2}$; 
however, the stated bound should hold more generally.

\subsubsection{Matrix Completion via Nuclear Norm Regularization}

Suppose that 
\begin{equation}
Y_{jt}^{N}=M_{jt}+u_{jt},\quad{\rm for}\ 1\leq j\leq J+1\ {\rm and}\ 1\leq t\leq T, \label{eq: matrix completion model}
\end{equation}
where $M_{jt}$ is the $(j,t)$-element of an unknown matrix $M\in\mathbb{R}^{(J+1)\times T}$ satisfying $\|M\|_{*}\leq K$, where $\|\cdot\|_{*}$ denotes the nuclear norm (the sum of singular values). We observe $Y_{jt}^N$ for $(j,t)\in \{1,\dots,T\}\times \{1,\dots,J+1\} \backslash \{(1,t):T_0+1\leq t\leq T\} $. The identifying condition is that $E(u\mid M)=0$ and that conditional on $M$, $\{u_{j}\}_{j=1}^{J+1}$ is independent across $j$, where $u_{j}=(u_{j1},\dots,u_{jT})'\in\mathbb{R}^{T}$.  The counterfactual proxy is $P_{t}^N=M_{1t} $ for $1\leq t\leq T$. 

The main challenge is to recover the entire matrix $M$ despite the missing entries $\{Y_{1t}^N:T_0+1\leq t\leq T\} $. The literature on matrix completion considers the model  \eqref{eq: matrix completion model} under the assumption of missingness at random and exploits the assumption that the rank of $M$ is low.\footnote{See, for example, \citet{candes2009exact}, \citet{recht2010guaranteed}, \citet{candes2011tight}, \citet{koltchinskii2011nuclear}, \citet{negahban2011estimation}, \citet{rohde2011estimation}, and \citet{chatterjee2015matrix}.} Recently, \citet{athey2018matrix} introduce this method to study treatment effects in panel data models and point out  the unobserved counterfactuals correspond to  entries that are missing in a very special pattern, rather than at random. Assuming the usual low rank condition on $M$, they employ the nuclear norm penalized estimator and provide bounds on the estimation error in the typical setup of causal panel data models. 

We take a different approach here since our main goal is hypothesis testing instead of estimation. The key observation is that under the null hypothesis, there are no missing entries in the data. By imposing the null hypothesis, we  replace the missing entries with the hypothesized values and obtain a dataset that contains $\{Y_{jt}^N:1\leq j \leq J+1,\ 1\leq t \leq T \}$. The estimator for $M$ we examine here is closely related to existing nuclear norm regularized estimators and is defined as
\begin{align}
\hat{M}= & \underset{A\in\mathbb{R}^{N\times T}}{\arg\min}\sum_{t=1}^{T}\sum_{j=1}^{N}(Y_{jt}^{N}-A_{jt})^{2}\quad  {\rm s.t. }  \ \ \|A\|_{*}\leq K, \label{eq: constrained nuclear est}
\end{align}
where $K>0$ is the bound on the nuclear norm of the true matrix. In principle, it can be a sequence that tends to infinity. When $M$ represents a factor structure with strong factors, $K$ can be shown to grow at the rate $\sqrt{NT} $. Clear guidance on how to choose $K$ is still unavailable, but following \cite{athey2018matrix} one can use  cross-validation.\footnote{The properties of cross-validation remain unknown in these settings.}  Alternatively one can use a pilot thresholded SVD estimator to get a sense of what $K$ is, and use a somewhat larger value of $K$. The following result guarantees the validity of this estimator in our context under mild regularity conditions.

\begin{lem}
\label{lem: constrained nuclear}Consider  the estimator $\hat{M}$
defined in (\ref{eq: constrained nuclear est}). Assume that $\|M\|_* \leq K $. Let the conditions listed at the beginning of the proof hold. Then, for any $T_{0}+1\leq t\leq T$,
$
\hat{P}^N_{t}-P^N_{t}=o_{P}(1)$ and $ \frac{1}{T}\sum_{t=K+1}^{T}\left(\hat{P}^N_{t}-P^N_{t}\right)^{2}=o_{P}(1).
$
\end{lem}

The result is notable because no sub-Gaussian assumptions are required.  The estimator in (\ref{eq: constrained nuclear est}) does not explicitly require a low-rank condition on $M$. Instead, we impose a growth restriction on $K$. When $M$ is generated by a strong factor structure and the null hypothesis contains full information on the missing entries, we can choose $K\asymp \sqrt{NT}$ and our consistency result holds as long as $N,T\rightarrow \infty $ and   $E(|u_{jt}|^{2+c}\mid M) $ is uniformly bounded for some $c>0$. In the case of weak factors, we can choose  $K \ll \sqrt{NT} $ and obtain consistency.

\subsection{Time Series and Fused Models}
As pointed out in Section \ref{subsec:ts_fused}, time series models can be used to model counterfactual proxies with or without control units. We now discuss  low-level conditions under which fitting these models yields estimates good enough for the purpose of our conformal inference approach.

\subsubsection{Autoregressive  Models}
The linear autoregressive model with $K$ lags can be written as
$
Y_{1t}^{N}=\rho_{0}+\sum_{j=1}^{K}\rho_{j}Y_{1t-j}^{N}+u_{t},
$
where $\{u_{t}\}_{t=1}^{T}$ is an iid sequence with $E(u_{t})=0$.\footnote{Here
the model seems different, but  Section \ref{subsec:ts_fused}'s model implies this one with  $\rho_0  = \mu (1 - \sum_{j=1}^{K}\rho_{j})$.} The counterfactual proxy for $Y_{1t}^N $ is $P_t^N= \rho_{0}+\sum_{j=1}^{K}\rho_{j}Y_{1t-j}^{N}$. We write $P_t^N$ as  $P_t^N = y_t'\rho$, where $y_{t}=(1,Y_{1t-1}^{N},Y_{1t-2}^{N},\dots,Y_{1t-K}^{N})'\in\mathbb{R}^{K+1}$ 
and  $\rho=(\rho_0,\dots,\rho_K)'\in \mathbb{R}^{K+1}$. The coefficient vector $\rho$ can be estimated using least squares: $\hat{\rho}=\left(\sum_{t=K+1}^{T}y_{t}y_{t}'\right)^{-1}\left(\sum_{t=K+1}^{T}y_{t}Y_{1t}^{N}\right)$. The estimator for $P_t^N $ is $\hat P_t^N = y_t '\hat \rho$. 

\begin{lem}[Linear AR Model]
\label{lem: low level AR}Suppose that $\{u_{t}\}_{t=1}^{T}$ is an iid sequence with $E(u_{1})=0$
and $E(u_{1}^{4})$ uniformly bounded and the roots of   $1-\sum_{j=1}^{K}\rho_{j}L^{j}=0$
are uniformly bounded away from the unit circle. Then, for any $T_{0}+1\leq t\leq T$,
$
\hat{P}^N_{t}-P^N_{t}=o_{P}(1)$ and $\frac{1}{T}\sum_{t=K+1}^{T}(\hat{P}^N_{t}-P^N_{t})^{2}=o_{P}(1).
$
\end{lem}

As mentioned in Section \ref{subsec:ts_fused}, we can also apply  nonlinear autoregressive models
$
Y_{1t}^N=\rho(Y_{1t-1}^N,Y_{1t-2}^N,\ldots,Y_{1t-K}^N)+u_{t},
$
where $\rho $ is a nonlinear function, in which case the counterfactual proxy is $P_t^N= \rho(Y_{1t-1}^N,Y_{1t-2}^N,\ldots,Y_{1t-K}^N)$.

Let $\hat{\rho} $ be an estimator for $\rho$ and $\hat{P}_{t}^N= \hat{\rho}(Y_{1t-1}^N,Y_{1t-2}^N,\ldots,Y_{1t-K}^N) $. This estimator can be parametric, semiparametric, or fully nonparametric and is only required to be consistent.

\begin{lem}[Nonlinear AR Model]
\label{lem: low level nonliear AR}  Suppose  that (1) $\|\hat{\rho}-\rho\|= O_P(r_T)$ with $r_T =o(1) $ for some appropriate norm $\|\cdot\| $ and  $\max_{K+1 \leq t \leq T }   |\hat{\rho}(Y_{1t-1}^N,Y_{1t-2}^N,\ldots,Y_{1t-K}^N) -\rho(Y_{1t-1}^N,Y_{1t-2}^N,\ldots,Y_{1t-K}^N)| \leq \ell_T \| \hat \rho - \rho \|$  for some $\ell_T r_T =o(1) $. Then, for any $T_{0}+1\leq t\leq T$,
$
\hat{P}^N_{t}-P^N_{t}=o_{P}(1)$ and $\frac{1}{T} \sum_{t=K+1}^{T}(\hat{P}^N_{t}-P^N_{t})^{2} =o_{P}(1).
$
\end{lem}

The primitive regularity conditions and the definitions of  the neural network estimators possessing these properties can be found, for example, in \citet{chen1999improved} and \cite{chen2001semiparametric}.

\subsubsection{Fused Panel/Time Series Models with AR Errors}

Here we provide generic conditions for the fused panel/time series models described in Section \ref{subsec:ts_fused}.  
In particular, AR models can be used to filter the estimated residuals and obtain
near iid errors. In Equation (\ref{eq: fused model AR error}) of Section \ref{subsec:ts_fused}, we introduce an autoregressive structure in the error terms: 
$
Y_{1t}^N=C_t^N+\varepsilon_t$ and $\varepsilon_t=\rho(\varepsilon_{t-1})+u_t,
$
where $C_t^N $ can be specified as a panel data model discussed before. Due to the autoregressive structure in $\varepsilon_{t}$, the counterfactual proxy is $P_t^N=C_t^N+\rho(\varepsilon_{t-1}) $. 

We estimate $P_t^N $ via a two-stage procedure. In the first stage, we estimate $C_t^N $ using the techniques we considered before and obtain say $\hat{C}_t^N $. In the second stage, we estimate $\rho(\varepsilon_{t-1}) $ by fitting an autoregressive model to the estimated residuals $\{\hat{\varepsilon}_t\}_{t=1}^T $, where $\hat{\varepsilon}_t=Y_{1t}^N-\hat{C}_t^N $. For simplicity, we consider a linear model in the second stage estimation. Analogous results can be obtained for more general models. To be specific, assume that  $\varepsilon_{t}=x_{t}'\rho+u_{t}$, where $x_{t}=(\varepsilon_{t-1},\varepsilon_{t-2},\dots,\varepsilon_{t-K})'\in\mathbb{R}^{K}$
and $\rho=(\rho_{1},\rho_{2},\dots,\rho_{K})'\in\mathbb{R}^{K}$. 

 Given $\{\hat{\varepsilon}_t \}_{t=1}^T $ from the first-stage estimation, we define $\hat{x}_{t}=(\hat{\varepsilon}_{t-1},\hat{\varepsilon}_{t-2},\dots ,\hat{\varepsilon}_{t-K})'\in\mathbb{R}^{K}$
and $\hat{\rho}=\left(\sum_{t=K+1}^{T}\hat{x}_{t}\hat{x}_{t}'\right)^{-1}\left(\sum_{t=K+1}^{T}\hat{x}_{t}\hat{\varepsilon}_{t}\right)$. To compute the $p$-value, we use  $\{\hat{u}_t\}_{t=K+1}^T $ with  $\hat{u}_t=\hat{\varepsilon}_t-\hat{x}_t'\hat{\rho} $ in the permutation. By the following result, this procedure is valid under very mild conditions for the first-stage estimation.

\begin{lem}[AR Errors]
\label{lem: pre-whitening u hat}Suppose that $\{u_{t}\}_{t=1}^{T}$
is an iid sequence with $E(u_{t})=0$ and $E(u_{1}^{4})$ uniformly
bounded and the roots of $1-\sum_{j=1}^{K}\rho_{j}L^{j}=0$ are uniformly
bounded away from the unit circle. We assume that 
(1) $\sum_{t=1}^{T}(\hat{C}^N_{t}-C^N_{t})^{2}=o_{P}(T)$, 
and (2) $\hat{C}^N_{t}-C^N_{t}=o_{P}(1)$ for $T_{0}-K+1\leq t\leq T$. Then, for any $T_{0}+1\leq t\leq T$,
$
\hat{P}^N_{t}-P^N_{t}=o_{P}(1)$ and $\sum_{t=K+1}^{T}\left(\hat{P}^N_{t}-P^N_{t}\right)^{2}=o_{P}(T)
$
\end{lem}

Note that the conditions in Lemma \ref{lem: pre-whitening u hat} for the autoregressive part are the same as in Lemma \ref{lem: low level AR}. Consistency of $\hat{C}_t^N $ can be verified using existing results, for example, those in Sections \ref{sub: low level DiD}--\ref{sub: low level factor models}.

\section{Empirical Application}
\label{sec:application}

We revisit the analysis in \citet{cunningham2018decriminalizing} who study the impact of decriminalizing indoor prostitution. 
They consider the case of Rhode Island, where a judge unanticipatedly decriminalized indoor sex work in July 2003 such that, until the recriminalization in November 2009, Rhode Island had decriminalized indoor and prohibited street prostitution. 

We focus on the effect of legalizing indoor prostitution on female gonorrhea incidence. Our outcome of interest is log female gonorrhea incidence per 100,000. We use the data on gonorrhea cases from the Center for Disease Control (CDC)'s Gonorrhea Surveillance Program previously analyzed by \citet{cunningham2018decriminalizing}; see their Section 3 for a detailed description and descriptive statistics. The female gonorrhea series date back to 1985 such that $T_0=19$ and $T_\ast=6$. Figure \ref{fig:raw_data} displays the raw data for Rhode Island and the rest of the U.S. states.

\begin{center}
[Figure \ref{fig:raw_data} around here.]
\end{center}

 We apply three different CSC methods: difference-in-differences, canonical SC, and constrained Lasso with $K=1$. Recall that constrained Lasso  nests both difference-in-differences and SC. Following \citet{cunningham2018decriminalizing}, the set of potential control units includes all other U.S. states and the District of Columbia ($J=50$). We choose $S_1$ as our test statistic and report $p$-values computed based on moving block and iid permutations.\footnote{To keep computation tractable, we randomly sample 10,000 iid permutations with replacement.} All computations were performed in \texttt{R} \citep{R2020}.

Before turning to the main results, we use the placebo tests proposed in the Appendix to assess the plausibility of the underlying assumptions. Specifically, based on the pre-treatment data, we test $H_0: \theta_{2003-\tau+1}=\dots=\theta_{2003}=0$ for $\tau\in \{1,2,3\}$. Rejections of this null undermine the credibility of the assumptions underlying our procedure and the inferences on policy effects in the post-treatment period. Table \ref{tab:placebo_specification_tests} presents the results. Figure \ref{fig:placebo_graphical} complements the formal tests with plots of the residuals from fitting the three models to the pre-treatment data. The placebo tests and the residual plots provide evidence in favor of the credibility of our inference method in conjunction with SC and, especially, constrained Lasso, but suggest that the difference-in-differences results need to be interpreted with caution.

\begin{center}
[Table \ref{tab:placebo_specification_tests} around here.]
\end{center}

\begin{center}
[Figure \ref{fig:placebo_graphical} around here.]
\end{center}

Table \ref{tab:zero_effect} reports $p$-values from testing the null hypothesis of a zero effect:
\begin{equation}
H_0: \theta_{2004}=\theta_{2005}=\dots =\theta_{2009}=0 \label{eq:zero_effect}.
\end{equation}
The null hypothesis \eqref{eq:zero_effect} is rejected at the 10\% level based on both permutation schemes and all three methods. 

\begin{center}
[Table \ref{tab:zero_effect} around here.]
\end{center}

Figure \ref{fig:ci_gonorrhoea} displays pointwise 90\% confidence intervals. The results are similar for all three methods. While the effect was not or only marginally significant during the first three years, legalizing indoor prostitution significantly decreased the incidence of female gonorrhea thereafter, corroborating the findings by \citet{cunningham2018decriminalizing}.

\begin{center}
[Figure \ref{fig:ci_gonorrhoea} around here.]
\end{center}

To investigate the robustness of our results, we perform a leave-one-out robustness check \citep[e.g.,][]{abadie2015comparative} to assess whether our findings are driven by a single control state. We iteratively exclude from the control group one of the states for which either the SC or constrained Lasso weights estimated based on the pre-treatment data are non-zero and compute the $p$-values for testing hypothesis \eqref{eq:zero_effect}. Figure \ref{fig:robustness} displays the distribution of the resulting $p$-values. Overall, our results are robust and not driven by a single control state: except for one specification, all results are significant at the 10\%-level.

\begin{center}
[Figure \ref{fig:robustness} around here.]
\end{center}

\if1\blind
{
\section*{Acknowledgements}
We are grateful to Guido Imbens, Jacopo Diquigiovanni, Bruno Ferman,  the Co-Editor (Matias Cattaneo), anonymous referees, and many seminar and conference participants for valuable comments. We would like to thank Scott Cunningham and Manisha Shah for sharing the data for the empirical application. W\"uthrich is also affiliated with CESifo and the Ifo Institute. Victor Chernozhukov gratefully acknowledges funding by the National Science Foundation. All errors are our own.
 
} \fi

\if0\blind
{

} \fi

\spacingset{1.0}
\setlength{\bibsep}{2pt}
\bibliographystyle{apalike}
\bibliography{SC_biblio}

\section*{Figures Main Text}

\begin{figure}[H]
\caption{Small Sample Size Properties (Nominal Level: 10\%)}
\begin{center}
\includegraphics[width=0.6\textwidth,trim={0 1cm 0 2.5cm}]{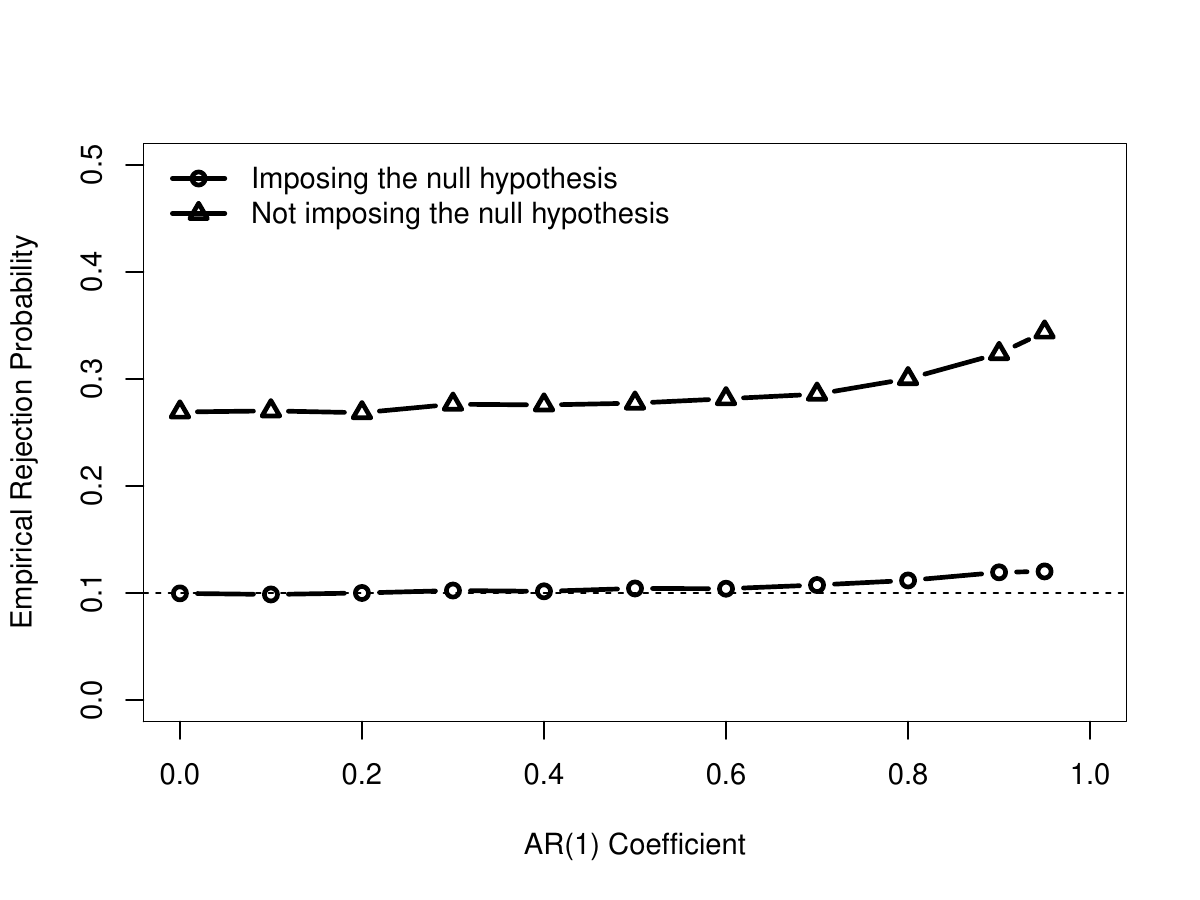}
\end{center}
    {\footnotesize \textit{Notes:} Empirical rejection probability from testing $H_0:\theta_{T_0+1}=0$. The data are generated as $Y^N_{1t}=\sum_{j=2}^{J+1}w_jY_{jt}^N+u_t$, where $Y_{jt}^N\sim N(0,1)$ is iid across $(j,t)$, $\{u_t\}$ is a Gaussian AR(1) process, $(w_2,\dots,w_{J+1})'=(1/3,1/3,1/3,0,\dots,0)'$, $T_0=19$, and $J=50$. The weights are estimated using the canonical SC method (cf. Section \ref{ex:sc}).}

\label{fig:importance_H0}
\end{figure}

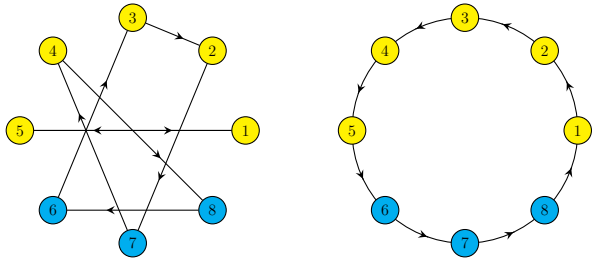
\begin{figure}[H]
    \caption{ Graphical Illustration Permutations}
    \label{fig:illustration_permutation}
\begin{center}
   \begin{tikzpicture}[scale=0.65, transform shape]
      \foreach \pt/\r/\ang in {1/3/0,2/3/45,3/3/90,4/3/135,5/3/180} {
         \node[circle,draw,fill=white] (\pt) at (\ang:\r){\pt};
      }
       \foreach \pt/\r/\ang in {6/3/225,7/3/270, 8/3/315} {
         \node[circle,draw,fill=gray] (\pt) at (\ang:\r){\pt};
      }
%     \foreach \x/\y in {1/5, 2/7, 3/2, 4/8, 5/1, 6/3, 7/4, 8/6} {
%         \draw[->-] (\x) -- (\y);
%      }
          \foreach \x/\y in {1/3, 2/6, 3/4, 4/7, 5/8, 6/5, 7/2, 8/1} {
         \draw[->-] (\x) -- (\y);
      } 
      
    \end{tikzpicture}    \hspace{.3in}  \begin{tikzpicture}[scale=0.65, transform shape]
      \draw (0,0) circle [radius=3];
       \foreach \pt/\r/\ang in {1/3/0,2/3/45,3/3/90,4/3/135,5/3/180} {
         \node[circle,draw,fill=white] (\pt) at (\ang:\r){\pt};
      }
       \foreach \pt/\r/\ang in {6/3/225,7/3/270, 8/3/315} {
         \node[circle,draw,fill=gray] (\pt) at (\ang:\r){\pt};
      }
      \foreach \ang in {25,70,115,160, 205, 250, 295, 340} {% draw remaining edges
      %     \draw[->-] (\ang-1:1) -- (\ang+1:1);
           \draw[->-] (\ang-1:3) -- (\ang+1:3);
      }
    \end{tikzpicture}
    \hspace{.3in}
    \end{center}
    {\footnotesize \textit{Notes:} The left figure gives an example of an iid permutation of $\{1,2,3,4,5,6,7,8\}$. The right figure gives an example of a moving block permutation of $\{1,2,3,4,5,6,7,8\}$. $T_0=5$, $T_\ast=3$. Pre-treatment periods are white; post-treatment periods are gray.} 
  \end{figure}

\begin{figure}[H]
\caption{Raw Data}
\begin{center}
\includegraphics[width=0.65\textwidth,trim={0 1cm 0 1cm}]{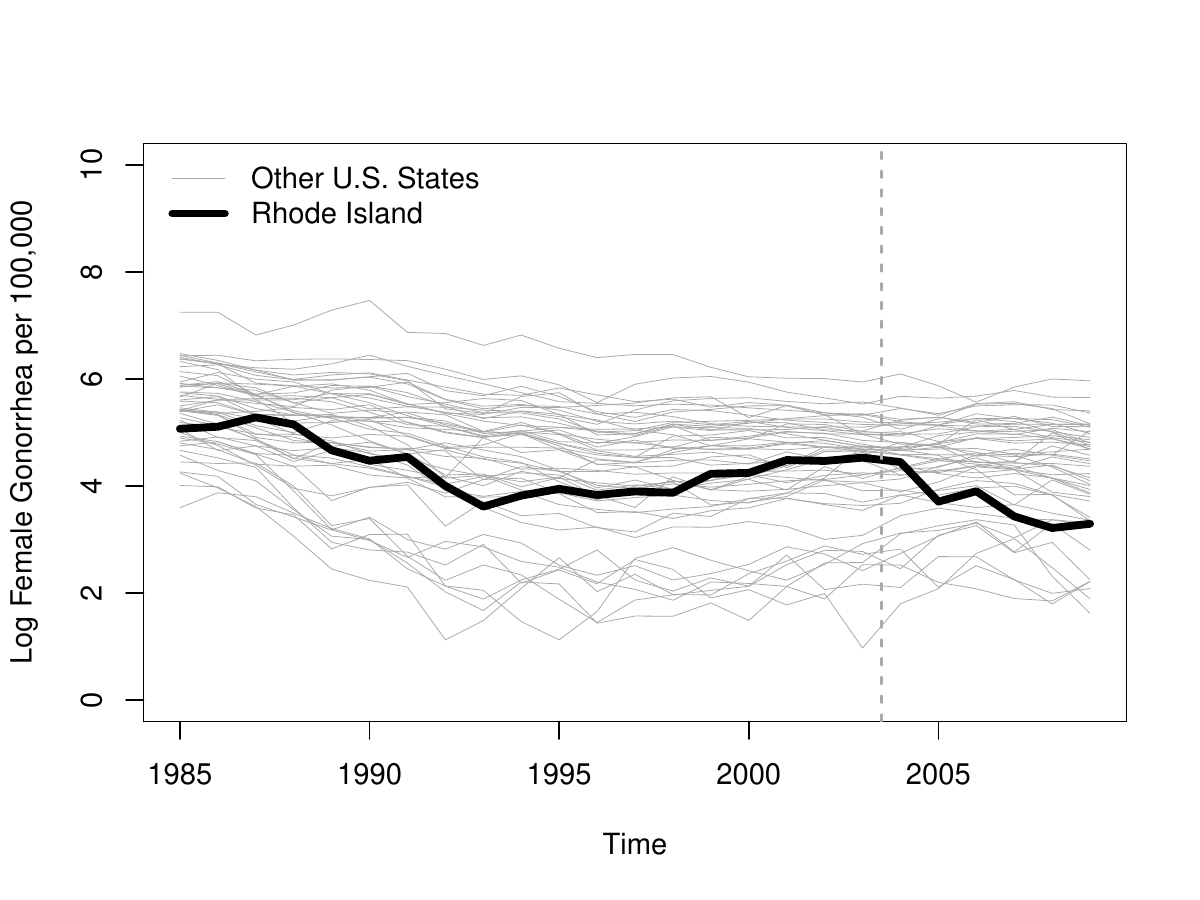}
\end{center}
\label{fig:raw_data}
    {\footnotesize \textit{Notes:} Data are from \citet{cunningham2018decriminalizing}. The figure shows the raw state-level data on log female gonorrhea cases per 100,000.} 

\end{figure}

\begin{figure}[H]
\caption{Graphical Placebo Checks}

\begin{center}

\includegraphics[width=0.325\textwidth,trim={0 1cm 0 1cm}]{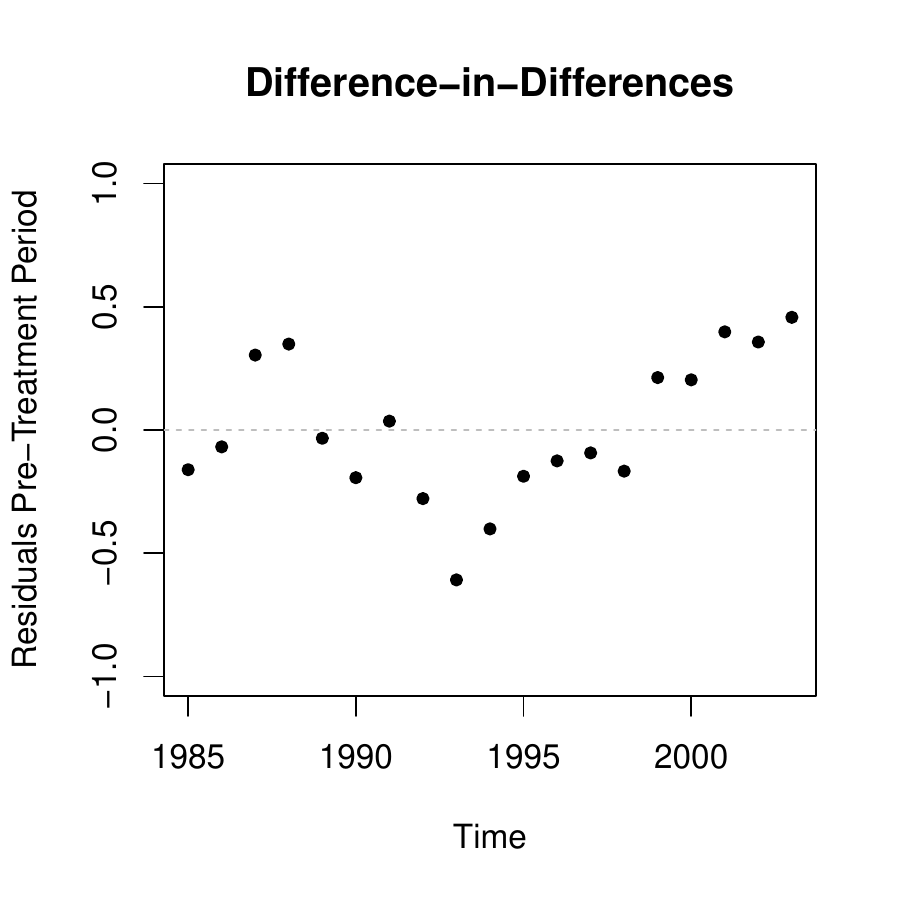}
\includegraphics[width=0.325\textwidth,trim={0 1cm 0 1cm}]{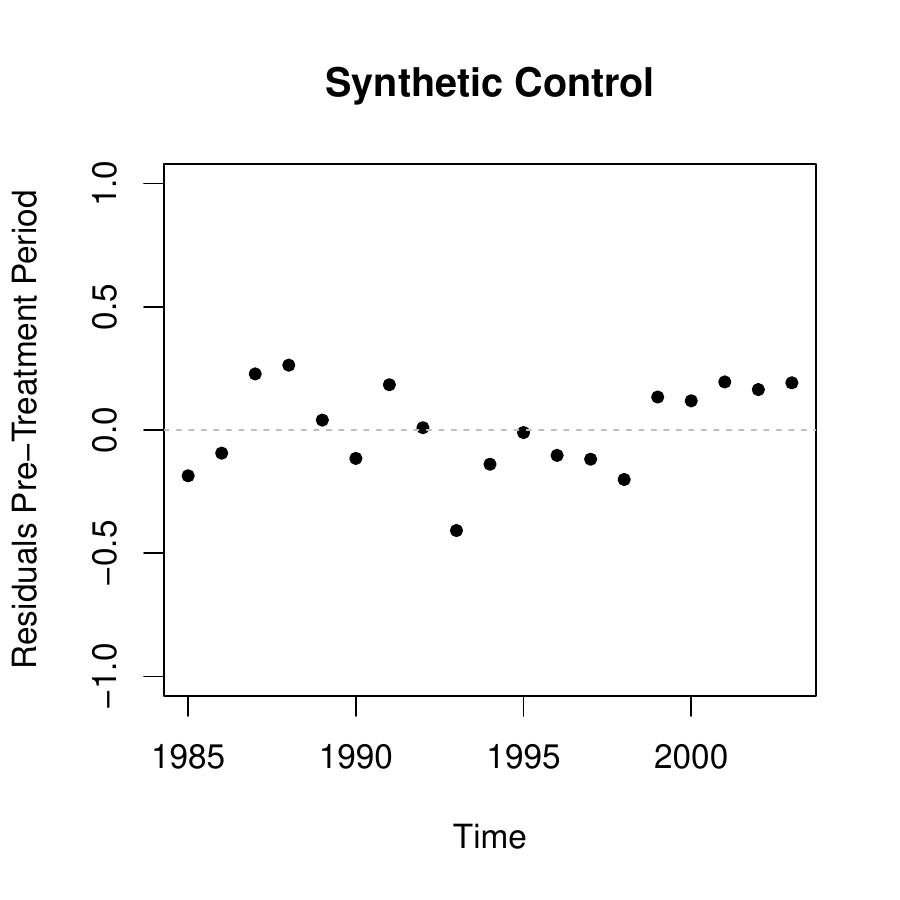}
\includegraphics[width=0.325\textwidth,trim={0 1cm 0 1cm}]{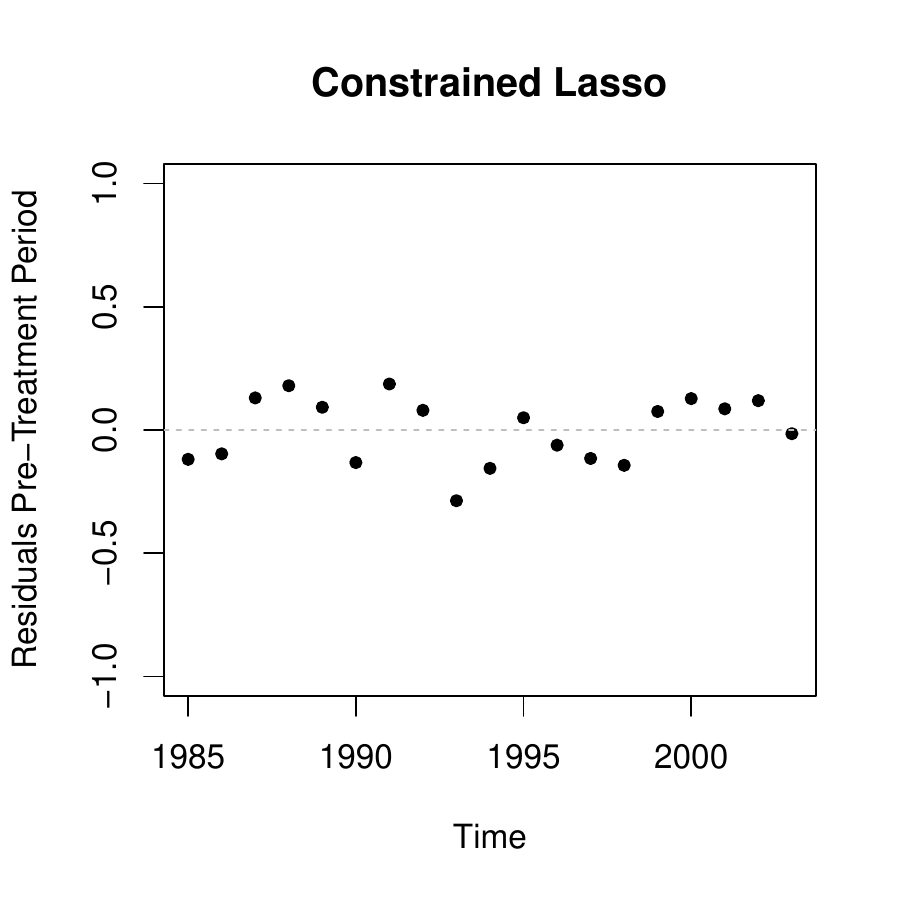}

\end{center}
\label{fig:placebo_graphical}
    {\footnotesize \textit{Notes:} Data are from \citet{cunningham2018decriminalizing}. The figure plots the pre-treatment residuals estimated using difference-in-differences, SC, and constrained Lasso.} 

\end{figure}

\begin{figure}[H]
\caption{Pointwise Confidence Intervals}

\begin{center}

\includegraphics[width=0.325\textwidth,trim={0 1cm 0 1cm}]{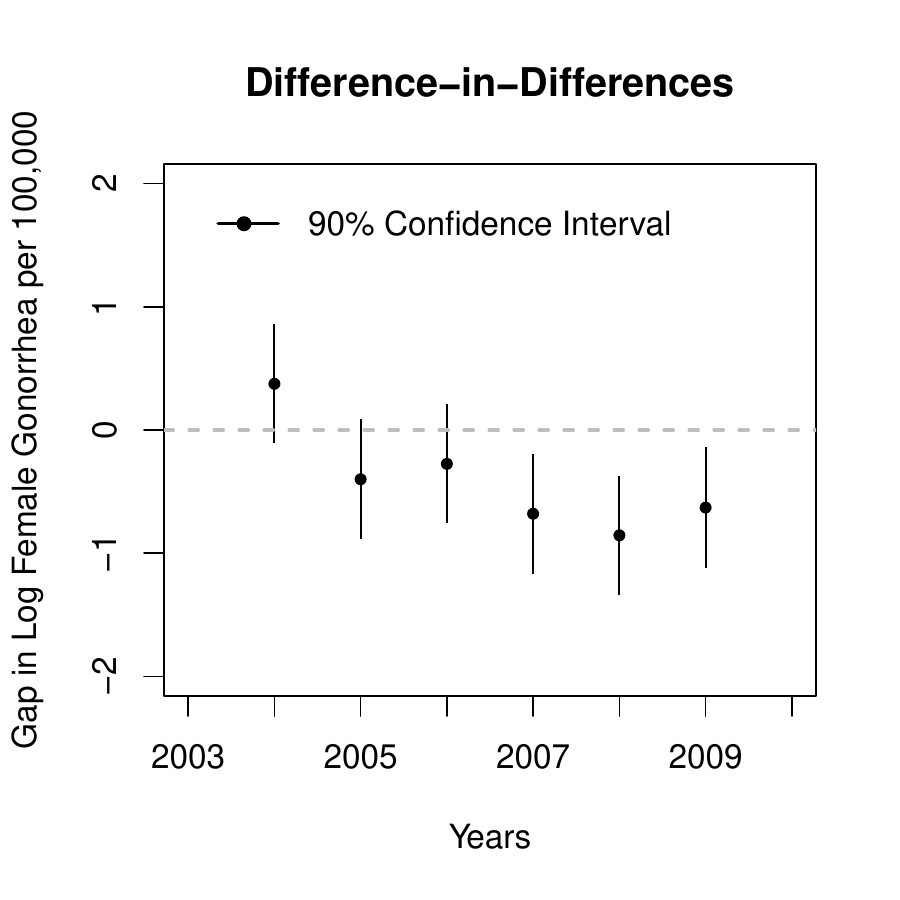}
\includegraphics[width=0.325\textwidth,trim={0 1cm 0 1cm}]{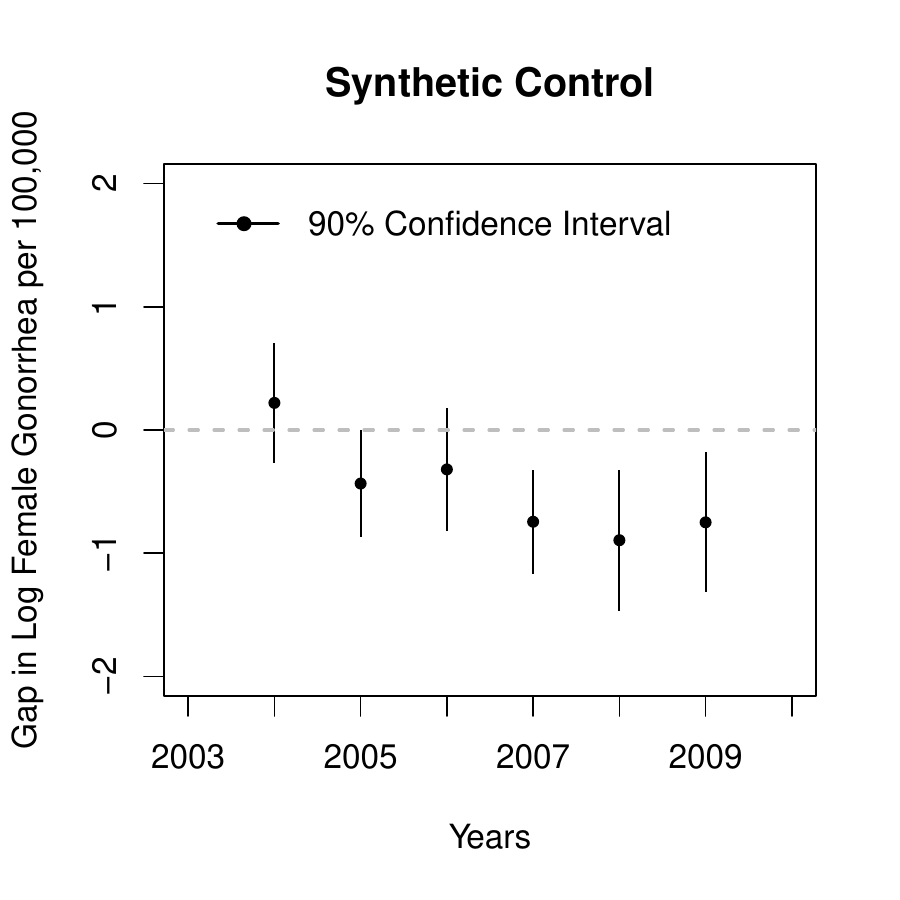}
\includegraphics[width=0.325\textwidth,trim={0 1cm 0 1cm}]{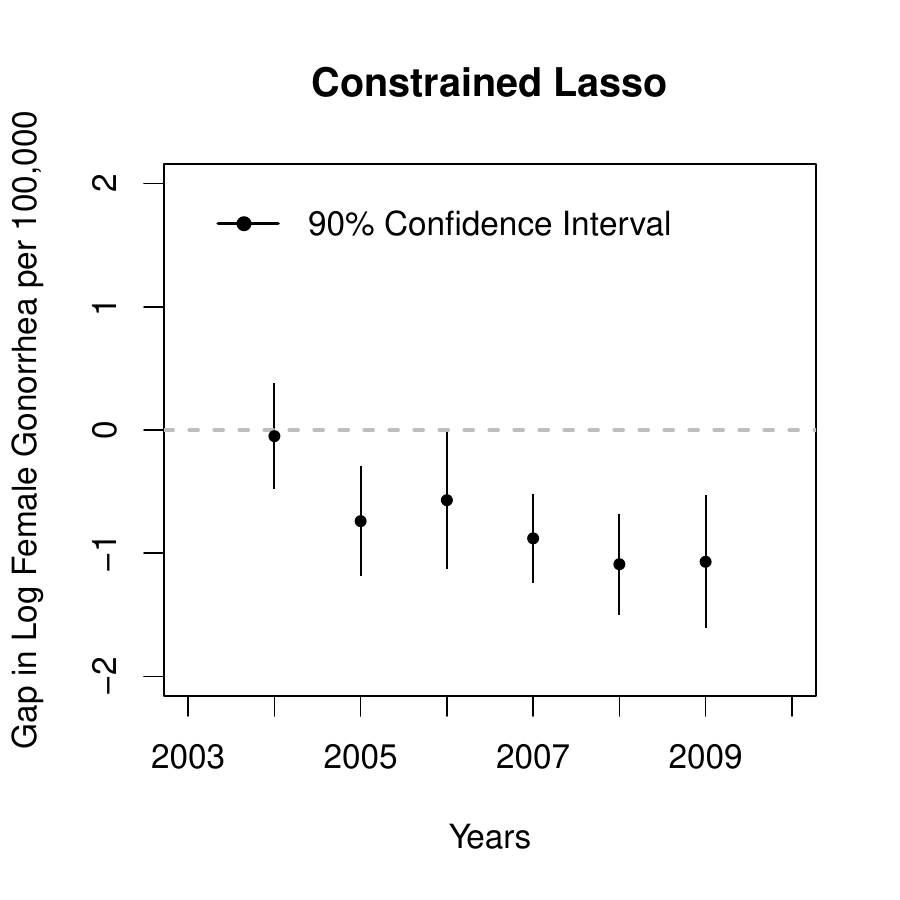}

\end{center}
\label{fig:ci_gonorrhoea}

  {\footnotesize \textit{Notes:} Data are from \citet{cunningham2018decriminalizing}. The figure plots pointwise 90\% confidence intervals computed using Algorithm \ref{algo:pointwise_ci}.} 

\end{figure}

\begin{figure}[H]
\caption{Leave-one-out Robustness Checks}

\begin{center}
\includegraphics[width=0.45\textwidth,trim={0 1cm 0 1cm}]{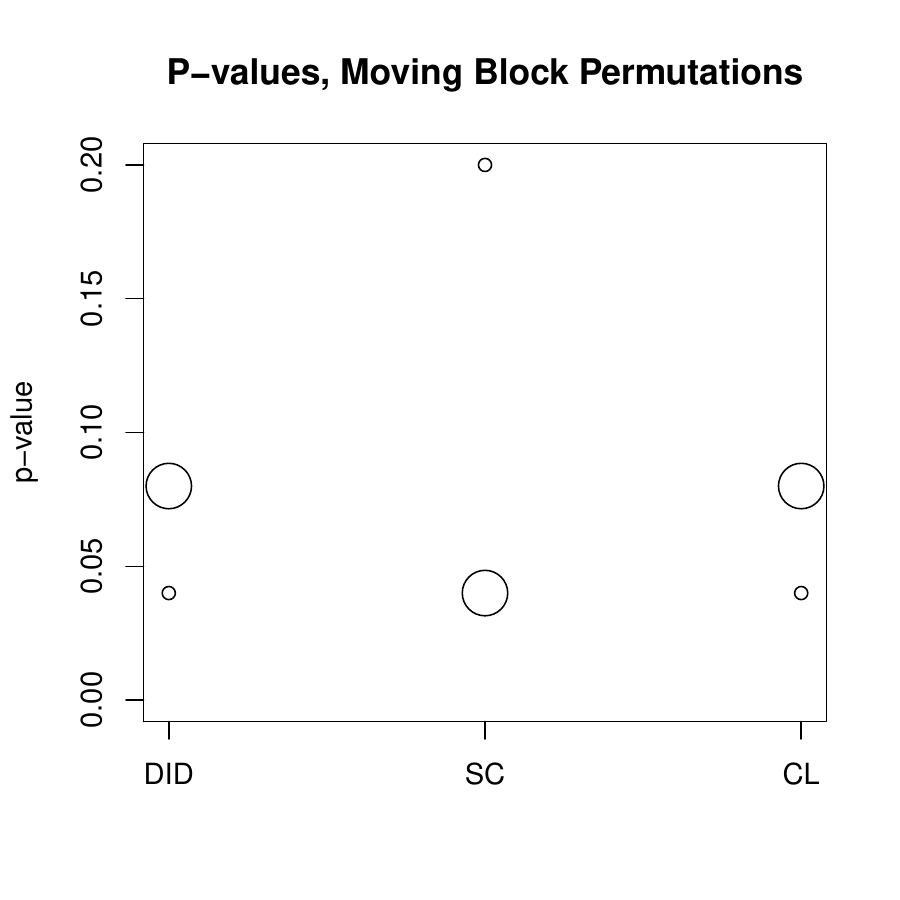}
\includegraphics[width=0.45\textwidth,trim={0 1cm 0 1cm}]{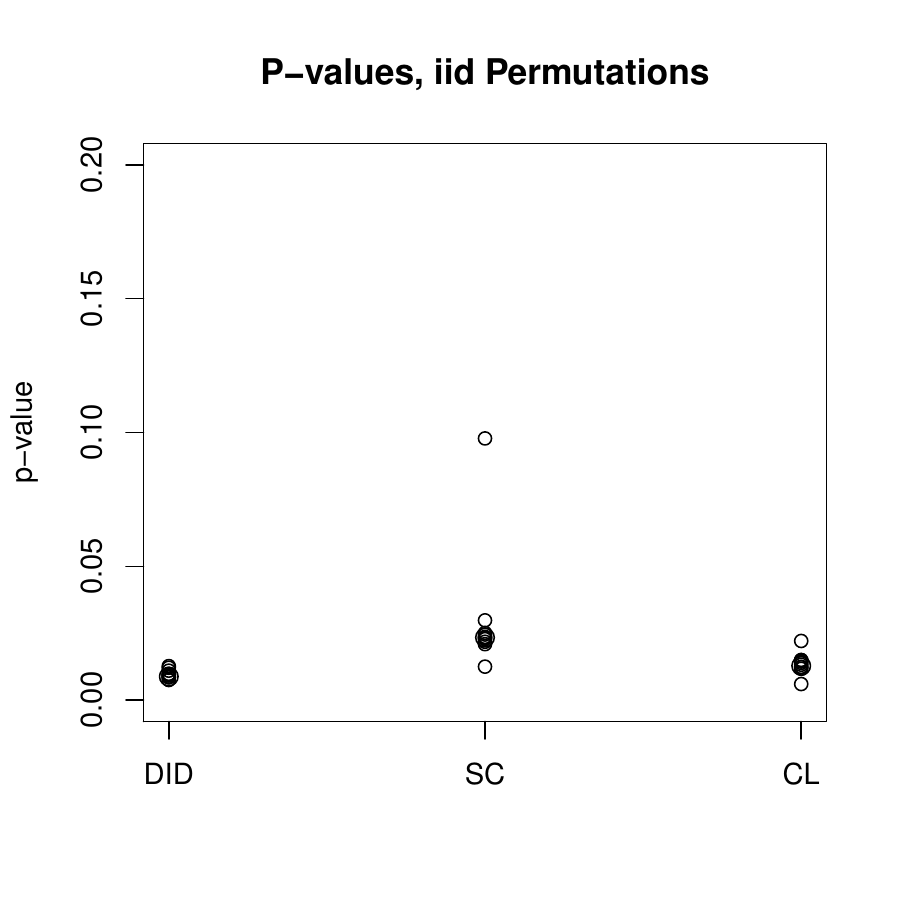}
\end{center}
    {\footnotesize \textit{Notes:} Data are from \citet{cunningham2018decriminalizing}. This figure shows the distribution of $p$-values from testing null hypothesis \eqref{eq:zero_effect}, leaving-out one of the control states with non-zero weight at the time. The size of the circles is proportional to the number of $p$-values. DID: difference-in-differences; SC: synthetic control; CL: constrained Lasso.} 

\label{fig:robustness}
\end{figure}

\section*{Tables Main Text}

\begin{table}[H]
\begin{center}
\caption{Placebo Specification Tests}
\label{tab:placebo_specification_tests}
\footnotesize
\begin{tabular}{lcccccc}
\\
\toprule
\midrule
 &\multicolumn{3}{c}{Moving Block Permutations} & \multicolumn{3}{c}{iid Permutations} \\
\cmidrule(l{5pt}r{5pt}){2-4} \cmidrule(l{5pt}r{5pt}){5-7}  \ 
$\tau$ &Diff-in-Diffs&Synth. Control&Constr. Lasso&Diff-in-Diffs&Synth. Control&Constr. Lasso  \\
\midrule
1 & 0.11 & 0.32 & 1.00 & 0.11 & 0.31 & 1.00 \\ 
  2 & 0.16 & 0.32 & 0.89 & 0.06 & 0.31 & 0.93 \\ 
  3 & 0.11 & 0.26 & 1.00 & 0.03 & 0.25 & 0.95 \\ 
  \midrule
\bottomrule
\end{tabular}
\end{center}
    {\footnotesize \textit{Notes:} Data are from \citet{cunningham2018decriminalizing}. Table shows $p$-values from testing $H_0: \theta_{2003-\tau+1}=\dots=\theta_{2003}=0$ for $\tau\in \{1,2,3\}$ based on the pre-treatment data.} 

\normalsize
\end{table}

\begin{table}[H]
\begin{center}
\caption{Zero Effect Null Hypothesis}
\label{tab:zero_effect}

\footnotesize

\begin{tabular}{cccccc}
\\
\toprule
\midrule

 \multicolumn{3}{c}{Moving Block Permutations} & \multicolumn{3}{c}{iid Permutations} \\
\cmidrule(l{5pt}r{5pt}){1-3} \cmidrule(l{5pt}r{5pt}){4-6}  \ 
Diff-in-Diffs&Synth. Control&Constr. Lasso&Diff-in-Diffs&Synth. Control&Constr. Lasso  \\
\midrule
 0.08 & 0.04 & 0.08 & 0.01 & 0.03 & 0.01 \\ 
  \midrule
\bottomrule

\end{tabular}
\end{center}
    {\footnotesize \textit{Notes:} Data are from \citet{cunningham2018decriminalizing}. Table shows $p$-values from testing $H_0: \theta_{2004}=\theta_{2005}=\dots =\theta_{2009}=0$.} 
\normalsize
\end{table}

\newpage

\appendix

\begin{center}
  {\bf Online Supplemental Appendix to ``An Exact and Robust Conformal Inference Method for Counterfactual and Synthetic Controls''}
\end{center}
\bigskip

  \begin{flushleft}
Victor Chernozhukov, Massachusetts Institute of Technology. Email: vchern@mit.edu\\
Kaspar W\"uthrich, University of California San Diego. Email: kwuthrich@ucsd.edu\\
 Yinchu Zhu, Brandeis University. Email: yinchuzhu@brandeis.edu\\
 \end{flushleft}

%\spacingset{1.5} % DON'T change the spacing!

\numberwithin{equation}{section}
\numberwithin{figure}{section}
\numberwithin{table}{section}

\numberwithin{thm}{section}
\numberwithin{lem}{section}

\startcontents[sections]
\printcontents[sections]{l}{1}{\setcounter{tocdepth}{1}}

\onehalfspacing

\section{Extensions}
\label{sec:extensions}

\subsection{Testing Hypotheses about Average Effects over Time}
\label{subsec:average_effects}

In addition to testing sharp null hypotheses, researchers are often also interested in testing hypotheses about average effects over time, $\bar\theta=T_\ast^{-1}\sum_{t=T_0+1}^{T}\theta_t$:
\begin{eqnarray}
H_0: \bar\theta=\bar\theta^0\label{eq:H0_average}
\end{eqnarray}

Hypothesis \eqref{eq:H0_average} can be tested by collapsing the data into averages of non-overlapping blocks of $T_\ast$ observations over the time dimension. To simplify the exposition, we assume that $T/T_\ast$ is an integer. Note that Assumption \ref{ass:dgp} implies the following model for the average potential outcomes $\bar{Y}_{1r}^N=T_\ast^{-1}\sum_{t=r}^{r+T_\ast-1}Y^N_{t}$ and $\bar{Y}_{1r}^I=T_\ast^{-1}\sum_{t=r}^{r+T_\ast-1}Y^I_{t}$: 
\begin{equation*}
\begin{array}{l}
\bar{Y}_{1r}^N = \bar{P}^N_r+\bar{u}_r\\
\bar{Y}_{1r}^I = \bar{P}^N_r+ \bar{\theta}_r +\bar{u}_r \\
\end{array}   \Bigg | \quad E( \bar{u}_r) = 0,\quad r=1,T_\ast+1,2T_\ast+1,\dots,T_0-T_\ast+1,T_0+1,\label{eq:dgp_average}\\
\end{equation*}
where $\bar{P}_{r}^N=T_\ast^{-1}\sum_{t=r}^{r+T_\ast-1}P^N_{t}$ and $\bar{u}_{r}=T_\ast^{-1}\sum_{t=r}^{r+T_\ast-1}u_{t}$. Define the aggregated (collapsed) data under the null as $\bar{\Zb}=(\bar{Z}_1,\dots,\bar{Z}_{T_0+1})'$, where
\[
\bar{Z}_r=\begin{cases}
\left(\bar{Y}_{1r}^{N},\bar{Y}_{2r}^N,\dots,\bar{Y}^N_{J+1r},\bar{X}'_{1r},\dots,\bar{X}'_{J+1r}\right)', & r<T_{0}+1\\
\left(\bar{Y}_{1r}^{I}-\bar{\theta}^{o},\bar{Y}_{2r}^N,\dots,\bar{Y}^N_{J+1r},\bar{X}'_{1r},\dots,\bar{X}'_{J+1r}\right)', & r=T_0+1
\end{cases}
\]
and $\bar{X}_{jr}=T_\ast^{-1}\sum_{t=r}^{r+T_\ast-1}X_{jt}$ for $j=1,\dots,J+1$. Note that testing hypothesis \eqref{eq:H0_average} is equivalent to testing a hypothesis concerning a per-period effect based on the aggregated data $\bar{\Zb}$. Specifically, we estimate the average proxy $\hat{\bar{P}}_r^N$ based on the aggregated data $\bar{\Zb}$ and obtain the residuals $\hat{\bar{u}}=(\hat{\bar{u}}_1,\hat{\bar{u}}_{T_\ast+1},\dots, \hat{\bar{u}}_{T_0+1})'$, 
where $\hat{\bar{u}}_r=\bar{Y}_{1r}^{N}- \hat{\bar{P}}^N_r$. The test statistic is $S\left( \hat{\bar{u}}\right)$, and $p$-values can be obtained based on permutations of $\hat{\bar{u}}$ as described in Section \ref{subsec:hypotheses}. 

The key assumption underlying this procedure is that the average mean proxy $\bar{P}_r^N$ can be identified and consistently estimated based on the aggregated data $\bar{\Zb}$. This is the case for SC and the other regression-based estimators discussed in Sections \ref{ex:sc}--\ref{ex:pen_reg}, provided that $E(\bar{u}_r\bar{Y}_{jr}^N)=0$ for $2\le j\le J+1$ and that the sufficient conditions for consistent estimation in Section \ref{sub: low level SC} hold for the aggregate data. By contrast, identification and estimation of $\bar{P}_r^N$ may not be possible for nonlinear and dynamic models. Under consistent estimation of $\hat{\bar{P}}^N_r$, the formal properties of the test follow from the results in Section \ref{subsec:small_error} since stationarity and weak dependence of $\{u_t\}$ imply stationarity and weak dependence of $\{\bar{u}_r\}$. Alternatively, if $\hat{\bar{P}}^N_r$ can be shown to be stable and the aggregate data are stationary and weakly dependent, the properties of the test follow from the results in Section \ref{subsec:stability}. Finally, we emphasize that the effective sample size is $T/T_\ast$ instead of $T$, such that $T$ needs to be substantially larger than $T_\ast$.

\subsection{Multiple Treated Units}\label{ssec:muiltiple_treated_units}
Our method can be extended to accommodate multiple treated units by collapsing the data into averages across the treated units. Consider a setup with $L$ treated units, indexed by $j=1,\dots,L$, and $J$ control units, indexed by $j=L+1,\dots,J+L$. Suppose that Assumption \ref{ass:dgp} holds for all treated units:
\begin{equation*}
\begin{array}{l}
Y_{jt}^N = P_{jt}^N + u_{jt}\\
Y_{jt}^I = P_{jt}^N + \theta_{jt} + u_{jt}  \\
\end{array}   \Bigg | \quad E( u_{jt}) = 0,  \quad t=1,\dots,T , \quad j=1,\dots,L. \\
\end{equation*}
Under this assumption, hypotheses about the unit-specific policy effects $\{\theta_{jt}\}$ can be tested by separately applying the proposed inference procedure to each treated unit. In addition, one is often also interested in conducting inferences about the average treatment effects on the treated units, $\{\bar\theta_t\}$, where $\bar{\theta}_t=L^{-1}\sum_{j=1}^L\theta_{jt}$.

Specifically, consider the following null hypothesis:
\begin{equation}
H_0: \left(\bar{\theta}_{T_0+1},\dots, \bar{\theta}_{T} \right)=\left(\bar{\theta}_{T_0+1}^0,\dots, \bar{\theta}_{T}^0 \right). \label{eq:null_multiple_treated_units}
\end{equation}
To test hypothesis \eqref{eq:null_multiple_treated_units}, note that if Assumption \ref{ass:dgp} holds for all treated units, we have the following model for the average potential outcomes $\bar{Y}^N_{t}=L^{-1}\sum_{j=1}^LY_{jt}^N$ and $\bar{Y}^I_{t}=L^{-1}\sum_{j=1}^LY_{jt}^I$:
\begin{equation*}
\begin{array}{l}
\bar{Y}_{t}^N = \bar{P}_{t}^N + \bar{u}_{t}\\
\bar{Y}_{t}^I = \bar{P}_{t}^N + \bar\theta_{t} + \bar{u}_{t}  \\
\end{array}   \Bigg | \quad E( \bar{u}_{t}) = 0,  \quad t=1,\dots,T, \\
\end{equation*}
where $\bar{P}_{t}^N=L^{-1}\sum_{j=1}^LP_{jt}^N$ and $\bar{u}_t=L^{-1}\sum_{j=1}^Lu_{jt}$. Define the data under the null as $\bar{\Zb}=(\bar{Z}_1,\dots,\bar{Z}_{T})'$, where
\[
\bar{Z}_{t}=\begin{cases}
\left(\bar{Y}^N_{t},Y^N_{L+1t},\dots,Y^N_{J+Lt},\bar{X}'_{t},X'_{L+1t},\dots,X'_{J+Lt}\right)', & t\le T_0,\\
\left(\bar{Y}_{t}^I-\bar\theta_t^{o},Y^N_{L+1t},\dots,Y^N_{J+Lt},\bar{X}'_{t},X'_{L+1t},\dots,X'_{J+Lt}\right)', & t>T_0,
\end{cases}
\]
and $\bar{X}_t=L^{-1}\sum_{j=1}^LX_{jt}$. To test hypothesis \eqref{eq:null_multiple_treated_units}, we compute the estimated average proxy $\hat{\bar{P}}_t^N$ based on the aggregated data $\bar{\Zb}$ and obtain the residuals $\hat{\bar{u}}=(\hat{\bar{u}}_1,\dots, \hat{\bar{u}}_{T})'$, where $\hat{\bar{u}}_t=\bar{Y}_{t}^{N}- \hat{\bar{P}}^N_t$ for  $t=1,\dots,T$. The test statistic is $S\left( \hat{\bar{u}}\right)$, and $p$-values can be obtained based on permutations of $\hat{\bar{u}}$ as described in Section \ref{subsec:hypotheses}.
 The formal properties of this test follow from the results in Section \ref{sec:theory}. 
 
\subsection{Placebo Tests}
\label{subsec:specification_tests}

Here we propose easy-to-implement placebo tests for assessing the credibility of inferences based on our method. We recommend applying these placebo tests when using our inference procedures.

Following \citet{abadie2015comparative}, the idea is to consider a placebo intervention before the actual intervention took place. For a given $\tau\geq 1$, we use our method to test the null hypothesis 
\begin{equation}
H_0: \theta_{T_0-\tau+1}=\dots=\theta_{T_0}=0 \label{eq:H0_placebo}
\end{equation}
based in the pre-treatment data $\Zb=\left( Z_1,\dots,Z_{T_0}\right)'$. Using an appropriate CSC method, we compute the counterfactual mean proxies $\hat{P}_t^N$ based on $\Zb$ and obtain the residuals 
\[
\hat{u}=\left(\hat{u}_1,\dots,\hat{u}_{T_0} \right)', \quad \hat{u}_t=Y_{1t}^N-\hat{P}^N_t, \quad t=1,\dots, T_0.
\]
We then apply the proposed inference method, treating $\{1,\dots,T_0-\tau\}$ as the pre-treatment period and $\{T_0-\tau+1,\dots,T_0\}$ as the post-treatment period. The theoretical properties of such placebo tests follow directly from the results in Section \ref{sec:theory}. As illustrated in our application, it may be useful to complement the formal testing results with plots of the pre-treatment residuals $\left(\hat{u}_1,\dots,\hat{u}_{T_0} \right)$.

A rejection of the null hypothesis \eqref{eq:H0_placebo} undermines the credibility of the assumptions underlying our procedure and the inferences on the policy effects in the post-treatment period. While non-rejections provide evidence in favor of our method, it is important to emphasize that such non-rejections do not ``prove'' that our method is valid. Moreover, by construction, the placebo tests cannot be used to assess some of the key assumptions underlying our approach, such as the invariance of the distribution of $\{u_t\}$ under the intervention.

Given our inference method's genericness, the placebo tests discussed here can be used to assess and compare the credibility of different CSC methods. For example, in our empirical application, the placebo tests provide evidence in favor of SC and constrained Lasso, but suggest that the difference-in-differences results need to be interpreted with caution.

 \section{Interpretation as a Structural Breaks Test}
 \label{app: interpretation structural break}
A key assumption underlying our method is the invariance of the distribution of $\{u_t\}$ under the intervention (Assumption \ref{ass:dgp}). In settings where this assumption fails, following the literature on end-of-sample structural breaks tests \citep[e.g.,][]{andrews2003end}, our procedure can be used as a test of the null hypothesis that the policy has no impact whatsoever against the alternative hypothesis that $\theta\ne 0$ and or the policy affects the distribution of $\{u_t\}$. 

Consider the following testing problem:
\begin{equation}
H_0:\begin{cases}
  Y_{1t}^N=P_t^N+u_t, ~~E(u_t)=0,~~t=1,\dots,T, ~~ \theta=0, ~~\text{and}\\
  \{u_t\}_{t=0}^T \text{ is stationary and weakly dependent}.
 \end{cases}\label{eq:H0_structural_break}
\end{equation}
against
\begin{equation}
H_1:\begin{cases}
  Y_{1t}^N=P_t^N+u_t,~~t=1,\dots,T,~~\theta\ne 0~~\text{and or the distribution of}\\
 \{u_t\}_{t=T_0+1}^{T} \text{ differs from that of } \{u_t\}_{t=s}^{s+T_\ast-1} \text{ for } s=1,\dots, T_0-T_\ast+1.
 \end{cases} \label{eq:H1_structural_break}
\end{equation}
Hypothesis \eqref{eq:H0_structural_break} can be tested by applying our procedure to the data under the null where $Y_{1t}^N=Y^I_{1t}=Y_{1t}$ for all $t=1,\dots,T$. The theoretical size properties of this test follow directly from the results in Section \ref{sec:theory}. This test has power against location shifts induced by $\theta \ne 0$ as well as changes in the distribution of $\{u_t\}$ that increase the quantiles of $S(u)$ (e.g., scale shifts);  see \citet[][Section 2.5]{andrews2003end} for a related discussion.

\section{Prediction Sets for Random Policy Effects}
\label{app: random treatment effects}

In the main text, we assume that the policy effect sequence $\{\theta_t\}$ is fixed (Assumption \ref{ass:dgp}). Here we show that our procedure generates valid prediction sets for $\theta_t$ when $\theta_t$ is assumed to be random, as, for example, in \citet{cattaneo2021prediction}.\footnote{Following the literature on conformal prediction \citep[e.g.,][]{lei2014distribution,lei2017distributionfree}, we use the terminology ``prediction set'' instead of  ``confidence set'' because the resulting sets are not conventional confidence sets; see \citet{cattaneo2021prediction} for a discussion of the differences between confidence intervals and prediction intervals for SC.} Our analysis here is in the same spirit as the classical conformal prediction literature, which aims at constructing prediction intervals for future values of a (random) target quantity of interest. Specifically, we will show that when $\theta_t$ is random, Algorithm \ref{algo:pointwise_ci} provides unconditionally valid $(1-\alpha)$-prediction sets with non-asymptotic performance guarantees. 

To state the result, define  $\Zb^{*}=(Z_1^*,\dots,Z_{T}^*)'$ with $ Z_t^*=\left(Y^N_{1t},Y^N_{2t},\dots,Y^N_{J+1t},X'_{1t},\dots,X'_{J+1t}\right)'$ for $1\leq t\leq T_0+T_*$. Notice that $\Zb^*$ contains the true counterfactuals.

\begin{thm}[Prediction Sets]
\label{thm:approximate_validity CS} 
Assume that $T_*$ is fixed. Suppose that $Y_{1t}^{N}=P_{t}^N+u_t$, $1\le t\le T$, where $\{u_t\}$ is a centered and stationary stochastic process, and that the policy effects $\theta_t:=Y_{1t}^I-Y_{1t}^N$ are random. Suppose that Assumption \ref{ass:high_level_est_error} holds for $\hat{P}_N$ computed using $\Zb^*$. Impose  Assumption \ref{ass:u}.\ref{ass:u_iid} if $\Pi=\Pi_{\text{all}}$.  Impose Assumption \ref{ass:u}.\ref{ass:u_weak_dependent},
if $\Pi=\Pi_{\to}$. Assume the statistic $S(u)$ has a density function bounded by $D$. Then
\[
|P\left(\theta_{t}\in \mathcal{C}_{1-\alpha}(t)\right)- (1-\alpha) | \leq C ( \tilde \delta_T + \delta_T + \sqrt{\delta_T} + \gamma_T),
\]
where $ \mathcal{C}_{1-\alpha}(t)$ is defined in Algorithm \ref{algo:pointwise_ci} and $\tilde \delta_T = (T_*/T_0)^{1/4}(\log T)$. The constant $C$ depends on $T_*$, $M$ and $D$,  but not on $T$.
\end{thm}

Let us briefly discuss the interpretation of our model when  $\theta_t$ is regarded as random. Suppose that the potential outcomes are $Y_{1t}^I=P^I_t+u_t^I$ and $Y_{1t}^N=P^N_t+u_t^N$. The counterfactual mean proxies and the prediction errors under the policy, $\{P_t^I\}$ and $\{u_t^I\}$, may differ from the counterfactual mean proxies and the errors in the absence of the policy, $\{P_t^N\}$ and $\{u_t^N\}$. Algorithm \ref{algo:pointwise_ci} provides prediction intervals for 
$$\theta_t:=Y_{1t}^{I}-Y_{1t}^{N}=\Delta_t^P+\Delta_t^u,$$
where $\Delta_t^P:=P^I_t-P_t^N$ and $\Delta_t^u:=u^I_t-u_t^N$. Here  $\theta_t$ captures both the effects of the policy on the mean proxy, $\Delta_t^P$, and the effect on distribution of the error, $\Delta_t^u$ (e.g., a scale shift).  Therefore, we are only assuming that in the absence of the policy, $\{u_t^N\}$ is a stationary process. All the changes to the mean, the variance, or other features of the distribution due to the policy are captured by the policy effects $\{\theta_t\}$.

In sum, with random policy effects, our procedure simply provides a prediction set for the policy effects defined as the difference between the two potential outcomes.

\section{Model-free Exact Validity under Exchangeability}
\label{app: exchangeability}
Every permutation procedure that is approximately valid in time series settings should have good properties in ``ideal'' settings where the data are iid or exchangeable. The following theorem, which is based on standard arguments, shows that under exchangeability of the data, our conformal inference approach achieves exact finite sample size control. This result is model-free in the sense that we do not need to use a correct or consistent estimator for the counterfactual mean proxy. As a result, our procedure controls size under arbitrary forms of misspecification and is fully robust against overfitting.

\begin{thm}[Exact Validity]
\label{thm:finite_sample} Let $\Pi$ be $\Pi_{\to}$ or $\Pi_{\rm{all}}$. Suppose that the data $\{Z_t\}_{t=1}^T$ are iid or exchangeable with respect to $\Pi$ under the null hypothesis and that $\hat u_t =g (Z_t, \hat \beta)$, where the estimator $\hat \beta = \hat \beta (\{Z_t\}_{t=1}^T)$ is invariant with respect to any permutation of the data. Then, under the null hypothesis, $\{\hat{u}_t\}_{t=1}^T$ is an exchangeable sequence and the permutation $p$-value is unbiased in level:
\[
P\left(\hat{p}\le \alpha\right)\le\alpha.
\]
Moreover, if $\left \{S(\hat{u}_\pi) \right \}_{\pi\in\Pi}$ has a continuous distribution,
\[
\alpha-\frac{1}{|\Pi|}\le P\left(\hat{p}\le \alpha\right).
\]
\end{thm}

Theorem \ref{thm:finite_sample} requires that the estimators are invariant under permutations of the data under the null hypothesis. Invariance holds for regression-based estimators such as SC, constrained Lasso, or penalized regression, provided that the null hypothesis is imposed for estimation but may fail for dynamic models such as linear and non-linear autoregressive models.

An inspection of the proof of Theorem \ref{thm:finite_sample} shows that our procedure achieves finite sample size control whenever the residuals $\{\hat{u}_t\}_{t=1}^T$ are exchangeable. We demonstrate that exchangeability of  $\{\hat{u}_t\}_{t=1}^T$ is implied if the data $\{Z_t\}_{t=1}^T$ are iid or exchangeable. Exchangeability $\{\hat{u}_t\}_{t=1}^T$ may hold even if the data are not exchangeable.  For example, in the difference-in-difference model, the outcome data can have an arbitrary common trend eliminated by differencing, making it possible for the residuals to be iid or exchangeable with non-iid data.

The finite sample validity under exchangeability is crucial for the robustness of our proposal. While high-dimensional CSC approaches may overfit in small samples, Theorem \ref{thm:finite_sample} shows that the proposed method does not suffer from overfitting or misspecification. The fundamental reason is that we exploit symmetry rather than to completely rely on the consistency of the estimator.  Imposing the null hypothesis when estimating the model leads to invariance of the estimator under permutations, exchangeable residuals, and finite sample validity. As a result, the proposed procedure is more robust than alternative approaches that estimate the proxies based on the pre-treatment data without imposing the null hypothesis.

Even in time series settings where exchangeability fails, imposing the null hypothesis for estimation is crucial for achieving a good performance in typical CSC applications where $T_0$ is rather small \citep[e.g.,][]{abadie2003economic,abadie10sc,abadie2015comparative,DI16,cunningham2018decriminalizing}. Figure \ref{fig:imposing_H0} augments Figure \ref{fig:importance_H0} with results for $T=100$  $(T_0=99,T_\ast=1)$. In the empirically relevant case where $T_0=19$, estimating $P_t^N$ under the null yields an excellent performance. Irrespective of the degree of persistence, size accuracy is substantially better than when $P_t^N$ is estimated based on pre-treatment data only. Even when $T_0=99$, which is much larger than the $T_0$ in many CSC applications, imposing the null yields notable performance improvements. In fact, the size accuracy is better for $T_0=19$ when $P_t^N$ is estimated under the null than for $T_0=99$ when $P_t^N$ is estimated based on the pre-treatment data only.

\section{Sufficient Conditions for Estimator Stability}
\label{sec: suff condition perturbation stability}

In this section, we provide sufficient conditions for the estimator stability Assumption \ref{assu: stability}. We first present generic sufficient conditions for low-dimensional models. For high-dimensional models, the theoretical analysis is more difficult and a case-by-case analysis is needed. We are not aware of any theoretical work that establishes Assumption \ref{assu: stability} for any high-dimensional model. Here we verify the stability condition for constrained Lasso; stability of Ridge regression is verified in Appendix \ref{app:ridge}.
% our preferred high-dimensional method: constrained Lasso.

\subsection{Generic Sufficient Condition for Low-dimensional Models}\label{sec: pertubation stability low dim}

Consider $\hat{\beta}(\Zb)=\arg\min_{\beta\in\Bcal}\hat{L}(\Zb;\beta)$,
where $\hat{L}(\Zb;\beta)$ is a loss function and $\Bcal\subset\mathbb{R}^{p}$
for a fixed $p$. Let $\Hcal$ be a set of subsets of $\{1,\dots,T\}$. Notice that Assumption \ref{assu: stability} only requires $\Hcal$ to be a singleton, but in this subsection and the next, we allow $\Hcal$ to be a class of subsets.

\begin{lem}
	\label{lem: low-dimensional stability}Suppose that the following
	conditions hold:
\begin{enumerate} \setlength{\itemsep}{0pt} \setlength{\parskip}{0pt}
\item $\sup_{\beta\in\Bcal}|\hat{L}(\Zb;\beta)-L(\beta)|=o_{P}(1)$
	for some non-random $L(\cdot)$. 
\item $\max_{H\in\Hcal}\sup_{\beta\in\Bcal}|\hat{L}(\Zb_{H};\beta)-L(\beta)|=o_{P}(1)$.
\item  $L(\cdot)$ is continuous at $\beta_{*}$, $\min_{\beta}L(\beta)$
	has a unique minimum at $\beta_{*}$ and $\Bcal$ is compact.
\end{enumerate}	
	Then $\max_{H\in\Hcal}\|\hbeta(\Zb)-\hbeta(\Zb_{H})\|_{2}=o_{P}(1)$. 
\end{lem}

In the literature of misspecified models, $\beta_{*}$ is usually
referred to as the pseudo-true value \citep[e.g.,][]{white1996estimation}.
In M-estimation with $\hat{L}(\Zb;\beta)=T^{-1}\sum_{t=1}^{T}l(Z_{t};\beta)$,
one can often show $\sup_{\beta}|\hat{L}(\Zb;\beta)-L(\beta)|=o_{P}(1)$
with $L(\beta)=El(Z_{1};\beta)$; in GMM models with $\hat{L}(\Zb;\beta)=\|T^{-1}\sum_{t=1}^{T}\psi(Z_{t};\beta)\|_{2}$,
one can often use $L(\beta)=\|E\psi(Z_{1};\beta)\|_{2}$. 

The proof of Lemma \ref{lem: low-dimensional stability} shows that
$\|\hbeta(\Zb)-\beta_{*}\|_{2}=o_{P}(1)$ and $\max_{H\in\Hcal}\|\hbeta(\Zb_{H})-\beta_{*}\|_{2}=o_{P}(1)$.
In other words, the stability of the estimator arises from the consistency
to the pseudo-true value $\beta_{*}$. Such consistency holds under
very weak conditions. We essentially only require a uniform law of
large numbers. This can be verified for many low-dimensional models
under weakly dependent data. The conclusion of Lemma  \ref{lem: low-dimensional stability} translates to  Assumption \ref{assu: stability} once we derive a bound on  $\max_{\pi\in\Pi}\sup_{\beta_1\neq \beta_2}|S(\Zb^\pi;\beta_1)-S(\Zb^\pi;\beta_2)|/\|\beta_1-\beta_2\|_2 $; this requires knowledge of  the model structure. 

For example, suppose that $\max_{\pi\in\Pi}\sup_{\beta_1\neq \beta_2}|S(\Zb^\pi;\beta_1)-S(\Zb^\pi;\beta_2)|/\|\beta_1-\beta_2\|_2 =O_P(1) $.\footnote{We only need $S(\Zb;\beta)$ to be Lipschitz with respect to $\beta$. In the simple example of $S(\Zb;\beta)=|Y_{T_0+1}-X_{T_0+1}'\beta|$, this only requires  $\|X_{T_0+1}\|_2$ to be bounded.} Then $\max_{\pi\in\Pi}|S(\Zb^\pi;\hbeta(\Zb))-S(\Zb^\pi;\hbeta(\Zb_{H}))|=o_P(1)$. Here is how we apply Theorem \ref{thm: approx exchange} to obtain the asymptotic size control. Fix an arbitrary $\delta>0$. We can simply choose the constant function $\varrho_T(x)=\delta$ for Assumption \ref{assu: stability}. Since $\max_{\pi\in\Pi}|S(\Zb^\pi;\hbeta(\Zb))-S(\Zb^\pi;\hbeta(\Zb_{H}))|=o_P(1)$,  Assumption \ref{assu: stability} holds for some $\gamma_{1,T}=o(1)$ (due to the definition of convergence in probability).
Assume that $\Psi(x,\beta)$ has bounded derivative with respect to $x$. Then we can choose $\xi_T=K_1$ for a large constant $K_1$ and $\gamma_{2,T}=0$. As a result, Theorem \ref{thm: approx exchange} has 
\begin{align*}
& \left|P\left(\hat{p}\leq\alpha\right)-\alpha\right|\\
&\leq C_{1}\sqrt{\xi_{T}\varrho_{T}(T_{0}/R+2k)}+C_{1}\left(T_{0}^{-1}R[\log(T_{0}/R)]^{1/D_{3}}\right)^{1/4}+C_{1}\exp\left(-(k-T_{*}+1)^{1/D_{3}}\right)\\
&\qquad+C_{1}\sqrt{\gamma_{1,T}}+C_{1}\sqrt{\gamma_{2,T}}\\
& =C_{1}\sqrt{K_1 \delta}+C_{1}\left(T_{0}^{-1}R[\log(T_{0}/R)]^{1/D_{3}}\right)^{1/4}+C_{1}\exp\left(-(k-T_{*}+1)^{1/D_{3}}\right)+C_{1}\sqrt{o(1)}.
\end{align*}
Since $R\asymp T_0/\log(T_0)$, we have $T_{0}^{-1}R[\log(T_{0}/R)]^{1/D_{3}}=o(1)$. Choosing $k\asymp \log(T_0)$ and assuming that $T_*$ is fixed, we obtain $\exp\left(-(k-T_{*}+1)^{1/D_{3}}\right)=o(1)$. Therefore, the above display implies 
$$
\left|P\left(\hat{p}\leq\alpha\right)-\alpha\right|=C_{1}\sqrt{K_1 \delta}+o(1).
$$
Since $\delta>0$ is arbitrary, we have $\left|P\left(\hat{p}\leq\alpha\right)-\alpha\right|=o_P(1)$.

\subsection{Constrained Lasso}

Here we propose sufficient conditions for estimator stability for constrained Lasso. In contrast to Sections \ref{ex:sc} and \ref{sub: low level SC}, we do not impose correct specification but study the behavior of the constrained Lasso estimator under potential misspecification. To make this explicit, we use $\beta$ instead of $w$ to denote the coefficient vector in this subsection.  Here, it is possible that $EX_t(Y_t-X_t'\beta)\neq 0 $ for any $\beta \in \Wcal$. In practice, this arises when  the relationship between $X_t$ and $Y_t$ is non-linear or when the constraint set $\Wcal $ is too small. For example, the true parameter could be non-sparse with exploding $\ell_1$-norm (e.g., $\beta=(1,\dots,1)'/\sqrt{J} $).

We first introduce some additional notation. Define $Y_{t}=Y_{1t}^N$ and $X_t=(Y^N_{2t},\ldots,Y^N_{J+1t})'$ and let $\{(\tilde{Y}_{t},\tilde{X}_{t})\}_{t=1}^{T}$ be iid from the
distribution of $(Y_{1},X_{1})$ and independent of the data $\{(Y_{t},X_{t})\}_{t=1}^{T}$. 
%With $Z_t=(Y_t,X_t)$ and $\tilde{Z}_t=(\tilde{Y}_t,\tilde{X}_t) $, we recall the notation of $\Zb$ and $\Zb_{H}$ for $H\subset\{1,...,T\}$. 
The constrained Lasso objective functions based on the data under the original data and after switching out observations with $t\in H$ are given by
\[
\hQ(\beta)=\frac{1}{T}\sum_{t=1}^{T}(Y_{t}-X_{t}'\beta)^{2}\quad \text{and} \quad \hQ_{H}(\beta)=T^{-1}\sum_{t=1}^{T}(Y_{t,H}-X_{t,H}'\beta)^{2},
\]
where $(Y_{t,H},X_{t,H})=(Y_{t},X_{t})$ for $t\notin H$ and $(Y_{t,H},X_{t,H})=(\tilde{Y}_{t},\tilde{X}_{t})$
for $t\in H$. The corresponding constrained Lasso estimators are
\[
\hbeta(\Zb)=\arg\min_{\beta\in\Wcal}\hQ(\beta)\quad \text{and}\quad \hbeta(\Zb_{H})=\arg\min_{\beta\in\Wcal}\hQ_{H}(\beta),
\]
where $\Wcal\subseteq\{v\in\RR^{J}: \|v\|_{1}\leq K\}$ and $K>0$ is a constant. Furthermore, we define $\hSigma=T^{-1}\sum_{t=1}^{T}X_{t}X_{t}'$ and  $\hmu=T^{-1}\sum_{t=1}^{T}X_{t}Y_{t}$. Similarly, for $H\subset\{1,\dots,T\}$, let $\hSigma_{H}=T^{-1}\sum_{t=1}^{T}X_{t,H}X_{t,H}'$ and
$\hmu_H=T^{-1}\sum_{t=1}^{T}X_{t,H}Y_{t,H}$. Finally, let $\Hcal$ be a set of subsets of $\{1,\dots,T\}$.

\begin{lem}
	\label{lem: classo stability}Suppose that the following conditions
	hold:
\begin{enumerate}	\setlength{\itemsep}{0pt} \setlength{\parskip}{0pt}
	
	\item with probability at least $1-\gamma_{1,T}$, $\|\hSigma_{H}-\hSigma\|_{\infty}\leq c_{T}$
	and $\|\hmu_{H}-\hmu\|_{\infty}\leq c_{T}$ for all $H\in\Hcal$.
	
	\item with probability at least $1-\gamma_{2,T}$, $\min_{\|v\|_{0}\leq s}v'\hSigma v/\|v\|_{2}^{2}\geq\kappa_{1}$.
	
	\item with probability at least $1-\gamma_{3,T}$, $\max_{H\in\Hcal}\|\hbeta(\Zb_{H})\|_{0}\leq s/2$
	and $\|\hbeta(\Zb)\|_{0}\leq s/2$. 
	\item $P(\max_{1\leq t\leq T}\|X_{t}\|_{\infty}\leq\kappa_{2})=1$.
\end{enumerate}	
	Let $\heps_{t}=Y_{t}-X_{t}'\hbeta(\Zb)$ and $\heps_{t,H}=Y_{t}-X_{t}'\hbeta(\Zb_{H})$. Then we have that 
	\[
	P\left(\max_{H\in\Hcal}\max_{1\leq t\leq T}\left|\heps_{t}-\heps_{t,H}\right|\leq2\kappa_{2}\sqrt{\kappa_{1}sc_{T}K(2K+1)}\right)\geq1-\gamma_{1,T}-\gamma_{2,T}-\gamma_{3,T}.
	\]
\end{lem}

Lemma \ref{lem: classo stability} provides sufficient conditions for  perturbation stability. Inspecting the proof, we notice that the argument does not require the estimator to converge to anything. To our knowledge, this is the first result of this kind. In the conformal prediction literature, one-observation perturbation stability has been considered in Assumption A3 of \citet{lei2017distributionfree}, who only verify it assuming correct model specification and consistent variable selection. There is also a strand of literature in statistics that considers misspecified models in high dimensions and focuses on the pseudo-true value. For example, for linear models, the pseudo-true value represents the best linear projection and is often assumed to be sparse, making it possible to establish consistency of Lasso to this pseudo-true value \citep[e.g.,][]{buhlmann2015high}. We do not make these assumptions. Lemma \ref{lem: classo stability}  allows the model to be misspecified and the pseudo-true value may or may not be consistently estimated by constrained Lasso.

Lemma \ref{lem: classo stability} says that when the solution of
constrained Lasso is sparse, the stability of $\hSigma$ and $\hmu$
guarantees the stability of the estimator. When $|H|\asymp\log T_{0}$
and the observed variables are bounded, we can choose $c_{T}\asymp T_{0}^{-1}\log(T_{0})$.
The sparse eigenvalue condition can typically be verified whenever
$s\leq cT$, where $c>0$ is a constant that depends on the eigenvalues
of $E\hSigma$. Thus, Lemma \ref{lem: classo stability} would guarantee
that when $\sup_{H\in\Hcal}|H|\lesssim\log T_{0}$, we have 
\[
\max_{H\in\Hcal}\max_{1\leq t\leq T}\left|\heps_{t}-\heps_{t,H}\right|=O_{P}(\sqrt{sT_{0}^{-1}\log T_{0}}).
\]

Therefore, whenever the solutions $\hbeta(\Zb)$ and $\hbeta(\Zb_{H})$ are
sparse enough with $s=o(T_{0}/\log(T_{0}))$, we can expect stability
of the estimated residuals. One implication is that since $\|\hbeta(\Zb)\|_{0}$
and $\|\hbeta(\Zb_{H})\|_{0}$ are clearly bounded above by $J$, the stability
should easily hold for $J\ll T_{0}/\log(T_{0})$.

We now give an explicit formula for $\varrho_T(\cdot)$ in Assumption \ref{assu: stability}. Suppose that $S(\Zb,\beta)=|T_*^{-1}\sum_{t=T_0+1}^{T_0+T_*}(Y_t-X_t'\beta)|$. Then the above display implies that $\max_H |S(\Zb;\hbeta(\Zb))-S(\Zb;\hbeta(\Zb_H))|=O_{P}(\sqrt{sT_{0}^{-1}\log T_{0}})$. As in Appendix \ref{sec: pertubation stability low dim}, we can use Theorem \ref{thm: approx exchange} to obtain the asymptotic size control. 
Assume $s=o(T_{0}/\log(T_{0}))$. Let $q_T=(sT_0^{-1}\log(T_0))^{-1/4}$. Notice that $q_T\rightarrow \infty$. Since  $q_T\rightarrow \infty$ and $|H|\lesssim \log(T_0)$, the above result of  $\max_H |S(\Zb;\hbeta(\Zb))-S(\Zb;\hbeta(\Zb_H))|=O_{P}(\sqrt{sT_{0}^{-1}\log T_{0}})$ implies that there exists $\gamma_{1,T}=o(1)$ such that  Assumption \ref{assu: stability} holds with  $\varrho_T(x)=q_T\sqrt{sT_0^{-1}x}$. Then by the same argument as in Appendix \ref{sec: pertubation stability low dim}, Theorem \ref{thm: approx exchange} implies
$$
 \left|P\left(\hat{p}\leq\alpha\right)-\alpha\right|
\leq C_{1}\sqrt{K_1 \varrho_{T}(T_{0}/R+2k)}+o(1),
$$
where $K_1$ is the same constant as in Appendix \ref{sec: pertubation stability low dim}. Since $k\leq T_0/R$, $R\asymp T_0/\log(T_0)$ and $q_T \rightarrow \infty$, we have
$$
\varrho_{T}(T_{0}/R+2k)\leq \varrho_{T}(3T_{0}/R) = \sqrt{3} q_T \sqrt{sT_0^{-1}(T_0/R)} \lesssim q_T \sqrt{sT_0^{-1}\log(T_0)} = 1/q_T \rightarrow 0.
$$

\begin{remark}
Stability does not imply that the constrained Lasso residuals, $\heps_{t}$, are close to $\varepsilon_{*,t}=Y_t-X_t'\beta_{*} $, where $\beta_{*}=\arg \min_{\beta\in \Wcal}E(Y_t-X_t'\beta)^2 $ is a pseudo-true value. In \cite{chernozhukov2019ttest}, we show that $\|\hbeta-\beta_* \|_2=O_P((T_0^{-1}\log J)^{1/4}) $. However, this is far from enough to conclude that $|X_t'(\hbeta-\beta_*)|=o_P(1) $ due to the high-dimensionality of $X_t $. The usual Cauchy-Schwarz bound $\|X_t\|_2 \|\hbeta-\beta_* \|_2$ would not converge to zero; the H\"older bound $\|X_t\|_\infty \|\hbeta-\beta_* \|_1$ does not suffices either since  $\|\hbeta-\beta_* \|_1 $ does not converge to zero. [When $\Wcal$ is a bounded $\ell_1$-ball, it is in fact impossible to achieve  $\|\hbeta-\beta_* \|_1=o_P(1) $ \citep[][]{ye2010rate}.] Even if  $\beta_*$ is assumed to be sparse, one would still require $\|\beta_*\|_0=o(\sqrt{T_0}/\log J)$. 
In contrast, our stability condition discussed above only requires the much weaker condition of $\|\beta_*\|_0=o(T_0/\log T_0)$. Therefore, stability condition could be satisfied even if the estimated residual does not converge to a pseudo-true target. \qed
\end{remark}

\section{Consistency and Estimator Stability}
\label{app:ridge}
In this section, we illustrate the difference between stability and consistency using a simple and analytically tractable example: Ridge regression. We shall show that under correct specification, Ridge may not be consistent, while still satisfying estimator stability. 

To keep theoretical analysis tractable, we work under stylized conditions. Given data $Z=(Y,X)$, we define the ridge estimator
$$
\hbeta_\lambda(Z)=(X'X+\lambda I_T)^{-1}X'Y,
$$
where $\lambda$ is a tuning parameter. Then we can compute the residuals $\hu(Z;\lambda)=Y-X \hbeta_{\lambda}(Z)$.

\begin{lem}
	\label{lem: ridge stability}Suppose that $Y=X\beta+u$, where  $X=(X_1,\dots,X_T)'\in\RR^{T\times J}$ and $E(u\mid X)=0$. Assume that $\|\beta\|_2$ is bounded away from zero and infinity, $E(uu'\mid X)=\sigma^2 I_T$ for  a constant $\sigma>0$, $J\lesssim  T^{\kappa_{0}}$ with $\kappa_{0}\in(0,2/3)$ and for some constants $\kappa_1,\kappa_2>0$, 
		\[
	P\left( \kappa_1 T \leq \lambda_{\min}(X'X)\leq  \lambda_{\max}(X'X)\leq \kappa_2 T \right)\geq1-o(1) 
	\]
	and 
		\[
	P\left(  \max_{1\leq t\leq T}\|X_t\| \leq \kappa_3 \sqrt{J}\quad {\rm and} \quad \|X'u\|_2\leq \kappa_4 \sqrt{JT} \right)\geq1-o(1). 
	\]
	Let $Z_{H}=(\tilde{Y},\tilde{X})$ be the perturbed data with $|H|\asymp \log T$. For $\tilde{u}=\tilde{Y}-\tilde{X}\beta$, assume that
		\[
	P\left( \|\tilde{X}'\tilde{u}-X'u\|_2\leq \kappa_5 \sqrt{J|H|}  \right)\geq1-o(1) 
	\]
	and 
		\[
	P\left(  \|\tilde{X}'\tilde{X}-X'X\|\leq \kappa_6 (|H|+J) \right)\geq1-o(1),
	\]
	where $\kappa_5,\kappa_6>0$ are constants. 
	  	If $\lambda\asymp T$, then $E \|X (\hbeta_{\lambda}(Z)-\beta)\|_2^2\gtrsim T $ and $\|\hu(Z;\lambda)-\hu(Z_{H};\lambda) \|_{\infty}=o_P(1) $\footnote{Notice that consistency is enough to derive asymptotic size control from Theorem \ref{thm: approx exchange}; see the discussion at the end of Appendix \ref{sec: pertubation stability low dim} for an explicit formula for quantities in Assumption \ref{assu: stability}. }. 
\end{lem}

Lemma \ref{lem: ridge stability} provides a robustness guarantee for the validity of the procedure. Under the ideal choice of the tuning parameter $\lambda$, we would expect consistency of $\hbeta_{\lambda}(Z)$ and hence validity of the procedure.   However, Lemma \ref{lem: ridge stability} states that even when the tuning parameter is badly chosen such that consistency fails, one might still expect the estimator to be stable under perturbations, which is sufficient for the validity of our inference procedure. This is important in practice since computing optimal tuning parameters is often difficult. 

Now we give very simple sufficient (but probably far from necessary) conditions for the assumptions of Lemma \ref{lem: ridge stability}. Suppose that the data is independent across $t$ and rows of $X$ are sub-Gaussian. Then standard results in random matrix theory can be used to verify these assumptions. For example, $X'X/T-E(X'X/T)$ has eigenvalues tending to zero (e.g., Theorem 5.39 and Remark 5.40 in \citet{vershynin2010introduction}). Hence, as long as $E(X'X/T)$ has bounded eigenvalues,  there exist constants $\kappa_1,\kappa_2$ that satisfy the condition in Lemma \ref{lem: ridge stability}. Let us further assume that entries of $X$ are also bounded and entries of $u$ are sub-Gaussian and independent of $X$. Then $\kappa_3$  can be chosen to be a large enough constant. 

The existence of $\kappa_4,\kappa_5$ is a consequence of the Hanson-Wright inequality. We consider $\|X'u\|_2$ conditional on $X$. By Theorem 6.3.2 in \citet{vershynin2018high} (proved using the Hanson-Wright inequality) applied to the conditional probability given $X$ and by the definition of sub-Gaussian variables, we have that with probability one, for any $z>0$,
$$
P \left(|\|X'u\|_2-\sigma \|X\|_F|>z \|X\| \mid X \right)\leq C_1 \exp \left(-C_2 z^2 \right),
$$
where $C_1,C_2>0$ are constants depending only on the sub-Gaussian norm of $u$. Choosing $z=\sqrt{\log(T)}$, we obtain
$$
P \left(\|X'u\|_2\leq \sigma \|X\|_F+ \sqrt{\log(T)} \|X\|  \right)\leq o(1).
$$
Notice that $E\|X\|_F^2=\sum_{j,t}EX_{j,t}^2=O(JT)$ and $\|X\|_F^2-E\|X\|_F^2=\sum_{t}(\sum_{j}(X_{j,t}^2-EX_{j,t}^2))$. By the boundedness of $X_{j,t}$, $\sum_{j}(X_{j,t}^2-EX_{j,t}^2)=O(J)$ has mean zero and variance  bounded by $O(J^2)$. It follows by the central limit theorem that  $\|X\|_F^2-E\|X\|_F^2=\sum_{t}(\sum_{j}(X_{j,t}^2-EX_{j,t}^2))=\sqrt{T}J$. Since $E\|X\|_F^2=O(JT)$, there exists a constant $C_3>0$ such that $P(\|X\|_F^2<C_3 JT)\geq 1-o(1)$. Again by the random matrix theory (e.g., Theorem 5.39 and Remark 5.40 in \citet{vershynin2010introduction}), $\|X\|=O_P(\sqrt{J+T})$. We have that with probability approaching one, $\|X'u\|_2 \leq \sigma \|X\|_F +\sqrt{\log(T)}\|X\|\leq \sqrt{C_3 JT}\sigma+O(\sqrt{(J+T)\log(T)})$. Assume $J\gg \log(T)$. Since $(J+T)\log(T) \ll JT$, there exists $\kappa_4$ satisfying $P(\|X'u\|_2\leq \kappa_4 \sqrt{JT})\geq 1-o(1)$. The same argument holds for $\kappa_5$; we just repeat the same argument with $T$ replaced by $|H|$ once we realize $\tilde{X}'\tilde{u}-X'u=\sum_{t\in H} \tilde{X}_t \tilde{u}_t- \sum_{t\in H} X_t u_t $, where $X_t$ denotes the $t$-th row of $X$. (Analogous notations apply to $\tilde{X}_t,u_t$, etc.)

To find $\kappa_6$, we observe that $\tilde{X}'\tilde{X}-X'X=\sum_{t\in H} \tilde{X}_t \tilde{X}_t' -\sum_{t\in H} X_t X_t' $. By the random matrix theory (e.g., Theorem 5.39 and Remark 5.40 in \citet{vershynin2010introduction}), $\|\sum_{t\in H} (X_t X_t' -EX_t X_t')\|\leq C_4 \sqrt{J |H|}$ with probability approaching one, where $C_4>0$ is a constant.  Again assuming that $EX_t X_t' $ has eigenvalues bounded by a constant $C_5>0$, we have that $\|\sum_{t\in H} X_t X_t' \| \leq C_5 |H|+C_4 \sqrt{J |H|}$ with probability approaching one. Recall the elementary inequality $\sqrt{J|H|}\leq (J+|H|)/2$. Therefore, we can simply set $\kappa_6=4(C_4+C_5)$ and obtain that with probability approaching one, $\|\sum_{t\in H} X_t X_t' \| \leq (|H|+J)\kappa_6/2$. Similarly, we can get a bound for $\|\sum_{t\in H} \tilde{X}_t \tilde{X}_t' \|$.

The above analysis is  based on simple sub-Gaussian or boundedness assumptions on $(X,u)$. However, the purpose is to show that even in the simple case, Ridge regression can exhibit stability without being consistent (due to Lemma \ref{lem: ridge stability}). We expect similar results to hold in more general or complicated data-generating processes. We leave this extension for future research.

\section{Simulation Study}
\label{sec:simulations}
This section presents simulation evidence on the finite sample properties of our inference procedures. We consider the three CSC methods used in the empirical application in Section \ref{sec:application}: difference-in-differences, canonical SC, and constrained Lasso with $K=1$.

We consider different data generating processes (DGPs) for the treated unit all of which specify the treated outcome as a weighted combination of the control outcomes:
\[
Y_{1t}=\begin{cases}
\sum_{j=2}^{J+1}w_jY^N_{jt}+u_{t} & {\rm if}\ t\le T_0, \\
\theta_t+\sum_{j=2}^{J+1}w_jY^N_{jt}+u_{t} & {\rm if} \ t >T_0,
\end{cases}
\]
where $u_{t}=\rho_{u}u_{t-1}+v_t$, $v_t\overset{iid}\sim N(0,1-\rho_u^2)$. Similar to \citet{hahn2017synthetic}, the control outcomes are generated using a factor model:
\[
Y^N_{jt}=\lambda_{1j}+F_{1t}+\lambda_{2j}F_{2t}+\epsilon_{jt},
\]
where $\lambda_{1j}=(j-1)/J$, $\lambda_{2j}=(j-1)/J$, $F_{1t}\overset{iid}\sim N(0,1)$, and $\epsilon_{jt}=\rho_\epsilon\epsilon_{jt-1}+\xi_{jt}$, $\xi_{jt}\overset{iid}\sim N(0,1-\rho_\epsilon^2)$. In the simulations, we vary $\rho_u$, $\rho_\epsilon$, $T_0$, $J$, and $F_{2t}$. The DGPs differ with respect to the specification of the weights $w$.
\begin{table}[ht]
\centering
%\caption{DGPs}
\begin{tabular}{| l | l | l |}
\hline
&\textbf{Weight Specification}&\textbf{Correctly Specified Model(s)}\\
\hline
DGP1 &$w=\left(\frac{1}{J},\dots,\frac{1}{J}\right)'$&Difference-in-differences, SC, constrained Lasso\\
DGP2&$w=\left(\frac{1}{3},\frac{1}{3},\frac{1}{3},0,\dots,0\right)'$&SC, constrained Lasso\\
DGP3&$w=-1\cdot \left(\frac{1}{J},\dots,\frac{1}{J}\right)'$&constrained Lasso\\
DGP4&$w=\left(1,-1,0,\dots,0\right)'$& --\\
\hline
\end{tabular}
\normalsize
\end{table}

We set $T_\ast=1$ and consider the problem of testing the null hypothesis of a zero effect:
\[
H_0: \theta_T=0.
\]
The $p$-values are computed using the set of moving block permutations $\Pi_{\rightarrow}$. The nominal level is $\alpha=0.1$. 

We first analyze the performance of our procedure with stationary data, letting $F_{2t}\overset{iid}\sim N(0,1)$. Table \ref{tab:size_iid} presents simulation evidence on the size properties of our method when the data are iid ($\rho_u=\rho_\epsilon=0$), which implies exchangeability of the residuals (cf.\ Theorem \ref{thm:finite_sample}). As expected, our procedure achieves exact size control, irrespective of whether or not the model for $P_t^N$ is correctly specified.  To study the finite sample performance with dependent data, we set $\rho_u=\rho_\epsilon=0.6$. Table \ref{tab:size_wd} shows that our method exhibits close-to-correct size under correct specification and under misspecification, confirming the theoretical results on the robustness of our procedure under estimator stability and stationarity.

To investigate the performance with non-stationary data, we use trending factors: $F_{2t}\sim N(t,1)$. Tables \ref{tab:size_iid_nonstat} and \ref{tab:size_wd_nonstat} show that under correct specification, our method exhibits excellent size properties. However, unlike in stationary settings, misspecification can cause size distortions. This finding is expected given the theoretical results and discussions in Section \ref{subsec:stability}.

Figure \ref{fig:pcs} displays power curves for a setting where $T_0=19$ and $J=50$ as in our empirical application, $\rho_u=\rho_\epsilon=0.6$, and $F_{2t}\overset{iid}\sim N(0,1)$. Under correct specification, our method exhibits excellent small sample power properties and comes close to achieving the oracle power bound based on the true marginal distribution of $u_t$. Moreover, we find that imposing additional constraints when estimating $P_t^N$ (e.g., using SC instead of the more general constrained Lasso) does not improve power when these additional restrictions are correct but can cause power losses when they are not.

\section{Proofs}

\subsection*{Additional Notation}
We introduce some additional notations that will be used in the proofs.
For $a,b\in\mathbb{R}$, $a\vee b=\max\{a,b\} $.  For two positive sequences $a_{n}, b_{n} $ (indexed by $n$), we use  $a_{n} \ll b_{n} $ to denote $a_n=o(b_n) $.  We use $\Phi(\cdot)$ to denote the cumulative distribution function of the standard normal distribution. Unless stated otherwise, $\|\cdot\|$ denotes the Euclidean norm for vectors or the spectral norm for matrices. We use $\overset{d}{=} $ to denote equal in distribution.

\subsection{Proof of Theorem \ref{thm:approximate_validity}}
The proof proceeds by verifying the high-level conditions in the following lemma. Let $n=|\Pi|$ so that $n=T!$ if $\Pi=\Pi_{\text{all}}$ and $n=T$ if $\Pi=\Pi_{\to}$.
\begin{lem}[Approximate Validity under High-Level Conditions\footnote{In \citet[][]{chernozhukov2018exact}, we use a version of Lemma \ref{lem: approximate validity}, which relies on permuting the data instead of permuting the residuals, to derive performance guarantees for prediction intervals obtained using classical conformal prediction methods with weakly dependent data. The proof of Lemma \ref{lem: approximate validity} (presented in Section \ref{sec:proof_lemma_approximate_validity} to make the exposition self-contained) follows from the same arguments as the proof of Theorem 2 in \citet{chernozhukov2018exact} modified to the problem of permuting residuals.}]\label{lem: approximate validity}
Let $\{\delta_{1n}, \delta_{2n}, \gamma_{1n}, \gamma_{2n}\}$ be sequences of numbers converging to zero. Assume the following conditions.
\begin{itemize} 
\item [(E)] With
probability $1-\gamma_{1n}$:  the randomization distribution 
$$\tilde{F}(x):=\frac{1}{n}\sum_{\pi\in\Pi}\mathbf{1}\{S(u_\pi)< x\},$$
is \textit{approximately ergodic} for $F(x)=P\left(S(u)< x\right)$, namely 
$$
\sup_{x\in\mathbb{R}}\left|\tilde{F}(x)-F\left(x\right)\right|\leq\delta_{1n},
$$
\end{itemize}
\begin{itemize} 
\item [(A)] With
probability $1-\gamma_{2n}$, estimation errors are small:
\begin{enumerate}\setlength{\itemsep}{0pt} \setlength{\parskip}{0pt}
\item the mean squared error is small, $n^{-1}\sum_{\pi\in\Pi}\left[S(\hat{u}_\pi)-S(u_\pi)\right]^{2}\leq\delta_{2n}^{2};$ 
\item the pointwise error at $\pi=\mathrm{Identity}$ is small, $|S(\hat{u})-S(u)|\leq\delta_{2n}$; 
\item The pdf of $S(u)$ is bounded above by a constant $D$. 
\end{enumerate}
\end{itemize}
Suppose in addition that the null hypothesis is true. Then, the approximate p-value obeys for any $\alpha\in(0,1)$ 
\[
\left|P\left(\hat{p}\leq\alpha\right)- \alpha \right|  \leq 4\delta_{1n}+ 4 \delta_{2n}+ 2 D(\delta_{2n}+ 2\sqrt{\delta_{2n}}) +\gamma_{1n}+\gamma_{2n}.
\]

\end{lem}

With this result at hand, the proof of the theorem is a consequence of following four lemmas, which
verify the approximate ergodicity conditions (E) and 
conditions on the estimation error (A) of Lemma \ref{lem: approximate validity}. 
Putting the bounds together and optimizing the error yields the result of the theorem.

The following lemma verifies approximate ergodicity (E)  (which  allows for large $T_{*}$) for the case of moving block permutations. 

\begin{lem}[Approximate Ergodicity under Moving Block Permutations]
\label{lem: primitive cond for ergodicity}Let $\Pi$ be the moving block
permutations. Suppose that $\{u_{t}\}_{t=1}^{T}$ is stationary and
strong mixing. Assume the following conditions:
(1) $\sum_{k=1}^{\infty}\alpha_{mixing}(k)$ is bounded by a constant
$M$,  (2) $T_{0}\geq T_{*}+2$, and (3) $S(u)$ has bounded pdf. Then there exists a constant $M'>0$ depending only on $M$ such that
for any $\delta_{1n}>0$, 
\[
P\left(\sup_{x\in\mathbb{R}}\left|\tilde{F}(x)-F(x)\right|\leq\delta_{1n}\right)\geq1-\gamma_{T},
\]
where $\gamma_{T}=\left(M'\sqrt{\frac{T_{*}}{T_{0}}}\log T_{0}+\frac{T_{*}+1}{T_{0}+T_{*}}\right)/\delta_{1n}$. 
\end{lem}

The following lemma verifies approximate ergodicity (E)  (which  allows for large $T_{*}$) for the case of iid permutations.

\begin{lem}[Approximate Ergodicity under iid Permutations]
\label{lem: primitive cond for ergodicity under iid}Let $\Pi$ be
the set of all permutations. Suppose that $\{u_{t}\}_{t=1}^{T}$ is
iid. Assume that $S(u)$ only depends on the last $T_{*}$ entries
of $u$. If $T_{0}\geq T_{*}+2$, then
\[
P\left(\sup_{x\in\mathbb{R}}\left|\tilde{F}(x)-F(x)\right|\leq\delta_{1n}\right)\geq1-\gamma_{T},
\]
where $\gamma_{T}=\sqrt{\pi/(2\left\lfloor T/T_{*}\right\rfloor) }/\delta_{1n}$. 
\end{lem}

The following lemma  verifies the condition on the estimation error (A) for moving block permutations.

\begin{lem}[Bounds on Estimation Errors under Moving Block Permutations]
\label{lem: lipschitz S}  Consider
moving block permutations $\Pi$. Let $T_{*}$ be fixed. Suppose that for
some constant $Q>0$, $|S(u)-S(v)|\leq Q\|D_{T_{*}}(u-v)\|_{2}$
for any $u,v\in\mathbb{R}^{T}$ and $D_{T_{*}}:={\rm Blockdiag}(0_{T_{*}},I_{T_{*}})$.
Then Condition (A) (1)-(2) is satisfied if there exist sequences $\gamma_{T},\delta_{2n}=o(1)$
such that with probability at least $1-\gamma_{T}$, $$ \|\hat{P}^N-P^N\|_{2}/\sqrt{T}\leq\delta_{2n}
\text{ and } |\hat{P}^N_{t}-P_{t}|\leq\delta_{2n}  \text{  for  } T_{0}+1\leq t\leq T.$$ 

\end{lem}

The following lemma  verifies the condition on the estimation error (A) for iid permutations.

\begin{lem}[Bounds on Estimation Errors under iid Permutations]
\label{lem: lipschitz S all perm} Consider
the set of all permutations $\Pi$. Let $T_{*}$ be fixed. Suppose
that for some constant $Q>0$, $|S(u)-S(v)|\leq Q\|D_{T_{*}}(u-v)\|_{2}$
for any $u,v\in\mathbb{R}^{T}$ and $D_{T_{*}}:={\rm Blockdiag}(0,I_{T_{*}})$. Then Condition (A) (1)-(2)
is satisfied if there exist sequences $\gamma_{T},\delta_{2n}=o(1)$
such that with probability at least $1-\gamma_{T}$, $$ \|\hat{P}^N-P^N\|_{2}/\sqrt{T}\leq\delta_{2n}
\text{ and } |\hat{P}^N_{t}-P_{t}|\leq\delta_{2n}  \text{  for  } T_{0}+1\leq t\leq T.$$ 
\end{lem}

Now we conclude the proof of Theorem \ref{thm:approximate_validity}.

For the moving block permutations, let $\delta_{1n}=(T_*/T_0)^{1/4} $. Then we apply Lemma \ref{lem: approximate validity} together with Lemmas  \ref{lem: primitive cond for ergodicity} and \ref{lem: lipschitz S}, obtaining
\begin{align*}
\left|P\left(\hat{p}\leq\alpha\right)- \alpha \right| &  \leq  4\delta_{1n}+ 4 \delta_{2n}+ 2 D(\delta_{2n}+ 2\sqrt{\delta_{2n}}) +\gamma_{1n}+\gamma_{2n} \\
& \leq 4\delta_{1n}+ 4 \delta_{2n}+ 2 D(\delta_{2n}+ 2\sqrt{\delta_{2n}}) + \left(M'\sqrt{\frac{T_{*}}{T_{0}}}\log T_{0}+\frac{T_{*}+1}{T_{0}+T_{*}}\right)/\delta_{1n}+ \gamma_{2n} \\
& \leq 4(T_*/T_0)^{1/4} + 4 \delta_{2n}+ 2 D(\delta_{2n}+ 2\sqrt{\delta_{2n}}) \\
& \qquad + \left(M'\sqrt{\frac{T_{*}}{T_{0}}}\log T_{0}+\frac{T_{*}+1}{T_{0}+T_{*}}\right)(T_*/T_0)^{-1/4} + \gamma_{2n}.
\end{align*}
The final result for moving block permutations follows by straight-forward computations and the observations that $ \delta_{2n} =O(\sqrt{\delta_{2n}})  $ (due to $ \delta_{2n}=o(1) $).

For iid permutations, we also use $\delta_{1n}=(T_*/T_0)^{1/4} $. Then we apply Lemma \ref{lem: approximate validity}  together with Lemmas \ref{lem: primitive cond for ergodicity under iid} and \ref{lem: lipschitz S all perm}, 
 obtaining
\begin{align*}
\left|P\left(\hat{p}\leq\alpha\right)- \alpha \right| &  \leq   4\delta_{1n}+ 4 \delta_{2n}+ 2 D(\delta_{2n}+ 2\sqrt{\delta_{2n}}) +\gamma_{1n}+\gamma_{2n} \\
& \leq  4\delta_{1n}+ 4 \delta_{2n}+ 2 D(\delta_{2n}+ 2\sqrt{\delta_{2n}}) + \sqrt{2\pi/\left\lfloor T/T_{*}\right\rfloor }/\delta_{1n}+ \gamma_{2n} \\
& \leq  4(T_*/T_0)^{1/4}  + 4 \delta_{2n}+ 2 D(\delta_{2n}+ 2\sqrt{\delta_{2n}}) +\sqrt{2\pi/\left\lfloor T/T_{*}\right\rfloor } (T_*/T_0)^{-1/4} + \gamma_{2n}\\
& \lesssim (T_*/T_0)^{1/4} +\delta_{2n}+ \sqrt{\delta_{2n}}+ \gamma_{2n}.
\end{align*}
This completes the proof for iid permutations.

{}

\subsubsection{Proof of Lemma \ref{lem: approximate validity}}
\label{sec:proof_lemma_approximate_validity}
The proof proceeds in two steps.\footnote{The proof follows from the same arguments as the proof of Theorem 2 in \citet{chernozhukov2018exact} modified to the problem of permuting residuals and is presented for completeness.} 

\textbf{Step 1:} We bound the difference between the $p$-value and the oracle $p$-value, $\hat{F}(S(\hat{u}))-F(S(u))$.

Let $\mathcal{M}$ be the event that the conditions (A) and (E) hold. By assumption, 
\begin{equation}
P\left(\mathcal{M}\right)\geq1-\gamma_{1n}-\gamma_{2n}.\label{eq: thm high level eq 0.5}
\end{equation}

Notice that on the event $\mathcal{M}$, 
\begin{align}
\left|\hat{F}(S(\hat{u}))-F(S(u))\right| & \leq\left|\hat{F}(S(\hat{u}))-F(S(\hat{u}))\right|+\left|F(S(\hat{u}))-F(S(u))\right|\nonumber \\
 & \overset{\text{ (i)}}{\leq}\sup_{x\in\mathbb{R}}\left|\hat{F}(x)-F(x)\right|+D\left|S(\hat{u})-S(u)\right|\nonumber \\
 & \leq\sup_{x\in\mathbb{R}}\left|\hat{F}(x)-\tilde{F}(x)\right|+\sup_{x\in\mathbb{R}}\left|\tilde{F}(x)-F(x)\right|+D\left|S(\hat{u})-S(u)\right|\nonumber \\
 & \leq\sup_{x\in\mathbb{R}}\left|\hat{F}(x)-\tilde{F}(x)\right|+\delta_{1n}+D\left|S(\hat{u})-S(u)\right|\nonumber \\
 & \leq\sup_{x\in\mathbb{R}}\left|\hat{F}(x)-\tilde{F}(x)\right|+\delta_{1n}+D\delta_{2n},\label{eq: thm high level eq 1}
\end{align}
where (i) holds by the fact that the bounded pdf of $S(u)$ implies
Lipschitz property for $F$. 

Let $A=\left\{ \pi\in\Pi:\ |S(\hat{u}_\pi)-S(u_\pi)|\geq\sqrt{\delta_{2n}}\right\} $.
Observe that on the event $\mathcal{M}$, by Chebyshev inequality
\[
|A|\delta_{2n}\leq\sum_{\pi\in\Pi}\left(S(\hat{u}_\pi)-S(u_\pi)\right)^{2}\leq n\delta_{2n}^{2}
\]
and thus $|A|/n\leq\delta_{2n}$. Also observe that on the event $\mathcal{M}$,
for any $x\in\mathbb{R}$, 
\begin{align}
 & \left|\hat{F}(x)-\tilde{F}(x)\right| \nonumber \\
% & =\left|\frac{1}{n}\sum_{\pi\in\Pi}\left(\mathbf{1}\left\{ S(\hat{u}_\pi)\leq x\right\} -\mathbf{1}\left\{ S(u_\pi)\leq x\right\} \right)\right|\\
 & \leq\frac{1}{n}\sum_{\pi\in A}\left|\mathbf{1}\left\{ S(\hat{u}_\pi)< x\right\} -\mathbf{1}\left\{ S(u_\pi)< x \right\} \right|+\frac{1}{n}\sum_{\pi\in(\Pi\backslash A)}\left|\mathbf{1}\left\{ S(\hat{u}_\pi)< x\right\} -\mathbf{1}\left\{ S(u_\pi)< x\right\} \right|\nonumber \\
 & \overset{\mathrm{(i)}}{\leq}\frac{|A|}{n}+\frac{1}{n}\sum_{\pi\in(\Pi\backslash A)}\mathbf{1}\left\{ \left|S(u_\pi)-x\right|\leq\sqrt{\delta_{2n}}\right\}  \leq\frac{|A|}{n}+\frac{1}{n}\sum_{\pi\in\Pi}\mathbf{1}\left\{ \left|S(u_\pi)-x\right|\leq\sqrt{\delta_{2n}}\right\} \nonumber \\
 & \leq\frac{|A|}{n}+P\left(\left|S(u)-x\right|\leq\sqrt{\delta_{2n}}\right)\nonumber \\
 &\qquad +\sup_{z\in\mathbb{R}}\left|\frac{1}{n}\sum_{\pi\in\Pi}\mathbf{1}\left\{ \left|S(u_\pi)-z\right|\leq\sqrt{\delta_{2n}}\right\} -P\left(\left|S(u)-z\right|\leq\sqrt{\delta_{2n}}\right)\right| \nonumber \\
 & =\frac{|A|}{n}+P\left(\left|S(u)-x\right|\leq\sqrt{\delta_{2n}}\right) \nonumber \\
 & \qquad+\sup_{x\in\mathbb{R}}\left|\left[\tilde{F}\left(z+\sqrt{\delta_{2n}}\right)-\tilde{F}\left(z-\sqrt{\delta_{2n}}\right)\right]-\left[F\left(z+\sqrt{\delta_{2n}}\right)-F\left(z-\sqrt{\delta_{2n}}\right)\right]\right| \nonumber \\
 & \leq\frac{|A|}{n}+P\left(\left|S(u)-x\right|\leq\sqrt{\delta_{2n}}\right) +2\sup_{z\in\mathbb{R}}\left|\tilde{F}(z)-F\left(z\right)\right|\nonumber \\
 & \overset{\mathrm{(ii)}}{\leq}\frac{|A|}{n}+2D\sqrt{\delta_{2n}}+2\delta_{1n} \overset{\mathrm{(iii)}}{\leq }\delta_{1n}+ 2 \delta_{2n}+ 2 D\sqrt{\delta_{2n}}, \label{eq: bound ecdf error}
\end{align}
where (i) follows by 
the elementary inequality of $|\mathbf{1}\{S(\hat{u}_\pi)< x\}-\mathbf{1}\{S(u_\pi)<x\}|\leq\mathbf{1}\{|S(u_\pi)-x|\leq|S(\hat{u}_\pi)-S(u_\pi)|\}$,
(ii) follows by the bounded pdf of $S(u)$ and (iii) follows by $|A|/n\leq\delta_{2n}$.
Since the above display holds for each $x\in\mathbb{R}$, it follows
that on the event $\mathcal{M}$, 
\begin{equation}
\sup_{x\in\mathbb{R}}\left|\hat{F}(x)-\tilde{F}(x)\right|\leq
\delta_{1n}+ 2 \delta_{2n}+ 2D\sqrt{\delta_{2n}}.\label{eq: thm high level eq 2}
\end{equation}

We combine (\ref{eq: thm high level eq 1}) and (\ref{eq: thm high level eq 2})
and obtain that on the event $\mathcal{M}$, 
\begin{equation}
\left|\hat{F}(S(\hat{u}))-F(S(u))\right|\leq
2\delta_{1n}+ 2 \delta_{2n}+ D(\delta_{2n}+2 \sqrt{\delta_{2n}}).\label{eq: thm high level eq 3}
\end{equation}

\textbf{Step 2:} Here we derive the desired result. Notice that 
\begin{align*}
& \left|P\left(1-\hat{F}(S(\hat{u}))\leq\alpha\right)-\alpha\right|\\
 &  =\left|E\left(\mathbf{1}\left\{ 1-\hat{F}(S(\hat{u}))\leq\alpha\right\} -\mathbf{1}\left\{ 1-F(S(u))\leq\alpha\right\} \right)\right|\\
 & \leq E\left|\mathbf{1}\left\{ 1-\hat{F}(S(\hat{u}))\leq\alpha\right\} -\mathbf{1}\left\{ 1-F(S(u))\leq\alpha\right\} \right|\\
 & \overset{\mathrm{(i)}}{\leq}P\left(\left|F(S(u))-1+\alpha\right|\leq\left|\hat{F}(S(\hat{u}))-F(S(u))\right|\right)\\
 & \leq P\left(\left|F(S(u))-1+\alpha\right|\leq\left|\hat{F}(S(\hat{u}))-F(S(u))\right|\ \text{ and}\ \mathcal{M}\right)+P(\mathcal{M}^{c})\\
 & \overset{\mathrm{(ii)}}{\leq}P\left(\left|F(S(u))-1+\alpha\right|\leq
  2\delta_{1n}+ 2 \delta_{2n}+ D(\delta_{2n}+ 2\sqrt{\delta_{2n}})\right)+P\left(\mathcal{M}^{c}\right)\\
 & \overset{\mathrm{(iii)}}{\leq}
4\delta_{1n}+ 4 \delta_{2n}+ 2 D(\delta_{2n}+ 2\sqrt{\delta_{2n}}) +\gamma_{1n}+\gamma_{2n},
\end{align*}
where (i) follows by the elementary inequality $|\mathbf{1}\{1-\hat{F}(S(\hat{u}))\leq\alpha\}-\mathbf{1}\{1-F(S(u))\leq\alpha\}|\leq\mathbf{1}\{|F(S(u))-1+\alpha|\leq|\hat{F}(S(\hat{u}))-F(S(u))|\}$,
(ii) follows by (\ref{eq: thm high level eq 3}), (iii) follows
by the fact that $F(S(u))$ has the uniform distribution on $(0,1)$
and hence has pdf equal to 1, and  by (\ref{eq: thm high level eq 0.5}).
The proof is complete.

\subsubsection{Proof of Lemma \ref{lem: primitive cond for ergodicity}}
We define 
\[
s_{t}=\begin{cases}
(\sum_{s=t}^{t+T_{*}-1}|u_{s}|^q)^{1/q} & {\rm if}\ 1\leq t\leq T_{0}\\
(\sum_{s=t}^{T}|u_{s}|^q+\sum_{s=1}^{t-T_{0}-1}|u_{s}|^q)^{1/q} & {\rm otherwise}.
\end{cases}
\]

It is straight-forward to verify that 
\[
\{S(u_\pi):\ \pi\in\Pi\}=\left\{ s_{t}:\ 1\leq t\leq T\right\} .
\]

Let $\tilde{\alpha}_{\mathrm{mixing}}$ be the strong-{mixing} coefficient for
$\{s_{t}\}_{t=1}^{T_{0}}$. Notice that $\{s_{t}\}_{t=1}^{T_{0}}$
is stationary (although $\{s_{t}\}_{t=1}^{T}$ is clearly not). Let $\check{F}(x)=T_{0}^{-1}\sum_{t=1}^{T_{0}}\mathbf{1}\{s_{t}\leq x\}$.
The bounded pdf of $S(u)$ implies the continuity of $F(\cdot)$.
It follows, by Proposition 7.1 of \citet{rio2017asymptotic}, that
\begin{equation}
E\left(\sup_{x\in\mathbb{R}}\left|\check{F}(x)-F(x)\right|^{2}\right)\leq\frac{1}{T_{0}}\left(1+4\sum_{k=0}^{T_{0}-1}\tilde{\alpha}_{\mathrm{mixing}}(t)\right)\left(3+\frac{\log T_{0}}{2\log2}\right)^{2}.\label{eq: ergodicity eq 1}
\end{equation}

Notice that $\tilde{\alpha}_{\mathrm{mixing}}(t)\leq2$ and that 
$
\tilde{\alpha}_{\mathrm{mixing}}(t)\leq\alpha_{\mathrm{mixing}}\left(\max\{t-T_{*},0\}\right)
$
so that
\begin{eqnarray*}
\sum_{k=0}^{T_{0}-1}\tilde{\alpha}_{\mathrm{mixing}}(t)& = & \sum_{k=0}^{T_{*}}\tilde{\alpha}_{\mathrm{mixing}}(t)+\sum_{k=T_{*}+1}^{T_{0}-1}\tilde{\alpha}_{\mathrm{mixing}}(t)\leq2(T_{*}+1)+\sum_{k=1}^{T_{0}-T_{*}-1}\alpha_{\mathrm{mixing}}(k)\\
& \leq & 2(T_{*}+1)+\sum_{k=1}^{\infty}\alpha_{\mathrm{mixing}}(k).
\end{eqnarray*}
Since $\sum_{k=1}^{\infty}\alpha_{\mathrm{mixing}}(k)$ is bounded by $M$, it follows
by (\ref{eq: ergodicity eq 1}) that 
\[
E\left(\sup_{x\in\mathbb{R}}\left|\check{F}(x)-F(x)\right|^{2}\right)\leq B_T: = \frac{ 1+ 4 (2 (T_{*} +1) + M) }{T_{0}}\left(3+\frac{\log T_{0}}{2\log2}\right)^{2}.\]
By Liapunov's inequality, 
\[
E\left(\sup_{x\in\mathbb{R}}\left|\check{F}(x)-F(x)\right|\right)\leq\sqrt{E\left(\sup_{x\in\mathbb{R}}\left|\check{F}(x)-F(x)\right|^{2}\right)}\leq \sqrt{B_T}.
\]

Since $(T_{0}+T_{*})\tilde{F}(x)-T_{0}\check{F}(x)=\sum_{t=T_{0}+1}^{T_{0}+T_{*}}\mathbf{1}\{s_{t}\leq x\}$,
it follows that 
\begin{align*}
\sup_{x\in\mathbb{R}}\left|\tilde{F}(x)-\check{F}(x)\right| & =\sup_{x\in\mathbb{R}}\left|\left(\frac{T_{0}}{T_{0}+T_{*}}\check{F}(x)+\frac{1}{T_{0}+T_{*}}\sum_{t=T_{0}+1}^{T_{0}+T_{*}}\mathbf{1}\{s_{t}\leq x\}\right)-\check{F}(x)\right|\\
 & =\sup_{x\in\mathbb{R}}\left|\frac{1}{T_{0}+T_{*}}\check{F}(x)+\frac{1}{T_{0}+T_{*}}\sum_{t=T_{0}+1}^{T_{0}+T_{*}}\mathbf{1}\{s_{t}\leq x\}\right| \overset{}{\leq}\frac{T_{*}+1}{T_{0}+T_{*}},
\end{align*}
where the last inequality follows by $\sup_{x\in\mathbb{R}}|\check{F}(x)|\leq1$
and the boundedness of the indicator function. Combining the above
two displays, we obtain that
\[
E\left(\sup_{x\in\mathbb{R}}\left|\tilde{F}(x)-F(x)\right|\right)\leq \sqrt{B_T} +\frac{T_{*}+1}{T_{0}+T_{*}}.
\]
The desired result follows by Markov's inequality.
%\[
%P\left(\sup_{x\in\mathbb{R}}\left|\tilde{F}(x)-F(x)\right|>\delta_{n}\right)\leq\frac{E\left(\sup_{x\in\mathbb{R}}\left|%\tilde{F}(x)-F(x)\right|\right)}{\delta_{n}}\leq \delta_n^{-1}
%\left( \sqrt{B_T} +\frac{T_{*}+1}{T_{0}+T_{*}} \right)
%\]

%\yinchu{In order to get the optimal rate, we just need to get rid of $\log T_0 $ in (\ref{eq: ergodicity eq 1}). In other words, we need to get a better maximal inequality than Prop 7.1 of Rio (2017). It seems that we can just prove a better one by adjusting the arguments in the proof of Prop 7.1 and 7.2 in Rio (2017). Alternatively, we can apply Theorem 2 in (Doukhan, Massart and Rio 1995, Invariance principles for absolutely regular empirical processes) to the class of indicator functions. Or similarly, we can apply Proposition 6.2 of (Dedecker, Jérôme, and Sana Louhichi. "Maximal inequalities and empirical central limit theorems." Empirical process techniques for dependent data (2002): 137-159.). All of these options require us to a separate lemma on the maximal inequality. Do you guys know some off-the-shelf results that we can just cite?    }

\subsubsection{Proof of Lemma \ref{lem: primitive cond for ergodicity under iid}}

The proof follows by an argument given by \citet{romano2012uniform} for subsampling.
We give a complete argument for our setting here for clarity and completeness.

Recall that $\Pi$ is the set of all bijections $\pi$ on $\{1,...,T\}$.
Let $k_{T}=\left\lfloor T/T_{*}\right\rfloor $. Define the blocks of indices $$b_{i}=(T-iT_{*}+1,T-iT_{*}+2,...,T-iT_{*}+T_{*})\in\mathbb{R}^{T_{*}},  \quad \quad   i = 1,...., k_{T} $$
%The notation  $\pi(b_i)$ is ill-defined/ambigious, since $\pi$ maps $[T]$ To $[T]$.

Since $S(u)$ only depends on $u_{b_{1}}$,  the last $T_{*}$ entries of $u$,
we can define $$Q(x;u_{b_{1}})=\mathbf{1}\{S(u)\leq x\}-F(x).$$
Therefore, 
\[
\tilde{F}(x)-F(x)=\frac{1}{|\Pi|}\sum_{\pi\in\Pi}Q(u_{\pi(b_{1})};x).
\]
Define  $\pi(b_i):=\pi_{\mid b_i} (b_i)$ to mean the restriction of the permutation map $\pi:
\{1, \dots, T\} \to \{1 , \dots, T\}$ to the domain $b_i$.

Notice that for $1\leq i\leq k_{T}$, the value of $\sum_{\pi\in\Pi}Q(u_{\pi(b_{i})};x)$
does not depend on $i$. It follows that 
\begin{align*}
\tilde{F}(x)-F(x)=\frac{1}{|\Pi|}\sum_{\pi\in\Pi}Q(u_{\pi(b_{1})};x) & =\frac{1}{k_{T}}\sum_{i=1}^{k_{T}}\left(\frac{1}{|\Pi|}\sum_{\pi\in\Pi}Q(u_{\pi(b_{i})};x)\right)\\
 & =\frac{1}{|\Pi|}\sum_{\pi\in\Pi}\left[\frac{1}{k_{T}}\sum_{i=1}^{k_{T}}Q(u_{\pi(b_{i})};x)\right].
\end{align*}

Hence by Jensen's inequality
\begin{align*}
E\left(\sup_{x\in\mathbb{R}}\left|\tilde{F}(x)-F(x)\right|\right) 
 & \leq\frac{1}{|\Pi|}\sum_{\pi\in\Pi}E\left(\sup_{x\in\mathbb{R}}\left|\frac{1}{k_{T}}\sum_{i=1}^{k_{T}}Q(u_{\pi(b_{i})};x)\right|\right).
\end{align*}

To compute the above expectation,
we observe that for any $\pi\in\Pi$,
\begin{align*}
E\left(\sup_{x\in\mathbb{R}}\left|\frac{1}{k_{T}}\sum_{i=1}^{k_{T}}Q(u_{\pi(b_{i})};x)\right|\right) & =\int_{0}^{1}P\left(\sup_{x\in\mathbb{R}}\left|\frac{1}{k_{T}}\sum_{i=1}^{k_{T}}Q(u_{\pi(b_{i})};x)\right|>z\right)dz\\
 & \leq \int_{0}^{1}2\exp\left(-2k_{T}z^{2}\right)dz  < \int_{0}^{\infty}2\exp\left(-2k_{T}z^{2}\right)dz = \sqrt{\pi/(2k_{T})},
\end{align*}
where the first inequality follows by the Dvoretsky-Kiefer-Wolfwitz inequality (e.g.,
Theorem 11.6 in \citet{kosorok2007introduction}) and the fact that
for any $\pi\in\Pi$, $\{Q(u_{\pi(b_{i})};x)\}_{i=1}^{k_{T}}$ is a
sequence of iid random variables (since $\pi$ is a bijection and
$\{b_{i}\}_{i=1}^{k_{T}}$ are disjoint blocks of indices); the last equality follows
from the properties of the normal density. Therefore, the above two
display imply that 
\[
E\left(\sup_{x\in\mathbb{R}}\left|\tilde{F}(x)-F(x)\right|\right)\leq \sqrt{\pi/(2k_{T})}.
\]

The desired result follows by Markov's inequality.

\subsubsection{Proof of Lemma \ref{lem: lipschitz S}}
Due to the Lipschitz property of $S(\cdot)$, we have 
\begin{align*}
\sum_{\pi\in\Pi}\left[S(\hat{u}_{\pi})-S(u_{\pi})\right]^{2} & \leq Q\sum_{\pi\in\Pi}\left\Vert D_{T_{*}}(\hat{u}_{\pi}-u_{\pi})\right\Vert _{2}^{2}  =Q\sum_{\pi\in\Pi}\sum_{t=T_{0}+1}^{T_{0}+T_{*}}\left(\hat{u}_{\pi(t)}-u_{\pi(t)}\right)^{2}\\
 & =Q\sum_{t=T_{0}+1}^{T_{0}+T_{*}}\sum_{\pi\in\Pi}\left(\hat{u}_{\pi(t)}-u_{\pi(t)}\right)^{2} {=}QT_{*}\|\hat{u}-u\|_{2}^{2} = QT_* \| \hat P^N - P^N\|^2
\end{align*}
where the penultimate equality follows by the observation that for moving block permutation
$\Pi$, $$\sum_{\pi\in\Pi}\left(\hat{u}_{\pi(t)}-u_{\pi(t)}\right)^{2}=\|\hat{u}-u\|_{2}^{2}.$$
Hence condition (A) (1)  follows with a rescaled value of $\delta_n$. Condition
(A) (2) holds by the Lipschitz property of $S(\cdot)$: 
\[
|S(\hat{u})-S(u)|\leq Q\|D_{T_{*}}(\hat{u}-u)\|_{2}\leq Q\sqrt{\sum_{t=T_{0}+1}^{T_{0}+T_{*}}(\hat{u}_{t}-u_{t})^{2}} 
\]
Hence, Condition (A) (2) follows
since $\| \hat P^N_t - P^N_t\| = |\hat{u}_{t}-u_{t}|\leq\delta_{n}$ for $T_{0}+1\leq t\leq T$
with high probability. The proof is complete.

\subsubsection{Proof of Lemma \ref{lem: lipschitz S all perm}}
For $t,s\in\{1,...,T\}$, we define $A_{t,s}=\{\pi\in\Pi:\ \pi(t)=s\}$.
Recall that $\Pi$ is the set of all bijections on $\{1,...,T\}$.
Thus, $|A_{t,s}|=(T-1)!$. It follows that for any $t\in\{1,...,T\}$,
\begin{align}
\sum_{\pi\in\Pi}\left(\hat{u}_{\pi(t)}-u_{\pi(t)}\right)^{2} & =\sum_{s=1}^{T}\sum_{\pi\in A_{t,s}}\left(\hat{u}_{\pi(t)}-u_{\pi(t)}\right)^{2}\nonumber \\
 & =\sum_{s=1}^{T}\sum_{\pi\in A_{t,s}}\left(\hat{u}_{s}-u_{s}\right)^{2}=\sum_{s=1}^{T}|A_{t,s}|\left(\hat{u}_{s}-u_{s}\right)^{2}=(T-1)!\times\|\hat{u}-u\|_{2}^{2}.\label{eq: lipschitz iid eq 1}
\end{align}

Due to the Lipschitz property of $S(\cdot)$, we have that
\begin{eqnarray*}
 && \frac{1}{|\Pi|}\sum_{\pi\in\Pi}\left[S(\hat{u}_{\pi})-S(u_{\pi})\right]^{2}
\leq  \frac{Q}{|\Pi|}\sum_{\pi\in\Pi}\left\Vert D_{T_{*}} (\hat{u}_{\pi}-u_{\pi})\right\Vert _{2}^{2}
  =\frac{Q}{|\Pi|}\sum_{\pi\in\Pi}\sum_{t=T_{0}+1}^{T_{0}+T_{*}}\left(\hat{u}_{\pi(t)}-u_{\pi(t)}\right)^{2}\\
 && = \frac{Q}{|\Pi|}\sum_{t=T_{0}+1}^{T_{0}+T_{*}}\sum_{\pi\in\Pi}\left(\hat{u}_{\pi(t)}-u_{\pi(t)}\right)^{2}
\overset{}{=}\frac{Q}{|\Pi|}T_{*}(T-1)!\times\|\hat{u}-u\|_{2}^{2}
\overset{}{=}QT^{-1}T_{*}\|\hat{u}-u\|_{2}^{2},
\end{eqnarray*}
where the penultimate equality follows by (\ref{eq: lipschitz iid eq 1}) and 
the last equality  follows by $|\Pi|=T!$. Thus, part 1 of Condition (A) follows since
$T_{*}$ is fixed.

To see part 2 of Condition (A), notice that the Lipschitz property
of $S(\cdot)$ implies 
\[
|S(\hat{u})-S(u)|\leq Q\|D_{T_{*}}(\hat{u}-u)\|_{2}\leq Q\sqrt{\sum_{t=T_{0}+1}^{T_{0}+T_{*}}(\hat{u}_{t}-u_{t})^{2}}.
\]
Hence, part 2 of Condition (A) follows since $|\hat{u}_{t}-u_{t}|\leq\delta_{n}$
for $T_{0}+1\leq t\leq T$ with high probability. The proof is complete.

\subsection{Proof of Theorem \ref{thm: approx exchange}}

We first state auxiliary results.
\begin{lem}
	\label{lem: easy bnd emp}Let $\{W_{t}\}_{t=1}^{T}$ be a stationary
	and $\beta$-mixing sequence with coefficient $\betamix(\cdot)$.
	Let $G(x)=P(W_{t}\leq x)$. 
	Then for any positive integer $1\leq m\leq T/2$, we
	have 
	\[
		E\left(\sup_{x\in\mathbb{R}}\left|T^{-1}\sum_{t=1}^{T}[\oneb\{W_{t}\leq x\}-G(x)]\right|\right)\leq  2\sqrt{T}\betamix(m)+ \sqrt{\pi m/(2T)}+(m-1)/T.
	\]
\end{lem}

The next lemma was derived by \citet[][Lemma 2.1]{berbee1987convergence}; see Theorem 16.2.1 in \citet{athreya2006measure} and Lemma 7.1 of \citet{chen2016self} for popular versions. We state the result to make the exposition more self-contained. The proof is omitted. %Following the suggestion of a referee, we state the result here without proof so that our argument is more self-contained and clear. 
\begin{lem}
	\label{lem: Berbee} Let $(W,R)$ be random vectors defined in the same probability space. Let  $P_{W,R}$ denote the probability distribution of $(W,R)$. Let $P_R$ and $P_W$ denote the probability distributions of $R$ and $W$, respectively. Let $\| \cdot\|_{TV}$ denote the total variation metric. Define the $\beta$-mixing coefficient $\beta(W,R)=\|P_{W,R}-P_W \otimes P_R\|_{TV}/2$. Then the probability space can be extended to construct random element $\tW$ such that (1) $\tW$ and $R$ are independent, (2) $\tW$ and $W$ have the same distribution and (3) $P(\tW \neq W)\leq \beta(W,R)$.
\end{lem}

% \begin{lem}
% 	\label{lem: basic comparison bound}Let $X$ and $Y$ be random variables. Suppose that the pdf of $X$ is bounded by $\kappa$. Then for any $c>0$, we have 
% 	\[
% \sup_{z\in\mathbb{R}}|P(X\leq z)-P(Y\leq z)|  \leq c\kappa+P(|X-Y|>c) .
% 	\]
% \end{lem}

\bigskip

Now we prove Theorem \ref{thm: approx exchange}. In this proof, universal constants refer to constants that depend
only on $D_{1},D_{2},D_{3}>0$. Define $\tF(x)=R^{-1}\sum_{j=1}^{R}\tF_{j}(x)$,
where
\[
\tF_{j}(x)=m^{-1}\sum_{t\in H_{j}}\oneb\left\{ \phi\left(Z_{t},...,Z_{t+T_{*}-1};\hat{\beta}(\Zb)\right)\leq x\right\} .
\]

We note that under  $\Pi=\Pi_{\to}$, $\hat{F}(x)$ can be written as
\begin{equation}
\hF(x)=T^{-1}\left(\sum_{t=1}^{T_{0}}\oneb\left\{ \phi\left(Z_{t},...,Z_{t+T_{*}-1};\hat{\beta}(\Zb)\right)\leq x\right\} +\sum_{t=T_{0}+1}^{T_{0}+T_{*}}\oneb\left\{ \phi\left(Z_{q(t)},...,Z_{q(t+T_{*}-1)};\hat{\beta}(\Zb)\right)\leq x\right\} \right),\label{eq: stability thm eq 4}
\end{equation}
where $q(t)=t\oneb\{t\leq T\}+(t-T)\oneb\{t> T\}$.

The rest of the proof proceeds in 4 steps. The first three steps bound
$\sup_{x\in\RR}|\hF(x)-\Psi\left(x;\hbeta(\Zb_{\tH_{R}})\right)|$, where we recall that $\Psi(x;\beta)=P(\phi\left(Z_{t},...,Z_{t+T_{*}-1};\beta\right)\leq x)$.
The fourth step derives the desired result. 

\textbf{Step 1:} bound $\sup_{x\in\mathbb{R}}\left|\tF_{j}(x)-\Psi\left(x;\hat{\beta}(\Zb_{\tH_{j}})\right)\right|$. 

Let $A_{j}=\bigcup_{t\in H_{j}}\{t,...,t+T_{*}-1\}$. Since $k>T_{*}$,
we have that $A_{j}\subset\tH_{j}$ and $\min_{t\in A_{j},\ s\in\tH_{j}^{c}}|t-s|\geq k-T_{*}+1$.
This means that $\{Z_{t}\}_{t\in\tH_{j}^{c}}$ and $\{Z_{t}\}_{t\in A_{j}}$
have a gap of at least $k-T_{*}+1$ time periods. By Lemma \ref{lem: Berbee} (applied with $W=\{Z_{t}\}_{t\in\tH_{j}^{c}}$ and $R=\{Z_{t}\}_{t\in A_{j}}$), there exist random elements
$\{\bZ_{t}\}_{t\in A_{j}}$ (on an enlarged probability space) such
that (1) $\{\bZ_{t}\}_{t\in A_{j}}$ is independent of $\{Z_{t}\}_{t\in\tH_{j}^{c}}$,
(2) $\{\bZ_{t}\}_{t\in A_{j}}\overset{d}{=}\{Z_{t}\}_{t\in A_{j}}$
and (3) $P(\{\bZ_{t}\}_{t\in A_{j}}\neq\{Z_{t}\}_{t\in A_{j}})\leq\betamix(k-T_{*}+1)$.
Since $\{\tilde{Z}_{t}\}_{t\in\tH_{j}}$ is independent of the data,
we can construct $\{\bZ_{t}\}_{t\in A_{j}}$ such that it is also
independent of $\Zb_{\tH_{j}}$. 

Define the event 
\begin{multline*}
\Mcal_{j}=\left\{ \{\bZ_{t}\}_{t\in A_{j}}=\{Z_{t}\}_{t\in A_{j}}\right\} \bigcap\left\{ \sup_{x\in\RR}\left|\partial\Psi\left(x;\hbeta(\Zb_{\tH_{j}})\right)/\partial x\right|\leq\xi_{T}\right\} \\
\bigcap\left\{ \max_{\pi\in\Pi}\left|S\left(\Zb^{\pi},\hat{\beta}(\Zb)\right)-S\left(\Zb^{\pi},\hat{\beta}(\Zb_{\tH_{j}})\right)\right|\leq\varrho_{T}(|\tH_{j}|)\right\} 
\end{multline*}
 as well as the functions 
\[
\begin{cases}
\cF_{j}(x)=m^{-1}\sum_{t\in H_{j}}\oneb\left\{ \phi\left(\bZ_{t},...,\bZ_{t+T_{*}-1};\hat{\beta}(\Zb_{\tH_{j}})\right)\leq x\right\} \\
\dot{F}_{j}(x)=m^{-1}\sum_{t\in H_{j}}\oneb\left\{ \phi\left(Z_{t},...,Z_{t+T_{*}-1};\hat{\beta}(\Zb_{\tH_{j}})\right)\leq x\right\} .
\end{cases}
\]

By the construction of $\{\bZ_{t}\}_{t\in A_{j}}$ and Assumptions
\ref{assu: stability} and \ref{assu: regularity resid}, $P(\Mcal_{j}^{c})\leq\betamix(k-T_{*}+1)+\gamma_{1,T}+\gamma_{2,T}$. 

Notice that conditional on $\hbeta(\Zb_{\tH_{j}})$, $\phi\left(\bZ_{t},...,\bZ_{t+T_{*}-1};\hat{\beta}(\Zb_{\tH_{j}})\right)$ is a stationary $\beta$-mixing across $t\in H_{j}$ with mixing coefficient $\tilde{\beta}_{{\rm mixing}}(i)\leq\betamix(i-T_{*}+1)$
for $i\geq T_{*}$. Moreover, since  $(\bZ_{t},...,\bZ_{t+T_{*}-1})$ is independent of $\Zb_{\tH_{j}}$, we have
\[
	\Psi\left(x;\hat{\beta}(\Zb_{\tH_{j}})\right)=P\left(\phi\left(\bZ_{t},...,\bZ_{t+T_{*}-1};\hat{\beta}(\Zb_{\tH_{j}})\right) \leq x \mid \hat{\beta}(\Zb_{\tH_{j}})  \right).
\]
  Hence, by Lemma \ref{lem: easy bnd emp}, we have
that for any $m_{1}\leq m/2,$
\[
E\left(\sup_{x\in\mathbb{R}}\left|\cF_{j}(x)-\Psi\left(x;\hat{\beta}(\Zb_{\tH_{j}})\right)\right|\right)\leq2m^{1/2}\betamix(m_{1}-T_{*}+1)+\sqrt{\pi m_{1}/(2m)}+(m_{1}-1)/m.
\]

We shall choose $m_{1}$ later. Observe that on the event $\Mcal_{j}$,
$\cF_{j}(\cdot)=\dot{F}_{j}(\cdot)$. Therefore, 
\begin{equation}
E(a_{j})\leq2m^{1/2}\betamix(m_{1}-T_{*}+1)+\sqrt{\pi m_{1}/(2m)}+(m_{1}-1)/m+2P(\Mcal_{j}^{c}),\label{eq: stability thm eq 5}
\end{equation}
where $a_{j}=\sup_{x\in\mathbb{R}}\left|\dot{F}_{j}(x)-\Psi\left(x;\hat{\beta}(\Zb_{\tH_{j}})\right)\right|$. 

Now we bound $\sup_{x\in\RR}|\dot{F}_{j}(x)-\tF_{j}(x)|$. Fix an
arbitrary $x\in\RR$. Observe that on the event $\Mcal_{j}$, 
\begin{align*}
& \left|\tF_{j}(x)-\dot{F}_{j}(x)\right|\\
 & =\left|m^{-1}\sum_{t\in H_{j}}\left(\oneb\left\{ \phi\left(Z_{t},...,Z_{t+T_{*}-1};\hat{\beta}(\Zb)\right)\leq x\right\} -\oneb\left\{ \phi\left(Z_{t},...,Z_{t+T_{*}-1};\hat{\beta}(\Zb_{\tH_{j}})\right)\leq x\right\} \right)\right|\\
 & \leq m^{-1}\sum_{t\in H_{j}}\left|\oneb\left\{ \phi\left(Z_{t},...,Z_{t+T_{*}-1};\hat{\beta}(\Zb)\right)\leq x\right\} -\oneb\left\{ \phi\left(Z_{t},...,Z_{t+T_{*}-1};\hat{\beta}(\Zb_{\tH_{j}})\right)\leq x\right\} \right|\\
 & \overset{\text{(i)}}{\leq}m^{-1}\sum_{t\in H_{j}}\oneb\left\{ \left|\phi\left(Z_{t},...,Z_{t+T_{*}-1};\hat{\beta}(\Zb_{\tH_{j}})\right)-x\right|\leq\varrho_{T}(|\tH_{j}|)\right\} \\
 & = m^{-1}\sum_{t\in H_{j}}\oneb\left\{ \phi\left(Z_{t},...,Z_{t+T_{*}-1};\hat{\beta}(\Zb_{\tH_{j}})\right)\leq x+\varrho_{T}(|\tH_{j}|)\right\} \\
 & \qquad -m^{-1}\sum_{t\in H_{j}}\oneb\left\{ \phi\left(Z_{t},...,Z_{t+T_{*}-1};\hat{\beta}(\Zb_{\tH_{j}})\right)<x -\varrho_{T}(|\tH_{j}|)\right\} \\
 & <\dot{F}_{j}\left(x+\varrho_{T}(|\tH_{j}|)\right)-\dot{F}_{j}\left(x-2\varrho_{T}(|\tH_{j}|)\right)\\
 & \le\Psi\left(x+\varrho_{T}(|\tH_{j}|);\hat{\beta}(\Zb_{\tH_{j}})\right)-\Psi\left(x-2\varrho_{T}(|\tH_{j}|);\hat{\beta}(\Zb_{\tH_{j}})\right)+2a_{j}\overset{\text{(ii)}}{\leq}3\xi_{T}\varrho_{T}(|\tH_{j}|)+2a_{j},
\end{align*}
where (i) follows by the elementary inequality $|\oneb\{x\leq z\}-\oneb\{y\leq z\}|\leq\oneb\{|y-z|\leq|x-y|\}$
for any $x,y,z\in\RR$ and (ii) follows by the definition of $\Mcal_{j}$.
Since the above bound holds for any $x\in\RR$ and $|\tH_{j}|\leq m+2k$,
we have that on the event $\Mcal_{j}$, 
\[
\sup_{x\in\RR}\left|\tF_{j}(x)-\dot{F}_{j}(x)\right|\leq3\xi_{T}\varrho_{T}(m+2k)+2a_{j}.
\]

By the definition of $a_{j}$, this means that on the event $\Mcal_{j}$,
\[
\sup_{x\in\mathbb{R}}\left|\tF_{j}(x)-\Psi\left(x;\hat{\beta}(\Zb_{\tH_{j}})\right)\right|\leq3\xi_{T}\varrho_{T}(m+2k)+3a_{j}.
\]

By (\ref{eq: stability thm eq 5}) and the fact that $\tF_{j}(\cdot)$
and $\Psi(\cdot,\cdot)$ take values in $[0,1]$, we have that for
a universal constant $C_{1}>0$, 

\begin{align}
 & E\left(\sup_{x\in\mathbb{R}}\left|\tF_{j}(x)-\Psi\left(x;\hat{\beta}(\Zb_{\tH_{j}})\right)\right|\right)\label{eq: stability thm eq 6}\\
 & \leq3\xi_{T}\varrho_{T}(m+2k)+3E(a_{j})+2P(\Mcal_{j}^{c})\nonumber \\
 & \leq3\xi_{T}\varrho_{T}(m+2k)+6m^{1/2}\betamix(m_{1}-T_{*}+1)+3\sqrt{\pi m_{1}/(2m)}+3(m_{1}-1)/m+8P(\Mcal_{j}^{c})\nonumber \\
 & \overset{\text{(i)}}{\leq}C_{1}\left(\xi_{T}\varrho_{T}(m+2k)+m^{1/2}\betamix(m_{1}-T_{*}+1)+\sqrt{m_{1}/m}+\betamix(k-T_{*}+1)+\gamma_{1,T}+\gamma_{2,T}\right),\nonumber 
\end{align}
where (i) follows by $P(\Mcal_{j}^{c})\leq\betamix(k-T_{*}+1)+\gamma_{1,T}+\gamma_{2,T}$
and $m_{1}/m\leq \sqrt{m_{1}/m}$. 

\textbf{Step 2:} bound $R^{-1}\sum_{j=1}^{R}\sup_{x\in\RR}\left|\Psi\left(x;\hat{\beta}(\Zb_{\tH_{j}})\right)-\Psi\left(x;\hat{\beta}(\Zb_{\tH_{R}})\right)\right|$. 

Let $\dot{\Zb}=\{\dot{Z}_{t}\}_{t=1}^{T}$ satisfy that $\dot{\Zb}\overset{d}{=}\Zb$
and $\dot{\Zb}$ is independent of $(\Zb,\{\tilde{Z}_{t}\}_{t=1}^{T})$.
Therefore, for any $1\leq j\leq R$, 
\begin{align*}
\Psi\left(x;\hat{\beta}(\Zb_{\tH_{j}})\right) & =P\left(\phi\left(\dot{Z}_{T_{0}+1},...,\dot{Z}_{T_{0}+T_{*}};\hbeta(\Zb_{\tH_{j}})\right)\leq x\mid\hbeta(\Zb_{\tH_{j}})\right)\\
 & \overset{\text{(i)}}{=}P\left(\phi\left(\dot{Z}_{T_{0}+1},...,\dot{Z}_{T_{0}+T_{*}};\hbeta(\Zb_{\tH_{j}})\right)\leq x\mid\Zb,\{\tilde{Z}_{t}\}_{t=1}^{T}\right),
\end{align*}
where (i) follows by the fact that $\dot{\Zb}$ and $(\Zb,\{\tilde{Z}_{t}\}_{t=1}^{T})$
are independent. (This is the identity that $E(f(X;g(Y))\mid g(Y))=E(f(X;g(Y))\mid Y)$
for any measurable functions $f$ and $g$ if $X$ and $Y$ are independent.
To see this, simply notice that the distribution of $X$ given $g(Y)$
and the distribution of $X$ given $Y$ are both equal to the unconditional
distribution of $X$.)

Define the event 
\[
\Qcal_{j}=\left\{ \sup_{x\in\RR}\left|\partial\Psi\left(x;\hbeta(\Zb_{\tH_{j}})\right)/\partial x\right|\leq\xi_{T}\right\} .
\]

Clearly, $P(\Qcal_{j})\geq1-\gamma_{2,T}$ by Assumption \ref{assu: regularity resid}.
Therefore, we have that on the event $\Qcal_{j}$, 
\begin{align*}
 & \left|\Psi\left(x;\hat{\beta}(\Zb_{\tH_{j}})\right)-\Psi\left(x;\hat{\beta}(\Zb_{\tH_{R}})\right)\right|\\
 & =\left|E\left(\oneb\left\{ \phi\left(\dot{Z}_{T_{0}+1},...,\dot{Z}_{T_{0}+T_{*}};\hbeta(\Zb_{\tH_{j}})\right)\leq x\right\} -\oneb\left\{ \phi\left(\dot{Z}_{T_{0}+1},...,\dot{Z}_{T_{0}+T_{*}};\hbeta(\Zb_{\tH_{R}})\right)\leq x\right\} \mid\Zb,\{\tilde{Z}_{t}\}_{t=1}^{T}\right)\right|\\
 & \leq E\left(\left|\oneb\left\{ \phi\left(\dot{Z}_{T_{0}+1},...,\dot{Z}_{T_{0}+T_{*}};\hbeta(\Zb_{\tH_{j}})\right)\leq x\right\} -\oneb\left\{ \phi\left(\dot{Z}_{T_{0}+1},...,\dot{Z}_{T_{0}+T_{*}};\hbeta(\Zb_{\tH_{R}})\right)\leq x\right\} \right|\mid\Zb,\{\tilde{Z}_{t}\}_{t=1}^{T}\right)\\
 & \overset{\text{(i)}}{\leq}E\Biggl[\oneb\biggl\{\left|\phi\left(\dot{Z}_{T_{0}+1},...,\dot{Z}_{T_{0}+T_{*}};\hbeta(\Zb_{\tH_{j}})\right)-x\right|\\
 & \qquad\qquad\qquad\leq\left|\phi\left(\dot{Z}_{T_{0}+1},...,\dot{Z}_{T_{0}+T_{*}};\hbeta(\Zb_{\tH_{R}})\right)-\phi\left(\dot{Z}_{T_{0}+1},...,\dot{Z}_{T_{0}+T_{*}};\hbeta(\Zb_{\tH_{j}})\right)\right|\biggr\}\mid\Zb,\{\tilde{Z}_{t}\}_{t=1}^{T}\Biggr]\\
 & =P\Biggl[\left|\phi\left(\dot{Z}_{T_{0}+1},...,\dot{Z}_{T_{0}+T_{*}};\hbeta(\Zb_{\tH_{j}})\right)-x\right|\\
 & \qquad\qquad\qquad\leq\left|\phi\left(\dot{Z}_{T_{0}+1},...,\dot{Z}_{T_{0}+T_{*}};\hbeta(\Zb_{\tH_{R}})\right)-\phi\left(\dot{Z}_{T_{0}+1},...,\dot{Z}_{T_{0}+T_{*}};\hbeta(\Zb_{\tH_{j}})\right)\right|\mid\Zb,\{\tilde{Z}_{t}\}_{t=1}^{T}\Biggr]\\
 & \leq P\left[\left|\phi\left(\dot{Z}_{T_{0}+1},...,\dot{Z}_{T_{0}+T_{*}};\hbeta(\Zb_{\tH_{j}})\right)-x\right|\leq2\varrho_{T}(|\tH_{j}|)\mid\Zb,\{\tilde{Z}_{t}\}_{t=1}^{T}\right]\\
 & \quad+P\left[\left|\phi\left(\dot{Z}_{T_{0}+1},...,\dot{Z}_{T_{0}+T_{*}};\hbeta(\Zb_{\tH_{R}})\right)-\phi\left(\dot{Z}_{T_{0}+1},...,\dot{Z}_{T_{0}+T_{*}};\hbeta(\Zb_{\tH_{j}})\right)\right|>2\varrho_{T}(|\tH_{j}|)\mid\Zb,\{\tilde{Z}_{t}\}_{t=1}^{T}\right]\\
 & =\Psi\left(x+2\varrho_{T}(|\tH_{j}|);\hbeta(\Zb_{\tH_{j}})\right)-\Psi\left(x-2\varrho_{T}(|\tH_{j}|);\hbeta(\Zb_{\tH_{j}})\right)\\
 & \quad+P\left[\left|\phi\left(\dot{Z}_{T_{0}+1},...,\dot{Z}_{T_{0}+T_{*}};\hbeta(\Zb_{\tH_{R}})\right)-\phi\left(\dot{Z}_{T_{0}+1},...,\dot{Z}_{T_{0}+T_{*}};\hbeta(\Zb_{\tH_{j}})\right)\right|>2\varrho_{T}(|\tH_{j}|)\mid\Zb,\{\tilde{Z}_{t}\}_{t=1}^{T}\right]\\
 & \overset{\text{(ii)}}{\leq}4\xi_{T}\varrho_{T}(m+2k)\\
 & \quad+P\left[\left|\phi\left(\dot{Z}_{T_{0}+1},...,\dot{Z}_{T_{0}+T_{*}};\hbeta(\Zb_{\tH_{R}})\right)-\phi\left(\dot{Z}_{T_{0}+1},...,\dot{Z}_{T_{0}+T_{*}};\hbeta(\Zb_{\tH_{j}})\right)\right|>2\varrho_{T}(|\tH_{j}|)\mid\Zb,\{\tilde{Z}_{t}\}_{t=1}^{T}\right],
\end{align*}
where (i) follows by the elementary inequality $|\oneb\{x\leq z\}-\oneb\{y\leq z\}|\leq\oneb\{|y-z|\leq|x-y|\}$
for any $x,y,z\in\RR$ and (ii) follows by $|\tH_{j}|\leq m+2k$ and
the definition of $\Qcal_{j}$. Since the above bound does not depend
on $x$, it holds uniformly in $x\in\RR$ on the event $\Qcal_{j}$.
Since $\Psi(\cdot,\cdot)$ is also bounded by one, we have that 
\begin{align*}
 & E\left(\sup_{x\in\RR}\left|\Psi\left(x;\hat{\beta}(\Zb_{\tH_{j}})\right)-\Psi\left(x;\hat{\beta}(\Zb_{\tH_{R}})\right)\right|\right)\\
 & \leq4\xi_{T}\varrho_{T}(m+2k)\\
 & \qquad+2P(\Qcal_{j}^{c})+P\left[\left|\phi\left(\dot{Z}_{T_{0}+1},...,\dot{Z}_{T_{0}+T_{*}};\hbeta(\Zb_{\tH_{R}})\right)-\phi\left(\dot{Z}_{T_{0}+1},...,\dot{Z}_{T_{0}+T_{*}};\hbeta(\Zb_{\tH_{j}})\right)\right|>2\varrho_{T}(|\tH_{j}|)\right]\\
 & \leq4\xi_{T}\varrho_{T}(m+2k)+2P(\Qcal_{j}^{c})\\
 & \qquad+P\left[\left|\phi\left(\dot{Z}_{T_{0}+1},...,\dot{Z}_{T_{0}+T_{*}};\hbeta(\Zb_{\tH_{R}})\right)-\phi\left(\dot{Z}_{T_{0}+1},...,\dot{Z}_{T_{0}+T_{*}};\hbeta(\Zb)\right)\right|>\varrho_{T}(|\tH_{j}|)\right]\\
 & \qquad+P\left[\left|\phi\left(\dot{Z}_{T_{0}+1},...,\dot{Z}_{T_{0}+T_{*}};\hbeta(\Zb_{\tH_{j}})\right)-\phi\left(\dot{Z}_{T_{0}+1},...,\dot{Z}_{T_{0}+T_{*}};\hbeta(\Zb)\right)\right|>\varrho_{T}(|\tH_{j}|)\right]\\
 & \overset{\text{(i)}}{\leq}4\xi_{T}\varrho_{T}(m+2k)+2P(\Qcal_{j}^{c})+2\gamma_{1,T}\overset{\text{(ii)}}{\leq}4\xi_{T}\varrho_{T}(m+2k)+2\gamma_{1,T}+2\gamma_{2,T},
\end{align*}
where (i) follows by Assumption \ref{assu: stability} and the fact
that $|\tH_{j}|=|\tH_{R}|$ and (ii) follows by $P(\Qcal_{j})\geq1-\gamma_{2,T}$.
Since the above bound holds for all $1\leq j\leq R$,  we have
\begin{equation}
E\left(R^{-1}\sum_{j=1}^{R}\sup_{x\in\RR}\left|\Psi\left(x;\hat{\beta}(\Zb_{\tH_{j}})\right)-\Psi\left(x;\hat{\beta}(\Zb_{\tH_{R}})\right)\right|\right)\leq4\xi_{T}\varrho_{T}(m+2k)+2\gamma_{1,T}+2\gamma_{2,T}.\label{eq: stability thm eq 7}
\end{equation}

\textbf{Step 3:} bound $\sup_{x\in\RR}\left|\hF(x)-\Psi\left(x;\hbeta(\Zb_{\tH_{R}})\right)\right|$.

By (\ref{eq: stability thm eq 4}), we notice that 
\begin{multline*}
\sup_{x\in\RR}\left|T\hF(x)-m\sum_{j=1}^{R}\tF_{j}(x)\right|=\sup_{x\in\RR}\Biggl|\sum_{t=T_{0}-mR+1}^{T_{0}}\oneb\left\{ \phi\left(Z_{t},...,Z_{t+T_{*}-1};\hat{\beta}(\Zb)\right)\leq x\right\} \\
+\sum_{t=T_{0}+1}^{T_{0}+T_{*}}\oneb\left\{ \phi\left(Z_{q(t)},...,Z_{q(t+T_{*}-1)};\hat{\beta}(\Zb)\right)\leq x\right\} \Biggr|\leq T_{*}+(T_{0}-mR)\leq T_{*}+R-1.
\end{multline*}

Moreover, by (\ref{eq: stability thm eq 6}) and (\ref{eq: stability thm eq 7}),
we have that 
\begin{multline*}
E\left(\sup_{x\in\RR}\Biggl|R^{-1}\sum_{j=1}^{R}\tF_{j}(x)-\Psi\left(x;\hat{\beta}(\Zb_{\tH_{R}})\right)\right)\Biggr|\\
\leq C_{2}\left(\xi_{T}\varrho_{T}(m+2k)+m^{1/2}\betamix(m_{1}-T_{*}+1)+\sqrt{m_{1}/m}+\betamix(k-T_{*}+1)+\gamma_{1,T}+\gamma_{2,T}\right)
\end{multline*}
for some universal constant $C_{2}>0$. 

The above two displays imply that
\begin{multline*}
\sup_{x\in\RR}\left|\frac{T}{mR}\hF(x)-\Psi\left(x;\hat{\beta}(\Zb_{\tH_{R}})\right)\right|\leq\frac{T_{*}+R-1}{mR}\\
+C_{2}\left(\xi_{T}\varrho_{T}(m+2k)+m^{1/2}\betamix(m_{1}-T_{*}+1)+\sqrt{m_{1}/m}+\betamix(k-T_{*}+1)+\gamma_{1,T}+\gamma_{2,T}\right).
\end{multline*}

Since $\hF(x)\in[0,1]$, we have that 
\[
\sup_{x\in\RR}|(1-T/(mR))\hF(x)|\leq\frac{T}{mR}-1\leq\frac{T-mR}{mR}\leq\frac{T_{*}+R-1}{mR}.
\]

Since $mR\geq T_{0}/2$ (due to $R<T_{0}/2$), we have $(T_{*}+R-1)/(mR)\leq2T_{*}T_{0}^{-1}+m^{-1}\lesssim\sqrt{m_{1}/m}$.
Hence, the above two displays imply that for some universal constant
$C_{3}>0$, 
\begin{multline}
E\left(\sup_{x\in\RR}\left|\hF(x)-\Psi\left(x;\hat{\beta}(\Zb_{\tH_{R}})\right)\right|\right)\\
\leq C_{3}\left(\xi_{T}\varrho_{T}(m+2k)+m^{1/2}\betamix(m_{1}-T_{*}+1)+\sqrt{m_{1}/m}+\betamix(k-T_{*}+1)+\gamma_{1,T}+\gamma_{2,T}\right).\label{eq: stability thm eq 9}
\end{multline}

\textbf{Step 4:} derive the desired result. 

Let $A_{R}$ be defined as in Step 1 with $j=R$. Following Step 1,
we can construct random elements $\{\bZ_{t}\}_{t\in A_{R}}$ (on an
enlarged probability space) such that (1) $\{\bZ_{t}\}_{t\in A_{R}}$
is independent of $\Zb_{\tH_{R}}$, (2) $\{\bZ_{t}\}_{t\in A_{R}}\overset{d}{=}\{Z_{t}\}_{t\in A_{R}}$
and (3) $P(\{\bZ_{t}\}_{t\in A_{R}}\neq\{Z_{t}\}_{t\in A_{R}})\leq\betamix(k-T_{*}+1)$.

Define $\bG(\Zb_{\tH_{R}})=\phi\left(\bZ_{T_{0}+1},...,\bZ_{T_{0}+T_{*}};\hbeta(\Zb_{\tH_{R}})\right)$.
Since $\{T_{0}+1,...,T_{0}+T_{*}\}\subset A_{R}$, we have that $(\bZ_{T_{0}+1},...,\bZ_{T_{0}+T_{*}})$
is independent of $\Zb_{\tH_{R}}$, which means that 
\[
P\left(\bG(\Zb_{\tH_{R}})\leq x\mid\hbeta(\Zb_{\tH_{R}})\right)=\Psi\left(x;\hbeta(\Zb_{\tH_{R}})\right)\qquad\forall x\in\RR.
\]

Therefore, 
\begin{equation}
\text{conditional on}\ \hbeta(\Zb_{\tH_{R}}),\ \Psi\left(\bG(\Zb_{\tH_{R}});\hbeta(\Zb_{\tH_{R}})\right)\ \text{has uniform distribution on}\ (0,1).\label{eq: stability thm eq 10}
\end{equation}

We also introduce the following notations to simplify the argument:\\
$\bG(\Zb)=\phi\left(\bZ_{T_{0}+1},...,\bZ_{T_{0}+T_{*}};\hbeta(\Zb)\right)$
and $G(\Zb)=\phi\left(Z_{T_{0}+1},...,Z_{T_{0}+T_{*}};\hbeta(\Zb)\right)$. 

For arbitrary $\alpha\in(0,1)$ and $c>0$, we observe that 
\begin{align*}
 & \left|P\left(\hF\left(G(\Zb)\right)<\alpha\mid\hbeta(\Zb_{\tH_{R}})\right)-\alpha\right|\\
 & \overset{\text{(i)}}{=}\left|E\left(\oneb\left\{ \hF(G(\Zb))<\alpha\right\} \mid\hbeta(\Zb_{\tH_{R}})\right)-E\left(\oneb\left\{ \Psi\left(\bG(\Zb_{\tH_{R}});\hbeta(\Zb_{\tH_{R}})\right)<\alpha\right\} \mid\hbeta(\Zb_{\tH_{R}})\right)\right|\\
 & \leq E\left(\left|\oneb\left\{ \hF(G(\Zb))<\alpha\right\} -\oneb\left\{ \Psi\left(\bG(\Zb_{\tH_{R}});\hbeta(\Zb_{\tH_{R}})\right)<\alpha\right\} \right|\mid\hbeta(\Zb_{\tH_{R}})\right)\\
 & \overset{\text{(ii)}}{\leq}E\left(\oneb\left\{ \left|\Psi\left(\bG(\Zb_{\tH_{R}});\hbeta(\Zb_{\tH_{R}})\right)-\alpha\right|\leq\left|\Psi\left(\bG(\Zb_{\tH_{R}});\hbeta(\Zb_{\tH_{R}})\right)-\hF(G(\Zb))\right|\right\} \mid\hbeta(\Zb_{\tH_{R}})\right)\\
 & =P\left(\left|\Psi\left(\bG(\Zb_{\tH_{R}});\hbeta(\Zb_{\tH_{R}})\right)-\alpha\right|\leq\left|\Psi\left(\bG(\Zb_{\tH_{R}});\hbeta(\Zb_{\tH_{R}})\right)-\hF(G(\Zb))\right|\mid\hbeta(\Zb_{\tH_{R}})\right)\\
 & \leq P\left(\left|\Psi\left(\bG(\Zb_{\tH_{R}});\hbeta(\Zb_{\tH_{R}})\right)-\alpha\right|\leq c\mid\hbeta(\Zb_{\tH_{R}})\right)\\
 &\qquad+P\left(\left|\Psi\left(\bG(\Zb_{\tH_{R}});\hbeta(\Zb_{\tH_{R}})\right)-\hF(G(\Zb))\right|>c\mid\hbeta(\Zb_{\tH_{R}})\right)\\
 & \overset{\text{(iii)}}{\leq}2c+P\left(\left|\Psi\left(\bG(\Zb_{\tH_{R}});\hbeta(\Zb_{\tH_{R}})\right)-\hF(G(\Zb))\right|>c\mid\hbeta(\Zb_{\tH_{R}})\right)
\end{align*}
where (i) follows by (\ref{eq: stability thm eq 10}), (ii) follows
by the elementary inequality $|\oneb\{x<z\}-\oneb\{y<z\}|\leq\oneb\{|y-z|\leq|x-y|\}$
for any $x,y,z\in\RR$ and (iii) follows by (\ref{eq: stability thm eq 10}).
Now we take expectation on both sides, obtaining 
\begin{align}
 & E\left|P\left(\hF\left(G(\Zb)\right)<\alpha\mid\hbeta(\Zb_{\tH_{R}})\right)-\alpha\right|\label{eq: stability thm eq 11}\\
 & \leq2c+P\left(\left|\Psi\left(\bG(\Zb_{\tH_{R}});\hbeta(\Zb_{\tH_{R}})\right)-\hF(G(\Zb))\right|>c\right)\nonumber \\
 & \leq2c+P\left(\left|\Psi\left(\bG(\Zb_{\tH_{R}});\hbeta(\Zb_{\tH_{R}})\right)-\hF(\bG(\Zb))\right|>c\right)+P\left(\{\bZ_{t}\}_{t\in A_{R}}\neq\{Z_{t}\}_{t\in A_{R}}\right)\nonumber \\
 & \leq2c+P\left(\left|\Psi\left(\bG(\Zb_{\tH_{R}});\hbeta(\Zb_{\tH_{R}})\right)-\hF(\bG(\Zb))\right|>c\right)+\betamix(k-T_{*}+1)\nonumber \\
 & \leq2c+c^{-1}E\left|\Psi\left(\bG(\Zb_{\tH_{R}});\hbeta(\Zb_{\tH_{R}})\right)-\hF(\bG(\Zb))\right|+\betamix(k-T_{*}+1)\nonumber 
\end{align}

Define the event 
\[
\Acal=\left\{ \sup_{x\in\RR}\left|\partial\Psi\left(x;\hbeta(\Zb_{\tH_{R}})\right)/\partial x\right|\leq\xi_{T}\right\} \bigcap\left\{ \left|\bG(\Zb)-\bG(\Zb_{\tH_{R}})\right| \leq \varrho_{T}(|\tH_{R}|) \right\} .
\]

By Assumptions \ref{assu: stability} and \ref{assu: regularity resid},
$P(\Acal^{c})\leq\gamma_{1,T}+\gamma_{2,T}$. Therefore, 
\begin{align*}
 & E\left|\Psi\left(\bG(\Zb);\hbeta(\Zb_{\tH_{R}})\right)-\Psi\left(\bG(\Zb_{\tH_{R}});\hbeta(\Zb_{\tH_{R}})\right)\right|\\
 & =E\left(\left|\Psi\left(\bG(\Zb);\hbeta(\Zb_{\tH_{R}})\right)-\Psi\left(\bG(\Zb_{\tH_{R}});\hbeta(\Zb_{\tH_{R}})\right)\right|\times\oneb_{\Acal}\right)\\
 & \qquad\qquad\qquad+E\left(\left|\Psi\left(\bG(\Zb);\hbeta(\Zb_{\tH_{R}})\right)-\Psi\left(\bG(\Zb_{\tH_{R}});\hbeta(\Zb_{\tH_{R}})\right)\right|\times\oneb_{\Acal^{c}}\right)\\
 & \overset{\text{(i)}}{\leq}E\left(\left|\Psi\left(\bG(\Zb);\hbeta(\Zb_{\tH_{R}})\right)-\Psi\left(\bG(\Zb_{\tH_{R}});\hbeta(\Zb_{\tH_{R}})\right)\right|\times\oneb_{\Acal}\right)+2P\left(\Acal^{c}\right)\\
 & \leq\xi_{T}\varrho_{T}(|\tH_{R}|)+2P\left(\Acal^{c}\right)\leq\xi_{T}\varrho_{T}(m+2k)+2\gamma_{1,T}+2\gamma_{2,T},
\end{align*}
where (i) follows by the fact that $\Psi(\cdot,\cdot)\in[0,1]$. Hence,
we have that 
\begin{align*}
 & E\left|\Psi\left(\bG(\Zb_{\tH_{R}});\hbeta(\Zb_{\tH_{R}})\right)-\hF(\bG(\Zb))\right|\\
 & \leq E\left|\Psi\left(\bG(\Zb);\hbeta(\Zb_{\tH_{R}})\right)-\hF(\bG(\Zb))\right|+E\left|\Psi\left(\bG(\Zb);\hbeta(\Zb_{\tH_{R}})\right)-\Psi\left(\bG(\Zb_{\tH_{R}});\hbeta(\Zb_{\tH_{R}})\right)\right|\\
 & \leq E\sup_{x\in\RR}\left|\Psi\left(x;\hbeta(\Zb_{\tH_{R}})\right)-\hF(x)\right|+\xi_{T}\varrho_{T}(m+2k)+2\gamma_{1,T}+2\gamma_{2,T}\\
 & \overset{\text{(i)}}{\leq}C_{4}\left(\xi_{T}\varrho_{T}(m+2k)+m^{1/2}\betamix(m_{1}-T_{*}+1)+\sqrt{m_{1}/m}+\betamix(k-T_{*}+1)+\gamma_{1,T}+\gamma_{2,T}\right)
\end{align*}
for a universal constant $C_{4}>0$, where (i) follows by (\ref{eq: stability thm eq 9}).

Now we combine (\ref{eq: stability thm eq 11}) and the above display.
We also choose 
\[
c\asymp\sqrt{\xi_{T}\varrho_{T}(m+2k)+m^{1/2}\betamix(m_{1}-T_{*}+1)+\sqrt{m_{1}/m}+\betamix(k-T_{*}+1)+\gamma_{1,T}+\gamma_{2,T}}.
\]

Then we can find a universal constant $C_{5}>0$ such that 
\begin{multline*}
E\left|P\left(\hF\left(G(\Zb)\right)<\alpha\mid\hbeta(\Zb_{\tH_{R}})\right)-\alpha\right|\\
\leq C_{5}\sqrt{\xi_{T}\varrho_{T}(m+2k)+m^{1/2}\betamix(m_{1}-T_{*}+1)+\sqrt{m_{1}/m}+\betamix(k-T_{*}+1)+\gamma_{1,T}+\gamma_{2,T}}.
\end{multline*}

Now we choose $m_{1}$ satisfying $m_{1}\asymp(\log m)^{1/D_{3}}$
and $m^{1/2}\betamix(m_{1}-T_{*}+1)\lesssim m^{-1}$. Hence, for some
universal constant $C_{6}>0$, 
\begin{multline*}
E\left|P\left(\hF\left(G(\Zb)\right)<\alpha\mid\hbeta(\Zb_{\tH_{R}})\right)-\alpha\right|\\
\leq C_{6}\sqrt{\xi_{T}\varrho_{T}(m+2k)}+C_{6}\left(m^{-1}(\log m)^{1/D_{3}}\right)^{1/4}+C_{6}\sqrt{\betamix(k-T_{*}+1)}+C_{6}\sqrt{\gamma_{1,T}}+C_{6}\sqrt{\gamma_{2,T}}.
\end{multline*}

Since $m\asymp T_{0}/R$, the desired result follows once we notice
that $\hat{p}\geq1-\alpha$ and $\hF\left(G(\Zb)\right)<\alpha$ are
the same event.

\subsubsection{Proof of Lemma \ref{lem: easy bnd emp}}

	Define $K=\left\lfloor T/m\right\rfloor $ and $\hF(x)=m^{-1/2}\sum_{r=1}^{m}\hF_{r}(x)$,
	where  $\hF_{r}(x)=K^{-1/2}\sum_{j=1}^{K}[\oneb\{W_{(j-1)m+r}\leq x\}-G(x)]$ for $1\leq r\leq m$.
	Let $\Delta(x)=\sum_{t=mK+1}^{T}[\oneb\{W_{t}\leq x\}-G(x)]$. Let $L_{T}(x)=T^{-1/2}\sum_{t=1}^{T}[\oneb\{W_{t}\leq x\}-G(x)]$. Notice
	that 
	\[
	\sqrt{T}L_{T}(x)=\sqrt{mK}\hF(x)+\Delta(x).
	\]
	
	Since $|\oneb\{W_{t}\leq x\}-G(x)|\leq1$, it follows that $\sup_{x\in\mathbb{R}}|\Delta(x)|\leq T-mK\leq m-1$
	and thus 
	\begin{equation}
	\sup_{x\in\mathbb{R}}\left|\sqrt{T}L_{T}(x)-\sqrt{mK}\hF(x)\right|\leq m-1.\label{eq: easy bnd emp eq 3}
	\end{equation}
	
	By Berbee's coupling (Lemma \ref{lem: Berbee}), we
	can enlarge the probability space and define random variables $\{\bar{W}_{t}\}_{t=1}^{mK}$
	such that (1) $\bar{W}_{t}\overset{d}{=}W_{t}$ for all $1\leq t\leq mT$,
	(2) $\bar{W}_{(j-1)m+r}$ is independent across $1\leq j\leq K$ for
	all $r$ and (3) $P(\bigcup_{t=1}^{mK}\{\bar{W}_{t}\neq W_{t}\})\leq mK\betamix(m)\leq T\betamix(m)$. 
	
	We now define $\bF(x)=m^{-1/2}\sum_{r=1}^{m}\bF_{r}(x)$, where $\bF_{r}(x)=K^{-1/2}\sum_{j=1}^{K}[\oneb\{\bar{W}_{(j-1)m+r}\leq x\}-G(x)]$. 
	
	Since $\{\bar{W}_{(j-1)m+r}\}_{j=1}^{K}$ is independent, it follows
	by Dvoretzky-Kiefer-Wolfowitz inequality that for any $z>0$, 
	\[
	P\left(\sup_{x\in\mathbb{R}}|\bF_{r}(x)|>z\right)\leq2\exp(-2z^{2}).
	\]
	    Therefore, we have that 
	\[
	E\left(\sup_{x\in\mathbb{R}}|\bF_{r}(x)|\right)=\int_{0}^{\infty}P\left(\sup_{x\in\mathbb{R}}|\bF_{r}(x)|>z\right) dz \leq 2 \int_{0}^{\infty} \exp(-2z^{2}) dz=\sqrt{\pi/2}.
	\]
	Hence, we have that 
	\[
	E\left(\sup_{x\in\mathbb{R}}|\bF(x)|\right) \leq m^{-1/2}\sum_{r=1}^{m} E\left(\sup_{x\in\mathbb{R}}|\bF_{r}(x)|\right) \leq \sqrt{\pi m/2}.
	\]

	Since $\bF(\cdot)=\hF(\cdot)$ with probability at least $1-T\betamix(m)$, we have that 
	\[
	E\left(\sup_{x\in\RR}|\bF(x)-\hF(x)|\right)\leq 2T\betamix(m).
	\]
	Therefore, $E\left(\sup_{x\in\mathbb{R}}|\hF(x)|\right)\leq  2T\betamix(m)+ \sqrt{\pi m/2}$.

	By (\ref{eq: easy bnd emp eq 3}) and $mK/T\leq1$, we have that 
	\[
	E\left(\sup_{x\in\mathbb{R}}|L_{T}(x)|\right)\leq  2T\betamix(m)+ \sqrt{\pi m/2}+(m-1)/\sqrt{T}.
	\]
	
	The proof is complete.

\subsection{Proof of Theorem \ref{thm:approximate_validity CS}}

Recall   $\Zb^{*}=(Z_1^*,\dots,Z_{T}^*)'$ with $ Z_t^*=\left(Y^N_{1t},Y^N_{2t},\dots,Y^N_{J+1t},X'_{1t},\dots,X'_{J+1t}\right)'$ for $1\leq t\leq T_0+T_*$. Let 
$$
\hat{p}_{\Zb^*}=1- \hat{F}\left( S(\hat{u}(\Zb^*))\right),
$$
where $\hat{F}\left( x;\Zb^*\right)=\frac{1}{|\Pi|}\sum_{\pi\in\Pi}\mathbf{1}\left\{ S\left(\hat{u}_\pi(\Zb^*) \right)<  x\right\}$ and $\hat{u}(\Zb^*)=Y^N-\hat{P}^N$ with $\hat{P}^N$ computed using $\Zb^*$.
By the proof of Theorem \ref{thm:approximate_validity}, we have
\[
|P\left(\hat{p}_{\Zb^*}\leq\alpha\right)- \alpha| \leq C ( \tilde \delta_T + \delta_T + \sqrt{\delta_T} + \gamma_T),
\]
where $\tilde \delta_T = (T_*/T_0)^{1/4}(\log T)$ and the constant $C$ depends on $T_*$, $M$ and $D$,  but not on $T$.
It follows that 
\[
|P\left(\hat{p}_{\Zb^*}> \alpha\right)- (1-\alpha) | \leq C ( \tilde \delta_T + \delta_T + \sqrt{\delta_T} + \gamma_T).
\]
The desired result follows by  observing that $\theta_{t}\in \mathcal{C}_{1-\alpha}(t)$ is the same event as $\hat{p}_{\Zb^*}> \alpha$; this is simply because $\Zb(\theta^0)=\Zb^*$, where $\Zb(\theta^0)$ is defined in Section \ref{subsec:hypotheses}.

\subsection{Proof of Theorem \ref{thm:finite_sample}}

Let $\{S^{(j)}(\hat u)\}_{j=1}^{n}$ denoted the non-decreasing rearrangement of $\{S(\hat u_\pi): \pi \in \Pi\}$, where $n= |\Pi|$, which we refer to as randomization quantiles. The $p$-value is 
$$
\hat p = \frac{1}{n} \sum_{\pi \in \Pi} \mathbf{1} ( S(\hat u_\pi) \geq S(\hat u)).
$$
Note that $$
\mathbf{1} (\hat p  \leq \alpha) = \mathbf{1}( S(\hat u ) > S^{(k)}( \hat u) ),
$$
where $k = k(\alpha) =  n - \lfloor n\alpha\rfloor = \lceil n(1-\alpha)\rceil$. 

The proof proceeds in three steps. First, we show that exchangeability of the data implies exchangeability of the residuals. Second, we show that exchangeability of the residuals implies that $P(\hat{p}\le \alpha)\le \alpha$. Third, we show that if there are no ties, $\alpha-1/n\le P(\hat{p}\le \alpha)$. The proof follows from standard arguments \citep[e.g.,][]{hoeffding1952large,romano1990behavior,chernozhukov2018exact,lei2017distributionfree}.

\textbf{Step 1:} By the  iid or exchangeability property of data, we have that 
$$
\underset{\{\hat{u}_{t}\}_{t=1}^{T}}{\underbrace{\{g(Z_{t},\hat{\beta}(\{Z_{t}\}_{t=1}^{T}))\}_{t=1}^{T}}}\overset{d}{=}\{g(Z_{\pi(t)},\hat{\beta}(\{Z_{\pi(t)}\}_{t=1}^{T})\}_{t=1}^{T}.
$$
Since $\hat{\beta}(\{Z_{\pi(t)}\}_{t=1}^{T}) $ does not depend on $\pi $, we have 
$$
\{g(Z_{\pi(t)},\hat{\beta}(\{Z_{\pi(t)}\}_{t=1}^{T})\}_{t=1}^{T}=\underset{\{\hat{u}_{\pi(t)}\}_{t=1}^{T}}{\underbrace{\{g(Z_{\pi(t)},\hat{\beta}(\{Z_{t}\}_{t=1}^{T})\}_{t=1}^{T}}}.
$$
Therefore, $\{\hat{u}_{\pi(t)}\}_{t=1}^{T} \overset{d}{=}\{\hat{u}_{t}\}_{t=1}^{T} $.

\textbf{Step 2:} Note that $\Pi_{\text{all}}$ and $\Pi_{\to}$ form groups in the sense that $\Pi \pi = \Pi$ for all $\pi \in \Pi$. Therefore, the randomization quantiles are invariant, 
$$ S^{(k(\alpha))}(\hat{u}_\pi) = S^{(k(\alpha))}(\hat{u}), \text{ for all } \pi \in \Pi.$$
Therefore,
\begin{eqnarray*}
\sum_{\pi \in \Pi} \mathbf{1}( S(\hat{u}_\pi ) > S^{(k(\alpha))}(\hat{u}_\pi) ) =\sum_{\pi \in \Pi} \mathbf{1}( S(\hat{u}_\pi ) > S^{(k(\alpha))}(\hat{u})) \leq \alpha n.
\end{eqnarray*}
Since $\mathbf{1}( S(\hat{u}) > S^{(k(\alpha))}(\hat{u}) )$ is equal in distribution to  $\mathbf{1}( S(\hat{u}_\pi) > S^{(k(\alpha))}(\hat{u}_\pi) )$  for any $\pi \in \Pi$ by exchangeability (Step 1), we have that
\begin{eqnarray*}
\alpha \geq E \sum_{\pi \in \Pi} \mathbf{1}( S(\hat{u}_\pi) > S^{(k(\alpha))}(\hat{u}_\pi) )/n = E  \mathbf{1}( S(\hat{u}) > S^{(k(\alpha))}(\hat{u}) ) = E \mathbf{1} (\hat p \leq \alpha).
\end{eqnarray*}

\textbf{Step 3:} By continuity of the distribution of $\left \{S(\hat{u}_\pi) \right \}_{\pi\in \Pi}$, there are no ties with probability one. 
Therefore,  
\begin{eqnarray*}
\sum_{\pi\in \Pi} \mathbf{1}( S(\hat{u}_\pi) \le S^{(k(\alpha))}(\hat{u}))= k(\alpha)\le n(1-\alpha)+1 \label{eq:upper_bound}
\end{eqnarray*}
Because
\begin{eqnarray*}
\sum_{\pi\in \Pi} \mathbf{1}( S(\hat{u}_\pi) \le  S^{(k(\alpha))}(\hat{u}))+\sum_{\pi\in \Pi} \mathbf{1}( S(\hat{u}_\pi) >  S^{(k(\alpha))}(\hat{u}))=n,\label{eq:sum_n}
\end{eqnarray*}
we have
\begin{eqnarray*}
\sum_{\pi\in \Pi} \mathbf{1}( S(\hat{u}_\pi) >  S^{(k(\alpha))}(\hat{u}))\ge n\alpha-1.
\end{eqnarray*}
The result now follows by similar arguments as in Step 2.

\subsection{Proof of Lemma \ref{lem: constrained LS}}

Let $X_{jt} $ denote the $(j,t)$ entry of the matrix $X\in\mathbb{R}^{T\times J} $. We assume the following conditions hold: (1) $E(u_{t}X_{jt})=0 $ for $1\leq j\leq J$. (2) there exist constants $c_{1},c_{2}>0$ such that  $ E|X_{jt}u_{t}|^{2}\geq c_{1} $ and  $E|X_{jt}u_{t}|^{3}\leq c_{2}$ for any $1\leq j\leq J$ and $1\leq t\leq T$;  (3) for each $1\leq j\leq J$, the sequence $\{X_{jt}u_{t}\}_{t=1}^{T}$
is $\beta$-mixing and the $\beta$-mixing coefficient satisfies that
$\beta(t)\leq a_{1}\exp(-a_{2}t^{\tau})$, where $a_{1},a_{2},\tau>0$
are constants. (4) there exists a constant $c_{3}>0$ such that $\max_{1\leq j\leq J}\sum_{t=1}^{T}X_{jt}^{2}u_{t}^{2}\leq c_{3}^{2}T$ with probability $1-o(1)$.  (5) $\log J=o(T^{4\tau/(3\tau+4)})$ and $w \in \mathcal{W}$. (6) There exists a sequence $\ell_{T}>0$ such that $(X_{t}'\delta)^{2}\leq \ell_T \|X \delta \|_{2}^{2}/T, \text{ for all } w+ \delta \in \mathcal{W}$ with probability $1-o(1)$
for $T_{0}+1\leq t\leq T$  and  (7) $\ell_T B_T \to 0$
for $B_T = M[\log(T\vee J)]^{(1+\tau)/(2\tau)}T^{-1/2}$.

Then we claim that under conditions (1)-(5)  listed above:
\begin{itemize}
\item[(1)] There exist a constant $M>0$ depending only on $K$ and the constants
listed above such that with probability $1-o(1)$
\[
\|X(\hat{w}-w)\|^2_{2}/ T \leq B_T = M[\log(T\vee J)]^{(1+\tau)/(2\tau)}T^{-1/2} 
\]
\item[(2)] Moreover, if (6) and (7) also hold, then
\[
\frac{1}{T}\sum_{t=1}^{T}\left(\hat{P}_{t}^{N}-P_{t}^{N}\right)^{2}=o_{P}(1) \text{ and }
\hat{P}_{t}^{N}-P_{t}^{N}=o_{P}(1), \text{ for any }  T_{0}+1\leq t\leq T. 
\]
\end{itemize}

The following result is useful in deriving the properties of the $\ell_{1}$-constrained
estimator. 
\begin{lem}
\label{lem: concentration lasso new} Suppose that (1)  $E(u_{t}X_{jt})=0 $ for $1\leq j\leq J$. (2) $\max_{1\leq j\leq J,1\leq t\leq T}E|X_{jt}u_{t}|^{3}\leq K_{1}$
for a constant $K_{1}>0$. (3) $\min_{1\leq j\leq J,1\leq t\leq T}E|X_{jt}u_{t}|^{2}\geq K_{2}$
for a constant $K_{2}>0$.  (4) For each $1\leq j\leq J$, $\{X_{jt}u_{t}\}_{t=1}^{T}$ is $\beta$-mixing
and the $\beta$-mixing coefficients satisfy $\beta(s)\leq D_{1}\exp\left(-D_{2}s^{\tau}\right)$
for some constants $D_{1},D_{2},\tau>0$.  Assume $\log J=o(T^{4\tau/(3\tau+4)})$. Then there exists a constant
$\kappa>0$ depending only on $K_{1},K_{2},D_{1},D_{2},\tau$ such
that with probability $1-o(1)$
\[
\max_{1\leq j\leq J}\left|\sum_{t=1}^{T}X_{jt}u_{t}\right|<\kappa[\log(T\vee J)]^{(1+\tau)/(2\tau)}\max_{1\leq j\leq J}\sqrt{\sum_{t=1}^{T}X_{jt}^{2}u_{t}^{2}}
\]
\end{lem}
\begin{proof}
Define $W_{j,t}=X_{jt}u_{t}$. Let $m=\left\lfloor [4D_{2}^{-1}\log(JT)]^{1/\tau}\right\rfloor $
and $k=\left\lfloor T/m\right\rfloor $. For simplicity, we assume
for now that $T/m$ is an integer. Define 
\[
H_{i}=\left\{ i,m+i,2m+i,...,(k-1)m+i\right\} \qquad\forall1\leq i\leq m.
\]

By Berbee's coupling (Lemma \ref{lem: Berbee}), there
exist a sequence of random variables $\{\tilde{W}_{j,t}\}_{t\in H_{i}}$ such that (1) $\{\tilde{W}_{j,t}\}_{t\in H_{i}}$
is independent across $t$, (2) $\tilde{W}_{j,t}$ has the same distribution
as $W_{j,t}$ for $t\in H_{i}$ and (3) $P\left(\bigcup_{t\in H_{i}}\{\tilde{W}_{j,t}\neq W_{j,t}\}\right)\leq k\beta(m)$. 

By assumption, $\max_{j,t}E|X_{jt}u_{t}|^{3}$ is uniformly bounded
above and $\min_{j,t}E|X_{jt}u_{t}|^{2}$ is uniformly bounded away
from zero. It follows, by Theorem 7.4 of \citet{pena2008self}, that
there exist constants $C_{0},C_{1}>0$ depending on $K_{1}$ and $K_{2}$
such that for any $0\leq x\leq C_{0}k^{1/6}$,
\[
P\left(\left|\frac{\sum_{t\in H_{i}}\tilde{W}_{j,t}}{\sqrt{\sum_{t\in H_{i}}\tilde{W}_{j,t}^{2}}}\right|>x\right)\leq C_{1}\left(1-\Phi(x)\right),
\]
where $\Phi(\cdot)$ is the cdf of $N(0,1)$. Therefore, for any $0\leq x\leq C_{0}k^{1/6}$, 
\begin{multline}
P\left(\left|\frac{\sum_{t\in H_{i}}W_{j,t}}{\sqrt{\sum_{t\in H_{i}}W_{j,t}^{2}}}\right|>x\right)\leq P\left(\left|\frac{\sum_{t\in H_{i}}\tilde{W}_{j,t}}{\sqrt{\sum_{t\in H_{i}}\tilde{W}_{j,t}^{2}}}\right|>x\right)+P\left(\bigcup_{t\in H_{i}}\{\tilde{W}_{j,t}\neq W_{j,t}\}\right)\\
\leq C_{1}\left(1-\Phi(x)\right)+k\beta(m).\label{eq: lasso concentration eq 1}
\end{multline}

The Cauchy-Schwarz inequality implies 
\begin{align*}
\left|\sum_{t=1}^{T}W_{j,t}\right|\leq\sum_{i=1}^{m}\left|\frac{\sum_{t\in H_{i}}W_{j,t}}{\sqrt{\sum_{t\in H_{i}}W_{j,t}^{2}}}\right|\sqrt{\sum_{t\in H_{i}}W_{j,t}^{2}} & \leq\sqrt{\sum_{i=1}^{m}\left(\frac{\sum_{t\in H_{i}}W_{j,t}}{\sqrt{\sum_{t\in H_{i}}W_{j,t}^{2}}}\right)^{2}}\times\sqrt{\sum_{i=1}^{m}\sum_{t\in H_{i}}W_{j,t}^{2}}\\
 & =\sqrt{\sum_{i=1}^{m}\left(\frac{\sum_{t\in H_{i}}W_{j,t}}{\sqrt{\sum_{t\in H_{i}}W_{j,t}^{2}}}\right)^{2}}\times\sqrt{\sum_{t=1}^{T}W_{j,t}^{2}}.
\end{align*}

Hence, 
\[
\left|\frac{\sum_{t=1}^{T}W_{j,t}}{\sqrt{\sum_{t=1}^{T}W_{j,t}^{2}}}\right|\leq\sqrt{\sum_{i=1}^{m}\left(\frac{\sum_{t\in H_{i}}W_{j,t}}{\sqrt{\sum_{t\in H_{i}}W_{j,t}^{2}}}\right)^{2}}.
\]

It follows that for any $0\leq x\leq C_{0}k^{1/6}\sqrt{m},$
\begin{align*}
& P\left(\left|\frac{\sum_{t=1}^{T}W_{j,t}}{\sqrt{\sum_{t=1}^{T}W_{j,t}^{2}}}\right|>x\right)  \leq P\left(\sqrt{\sum_{i=1}^{m}\left(\frac{\sum_{t\in H_{i}}W_{j,t}}{\sqrt{\sum_{t\in H_{i}}W_{j,t}^{2}}}\right)^{2}}>x\right)\\
 & =P\left(\sum_{i=1}^{m}\left(\frac{\sum_{t\in H_{i}}W_{j,t}}{\sqrt{\sum_{t\in H_{i}}W_{j,t}^{2}}}\right)^{2}>x^{2}\right)  \leq\sum_{i=1}^{m}P\left(\left|\frac{\sum_{t\in H_{i}}W_{j,t}}{\sqrt{\sum_{t\in H_{i}}W_{j,t}^{2}}}\right|>\frac{x}{\sqrt{m}}\right)\\
 & \overset{{\rm(i)}}{\leq}m\left[C_{1}\left(1-\Phi(x/\sqrt{m})\right)+k\beta(m)\right] \overset{{\rm (ii)}}{\leq}C_{1}m\sqrt{\frac{m}{2\pi}}x^{-1}\exp\left(-\frac{x^{2}}{2m}\right)+D_{1}km\exp\left(-D_{2}m^{\tau}\right)\\
 & <C_{1}m^{3/2}x^{-1}\exp\left(-\frac{x^{2}}{2m}\right)+D_{1}T\exp\left(-D_{2}m^{\tau}\right)
\end{align*}
where (i) follows by (\ref{eq: lasso concentration eq 1}) and (ii)
follows by the inequality $1-\Phi(a)\leq a^{-1}\phi(a)$ (with $\phi$
being the pdf of $N(0,1)$) and $\beta(m)\leq D_{1}\exp(-D_{2}m^{\tau})$. 

By the union bound, it follows that for any $0\leq x\leq C_{0}k^{1/6}\sqrt{m},$
\[
P\left(\max_{1\leq j\leq J}\left|\frac{\sum_{t=1}^{T}W_{j,t}}{\sqrt{\sum_{t=1}^{T}W_{j,t}^{2}}}\right|>x\right)\leq C_{1}Jm^{3/2}x^{-1}\exp\left(-\frac{x^{2}}{2m}\right)+D_{1}JT\exp\left(-D_{2}m^{\tau}\right).
\]

Now we choose $x=2\sqrt{m\log(Jm^{3/2})}$. Since $\log J=o(T^{4\tau/(3\tau+4)})$
and $k\asymp T/m$, it can be very easily verified that $x\ll C_{0}k^{1/6}\sqrt{m}$
and the two terms on the right-hand side of the above display tend
to zero. The desired result follows. 

If $T/k$ is not an integer, then we simply add one observation from
$\{W_{j,t}\}_{t=km+1}^{T}$ to each of $H_{i}$ for $1\leq i\leq m$.
The bound in (\ref{eq: lasso concentration eq 1}) holds with $C_{1}$
large enough. The proof is complete. 
\end{proof}

Now we are ready to prove Lemma \ref{lem: constrained LS}.

\begin{proof}[Proof of Lemma \ref{lem: constrained LS}]
Let $\Delta=\hat{w}-w$. Since $\|w\|_{1}\leq K$, we have $\|Y-X\hat{w}\|_{2}^{2}\leq\|Y-Xw\|_{2}^{2}$.
Notice that $Y-Xw=u$ and $Y-X\hat{w}=u-X\Delta$. Therefore, $\|u-X\Delta\|_{2}^{2}\leq\|u\|_{2}^{2}$,
which means $\|X\Delta\|_{2}^{2}\leq2u'X\Delta$. Now we observe that
\begin{equation}
\|X\Delta\|_{2}^{2}\leq2u'X\Delta\overset{{\rm (i)}}{\leq}2\|Xu\|_{\infty}\|\Delta\|_{1}\overset{{\rm (ii)}}{\leq}4K\|Xu\|_{\infty},\label{eq: classo eq 1}
\end{equation}
where (i) follows by H\"older's inequality and (ii) follows by
$\|\Delta\|_{1}\leq2K$ (since $\|\hat{w}\|_{1}\leq K$ and $\|w\|_{1}\leq K$).
By Lemma \ref{lem: concentration lasso new}, there exists a constant
$\kappa>0$ such that 
\[
P\left(\max_{1\leq j\leq J}\left|\sum_{t=1}^{T}X_{jt}u_{t}\right|>\kappa[\log(T\vee J)]^{(1+\tau)/(2\tau)}\max_{1\leq j\leq J}\sqrt{\sum_{t=1}^{T}X_{jt}^{2}u_{t}^{2}}\right)=o(1).
\]

Since $P\left(\max_{1\leq j\leq J}\sum_{t=1}^{T}X_{jt}^{2}u_{t}^{2}\leq c_{3}^{2}T\right)\rightarrow1$,
it follows that 
\begin{equation}
P\left(\max_{1\leq j\leq J}\left|\sum_{t=1}^{T}X_{jt}u_{t}\right|>\kappa c_{3}[\log(T\vee J)]^{(1+\tau)/(2\tau)}\sqrt{T}\right)=o(1).\label{eq: classo eq 2}
\end{equation}

Part (1) follows by combining (\ref{eq: classo eq 1}) and (\ref{eq: classo eq 2}).
Part (2) follows by part (1) and $\ell_{T}B_{T}=o(1)$. 
\end{proof}

\subsection{Proof of Lemma \ref{thm: low level pure factor}} We borrow results and notations from \citet{bai2003inferential}. Following standard notation, we use $i$ instead of $j$ to denote units.
Here are the regularity conditions  from \citet{bai2003inferential}.

Suppose that there exists a constant $D_{0}>0$ the following conditions
hold:\\
 (1) $\max_{1\leq t\leq T}E\|F_{t}\|_{2}^{4}\leq D_{0}$, $\max_{1\leq j\leq N}\|\lambda_{j}\|_{2}^{4}\leq D_{0}$,
$\max_{jt}E|u_{jt}|^{8}\leq D_{0}$ and $E(u_{jt})=0$.\\
  (2) $\max_{s}N^{-1}\sum_{t=1}^{T}|\sum_{i=1}^{N}E(u_{is}u_{it})|\leq D_{0}$ and $\max_{i}\sum_{j=1}^{N}\max_{1\leq t\leq T}|E(u_{it}u_{jt})|\leq D_{0}$.\\
(3) $(NT)^{-1}\sum_{s=1}^{T}\sum_{t=1}^{T}\sum_{i=1}^{N}\sum_{j=1}^{N}|E(u_{it}u_{js})|\leq D_{0}$
and $\max_{s,t}E|N^{-1/2}\sum_{i=1}^{N}[u_{is}u_{it}-E(u_{is}u_{it})]|^{4}\leq D_{0}$.\\
(4) $N^{-1}\sum_{i=1}^{N}E\| T^{-1/2}\sum_{t=1}^{T}F_{t}u_{it}\|_{2}^{2}\leq D_{0}$.\\
(5) $\max_{t}E \| (NT)^{-1/2}\sum_{s=1}^{T}\sum_{i=1}^{N}F_{s}[u_{is}u_{it}-E(u_{is}u_{it})]\|_{2}^{2}\leq D_{0}$.\\
(6) $E\| (NT)^{-1/2}\sum_{t=1}^{T}\sum_{i=1}^{N}F_{t}\lambda_{i}'u_{it}\|_{2}^{2}\leq D_{0}.$

Moreover, we assume the following conditions:
(7) for each $t$, $N^{-1/2}\sum_{i=1}^{N}\lambda_{i}u_{it}\rightarrow^{d}N(0,\Gamma_{t})$
as $N\rightarrow\infty$, where $\Gamma_{t}=\lim_{N\rightarrow\infty}N^{-1}\sum_{i=1}^{N}\sum_{j=1}^{N}\lambda_{i}\lambda_{j}'E(u_{it}u_{jt})$.
(8) for each $i$, $T^{-1/2}\sum_{t=1}^{T}F_{t}u_{it}\rightarrow^{d}N(0,\Phi_{i})$
as $T\rightarrow\infty$, where $\Phi_{i}=\lim_{T\rightarrow\infty}T^{-1}\sum_{s=1}^{T}\sum_{t=1}^{T}E(F_{t}F_{s}'u_{is}u_{it})$.
(9) $N^{-1}\sum_{i=1}^{N}\lambda_{i}\lambda_{i}'\rightarrow\Sigma_{\Lambda}$
and $T^{-1}\sum_{t=1}^{T}F_{t}F_{t}'=\Sigma_{F}+o_{P}(1)$ for some
$k\times k$ positive definite matrices $\Sigma_{\Lambda}$ and $\Sigma_{F}$
satisfying that $\Sigma_{\Lambda}\Sigma_{F}$ has distinct eigenvalues.

What follows below is the proof of the lemma.  We recall some notations used by \citet{bai2003inferential}. Let $F=(F_1,\ldots,F_T)'\in \RR^{T\times k}$ and $\Lambda=(\lambda_{1},\ldots,\lambda_{N})'\in\RR^{N\times k}$. Define
$H=(\Lambda'\Lambda/N)(F'\tilde{F}/T)V_{NT}^{-1}$, where $V_{NT}\in\mathbb{R}^{k\times k}$
is the diagonal matrix with the largest $k$ eigenvalues of $Y^{N}(Y^{N})'/(NT)$
on the diagonal and $\tilde{F}$ is the normalized $F$, namely $\tilde{F}'\tilde{F}/T=I_{k}$.

We start with the first equation in the proof of Theorem 3 in \citet{bai2003inferential}
(on page 166):
\begin{equation}
\hat{\lambda}_{1}'\hat{F}_{t}-\lambda_{1}'F_{t}=\left(\hat{F}_{t}-H'F_{t}\right)'H^{-1}\lambda_{1}+\hat{F}_{t}'(\hat{\lambda}_{1}-H^{-1}\lambda_{1}).\label{eq: pure factor eq 1}
\end{equation}

The rest of the proof proceeds in two steps. We first recall some
results from \citet{bai2003inferential} and then derive the desired
result.

\textbf{Step 1:} recall useful results from \citet{bai2003inferential}. By Lemma A.1 of \citet{bai2003inferential}, 
\begin{equation}
\sum_{t=1}^{T}\|\hat{F}_{t}-H'F_{t}\|_{2}^{2}=O_{P}(T/\delta_{NT}^{2}),\label{eq: pure factor eq 2}
\end{equation}
where $\delta_{NT}=\min\{\sqrt{N},\sqrt{T}\}$. By definition, $\hat{F}'\hat{F}/T=I_{k}$,
which means 
\begin{equation}
\sum_{t=1}^{T}\|\hat{F}_{t}\|_{2}^{2}=\sum_{t=1}^{T}{\rm trace}(\hat{F}_{t}\hat{F}_{t}')={\rm trace}(\hat{F}'\hat{F})=kT.\label{eq: pure factor eq 3}
\end{equation}

By Theorem 2 of \citet{bai2003inferential}, 
\begin{equation}
\hat{\lambda}_{1}=H^{-1}\lambda_{1}+O_{P}(\max\{T^{-1/2},N^{-1}\} ).\label{eq: pure factor eq 4}
\end{equation}

By the proof of part (i) in Theorem 2 of \citet{bai2003inferential},
$H$ converges in probability to a nonsingular matrix; see page 166
of \citet{bai2003inferential}. Hence, $\|H^{-1}\|=O_{P}(1)$. By
assumption, $\|\lambda_{1}\|_{2}=O(1)$. Hence, 
\begin{equation}
\|H^{-1}\lambda_{1}\|_{2}=O_{P}(1).\label{eq: pure factor eq 5}
\end{equation}

\textbf{Step 2:} prove the desired result.

Therefore, 
\begin{align*}
\sum_{t=1}^{T}\left(\hat{\lambda}_{1}'\hat{F}_{t}-\lambda_{1}'F_{t}\right)^{2} & \overset{{\rm (i)}}{\leq}2\sum_{t=1}^{T}\left[\left(\hat{F}_{t}-H'F_{t}\right)'H^{-1}\lambda_{1}\right]^{2}+2\sum_{t=1}^{T}\left[\hat{F}_{t}'(\hat{\lambda}_{1}-H^{-1}\lambda_{1})\right]^{2}\\
 & \leq2\sum_{t=1}^{T}\|\hat{F}_{t}-H'F_{t}\|_{2}^{2}\times\|H^{-1}\lambda_{1}\|_{2}^{2}+2\sum_{t=1}^{T}\|\hat{F}_{t}\|_{2}^{2}\times\|\hat{\lambda}_{1}-H^{-1}\lambda_{1}\|_{2}^{2}\\
 & \overset{{\rm (ii)}}{=}O_{P}(T/\delta_{NT}^{2})\times O_{P}(1)+2kT\times O_{P}(\max\{T^{-1},N^{-2}\} )\\
 & =O_{P}(T/\delta_{NT}^{2}),
\end{align*}
where (i) follows by (\ref{eq: pure factor eq 1}) and the elementary
inequality of $(a+b)^{2}\leq2a^{2}+2b^{2}$ for any $a,b\in\mathbb{R}$
and (ii) follows by (\ref{eq: pure factor eq 2}), (\ref{eq: pure factor eq 3}),
(\ref{eq: pure factor eq 4}) and (\ref{eq: pure factor eq 5}). Since
$n=|\Pi|=T$ for moving block permutation, we have 
\[
\frac{1}{n}\sum_{t=1}^{T}\left(\hat{\lambda}_{1}'\hat{F}_{t}-\lambda_{1}'F_{t}\right)^{2}=O_{P}\left(\frac{1}{\min\{N,T\}}\right).
\]
% This proves part (1) of condition (A).
Finally, notice that Theorem 3 of \citet{bai2003inferential} implies $\hat{\lambda}_{1}'\hat{F}_{t}-\lambda_{1}'F_{t}=O_{P}(1/\delta_{NT})$. 
The proof is complete.

\subsection{Proof of Lemma \ref{thm: low level interactive FE}} We recite
conditions from \citet{bai2009panel}. Following standard notation, we use $i$ instead of $j$ to denote units.

Suppose that there exists a constant $D_{1}>0$ the following conditions hold:\\
(1) $\max_{i,t}E\|X_{it}\|_{2}^{4}\leq D_{1}$, $\max_{t}E\|F_{t}\|_{2}^{4}\leq D_{1}$, $\max_{i}E\|\lambda_{i}\|_{2}^{4}\leq D_{1}$ and $\max_{i,t}E|u_{it}|^{8}\leq D_{1}$. \\
(2) $N^{-1}\sum_{i=1}^{N}\sum_{j=1}^{N}\max_{t,s}|E(u_{it}u_{js})|$ $\leq$ $D_{1}$and $T^{-1}\sum_{s=1}^{T}\sum_{t=1}^{T}\max_{i,j}|E(u_{it}u_{js})|\leq D_{1}$.\\
(3) $(NT)^{-1}\sum_{i=1}^{N}\sum_{j=1}^{N}\sum_{s=1}^{T}\sum_{t=1}^{T}|E(u_{it}u_{js})|\leq D_{1}$. \\
(4) $\max_{t,s}E\left|N^{-1/2}\sum_{i=1}^{N}[u_{is}u_{it}-E(u_{is}u_{it})]\right|^{4}\leq D_{1}$. \\
(5) $T^{-2}N^{-1}\sum_{t,s,q,v}\sum_{i,j}|cov(u_{it}u_{ts},u_{jq}u_{jv})|\leq D_{1}$\\
 (6) $T^{-1}N^{-2}\sum_{t,s}\sum_{i,j,k,q}|cov(u_{it}u_{jt},u_{ks}u_{qs})|\leq D_{1}$.\\
(7) the largest eigenvalue of $E(u_{i}u_{i}')$ is bounded by $D_{1}$,
where $u_{i}=(u_{i1},...,u_{iT})'\in\mathbb{R}^{T}$.

Moreover, the following conditions also hold: (8) $u=(u_{1},\ldots,u_{N})$ is independent of $(X,F,\Lambda)$. 
(9) $F'F/T=\Sigma_{F}+o_{P}(1)$ and $\Lambda'\Lambda/N=\Sigma_{\Lambda}+o_{P}(1)$
for some matrices $\Sigma_{F}$ and $\Sigma_{\Lambda}$. 
(10) $N/T$ is bounded away from zero and infinity.
(11) For $X_{i}=(X_{i1},...,X_{iT})'\in\mathbb{R}^{T\times k_{x}}$
and $M_{F}=I_{T}-F(F'F)^{-1}F'$, we have 
\[
\inf_{F:\ F'F/T=I_{k}}\frac{1}{NT}\sum_{i=1}^{N}X_{i}'M_{F}X_{i}-\frac{1}{T}\left[\frac{1}{N^{2}}\sum_{i=1}^{N}\sum_{j=1}^{N}X_{i}'M_{F}X_{j}\lambda_{i}'(\Lambda'\Lambda/N)^{-1}\lambda_{j}\right]>0.
\]

What follows below is the proof of the lemma.  We introduce some notations used in \citet{bai2009panel}. Let $H=(\Lambda'\Lambda/N)(F'\hat{F}/T)V_{NT}^{-1}$,
where $V_{NT}$ is the diagonal matrix that contains the $k$ largest
eigenvalues of $(NT)^{-1}\sum_{i=1}^{N}(Y_{i}^{N}-X_{i}\hat{\beta})(Y_{i}^{N}-X_{i}\hat{\beta})'$
with  $Y_{i}^{N}=(Y_{i1}^{N},Y_{i2}^{N},...,Y_{iT}^{N})'\in\mathbb{R}^{T}$.
Let $\delta_{NT}=\min\{\sqrt{N},\sqrt{T}\}$. The rest of the proof
proceeds in two steps. We first derive bounds for $\sum_{t=1}^{T}\left(\hat{u}_{1t}-u_{1t}\right)^{2}$
and then prove the pointwise result.

\textbf{Step 1:} derive bounds for $\sum_{t=1}^{T}\left(\hat{u}_{1t}-u_{1t}\right)^{2}$. 

Define $\Delta_{\beta}=\hat{\beta}-\beta$ and $\Delta_{F,t}=\hat{F}_{t}-H'F_{t}$.
Denote $\Delta_{F}=(\Delta_{F,1},...,\Delta_{F,T})'\in\mathbb{R}^{T\times k}$.
Notice that $\hat{F}-FH=\Delta_{F}$. As pointed out on page 1237
of \citet{bai2009panel}, 
\begin{equation}
\hat{\lambda}_{1}=T^{-1}\hat{F}'(Y_{1}^{N}-X_{1}\hat{\beta})=T^{-1}\hat{F}'(u_{1}+F\lambda_{1}-X_{1}\Delta_{\beta}).\label{eq: interactive FE eq 1}
\end{equation}
Notice
that
\begin{align}
 & \left|\hat{u}_{1t}-u_{1t}\right|^{2}\nonumber  =\left|F_{t}'\lambda_{1}-\hat{F}_{t}'\hat{\lambda}_{1}-X_{1t}'\Delta_{\beta}\right|^{2}\nonumber \\
 & \overset{{\rm (i)}}{=}\left|F_{t}'\lambda_{1}-T^{-1}(H'F_{t}+\Delta_{F,t})'\hat{F}'(u_{1}+F\lambda_{1}-X_{1}\Delta_{\beta})-X_{1,t}'\Delta_{\beta}\right|^{2}\nonumber \\
 & \leq\left[\left|F_{t}'\left(I_{k}-H\hat{F}'F/T\right)\lambda_{1}\right|+\left|T^{-1}\Delta_{F,t}'\hat{F}'F\lambda_{1}\right|+\left|T^{-1}\hat{F}_{t}'\hat{F}'(u_{1}-X_{1}\Delta_{\beta})\right|+\left|X_{1t}'\Delta_{\beta}\right|\right]^{2}\nonumber \\
 & \lesssim\left[F_{t}'\left(I_{k}-H\hat{F}'F/T\right)\lambda_{1}\right]^{2}+\left[T^{-1}\Delta_{F,t}'\hat{F}'F\lambda_{1}\right]^{2}+\left[T^{-1}\hat{F}_{t}'\hat{F}'(u_{1}-X_{1}\Delta_{\beta})\right]^{2}+\left[X_{1t}'\Delta_{\beta}\right]^{2},\label{eq: interactive FE eq 2}
\end{align}
where (i) follows by (\ref{eq: interactive FE eq 1}) and $\hat{F}_{t}=H'F_{t}+\Delta_{F,t}$.
Therefore, 
\begin{align*}
& \sum_{t=1}^{T}\left(\hat{u}_{1t}-u_{1t}\right)^{2} \lesssim\sum_{t=1}^{T}\left[F_{t}'\left(I_{k}-H\hat{F}'F/T\right)\lambda_{1}\right]^{2}+\sum_{t=1}^{T}\left[T^{-1}\Delta_{F,t}'\hat{F}'F\lambda_{1}\right]^{2}\\
 & \qquad+\sum_{t=1}^{T}\left[T^{-1}\hat{F}_{t}'\hat{F}'(u_{1}-X_{1}\Delta_{\beta})\right]^{2}+\sum_{t=1}^{T}\left[X_{1t}'\Delta_{\beta}\right]^{2}\\
 & \overset{{\rm (i)}}{=}\lambda_{1}'\left(I_{k}-H\hat{F}'F/T\right)'(F'F)\left(I_{k}-H\hat{F}'F/T\right)\lambda_{1}\\
 & \qquad+T^{-2}\left(\hat{F}'F\lambda_{1}\right)'\left(\Delta_{F}'\Delta_{F}\right)\left(\hat{F}'F\lambda_{1}\right)+T^{-1}\left\Vert \hat{F}'(u_{1}-X_{1}\Delta_{\beta})\right\Vert_{2}^{2}+\|X_{1}\Delta_{\beta}\|_{2}^{2}\\
 & \overset{{\rm (ii)}}{=}O_{P}\left(T\|\Delta_{\beta}\|_{2}^{2}+T\delta_{NT}^{-4}\right)+O_{P}\left(T\|\Delta_{\beta}\|_{2}^{2}+T\delta_{NT}^{-2}\right)+O_{P}\left(1+T\delta_{NT}^{-4}+T\|\Delta_{\beta}\|_{2}^{2}\right)+O_{P}(T\|\Delta_{\beta}\|_{2}^{2})\\
 & =O_{P}\left(1+T\|\Delta_{\beta}\|_{2}^{2}+T\delta_{NT}^{-2}\right),
\end{align*}
where (i) follows by $\sum_{t=1}^{T}\hat{F}_{t}\hat{F}_{t}'=\hat{F}'\hat{F}=TI_{k}$
and (ii) follows by Lemma \ref{lem: interactive FE tool}, together
with $\|F\|=O_{P}(\sqrt{T})$, $\lambda_{1}=O(1)$ and $\|\hat{F}\|=O_{P}(\sqrt{T})$.
Since $N\asymp T$, Theorem 1 of \citet{bai2009panel} implies $\|\Delta_{\beta}\|_{2}=O_{P}(1/\sqrt{NT})=O_{P}(T^{-1})$.
Therefore, the above display implies 
\[
\sum_{t=1}^{T}\left(\hat{u}_{1t}-u_{1t}\right)^{2}=O_{P}(1).
\]

\textbf{Step 2:} show the pointwise result.

By (\ref{eq: interactive FE eq 2}), we have 
\begin{align*}
\left|\hat{u}_{1t}-u_{1t}\right| & \leq\left|F_{t}'\left(I_{k}-H\hat{F}'F/T\right)\lambda_{1}\right|+\left|T^{-1}\Delta_{F,t}'\hat{F}'F\lambda_{1}\right|+\left|T^{-1}\hat{F}_{t}'\hat{F}'(u_{1}-X_{1}\Delta_{\beta})\right|+\left|X_{1t}'\Delta_{\beta}\right|\\
 & \overset{{\rm (i)}}{\leq}\|F_{t}\|_{2}\cdot\|\lambda_{1}\|_{2}\cdot O_{P}\left(\|\Delta_{\beta}\|_{2}+\delta_{NT}^{-2}\right)+O_{P}\left(T\|\Delta_{\beta}\|_{2}+T\delta_{NT}^{-2}\right)\cdot T^{-1}\|F\lambda_{1}\|_{2}\\
 & \qquad+T^{-1}\|\hat{F}_{t}\|_{2}\cdot O_{P}\left(\sqrt{T}+T\delta_{NT}^{-2}+T\|\Delta_{\beta}\|_{2}\right)+\|X_{1t}\|_{2}\cdot\|\Delta_{\beta}\|_{2} \overset{{\rm (ii)}}{\leq}O_{P}(T^{-1/2}),
\end{align*}
where (i) follows by $I_{k}-H\hat{F}'F/T=O_{P}(\|\Delta_{\beta}\|_{2}+\delta_{NT}^{-2})$, $\|\hat{F}\Delta_{F,t}\|=O_{P}(T\|\Delta_{\beta}\|_{2}+T\delta_{NT}^{-2})$, 
and $\|\hat{F}'(u_{1}-X_{1}\Delta_{\beta})\|_{2} =O_{P}(\sqrt{T}+T\delta_{NT}^{-2}+T\|\Delta_{\beta}\|_{2})$
(due to Lemma \ref{lem: interactive FE tool}), whereas (ii) follows
by $\|\hat{F}_{t}\|_{2}=O_{P}(1)$ (Lemma \ref{lem: interactive FE tool}),
$\|X_{1t}\|_{2}=O_{P}(1)$, $\|F_{t}\|_{2}=O_{P}(1)$, $\lambda_{1}=O(1)$,
$\|\Delta_{\beta}\|_{2}=O_{P}(T^{-1})$ and $\|F\lambda_{1}\|_{2}=O_{P}(\sqrt{T})$. 

\begin{lem}
\label{lem: interactive FE tool}Suppose that the assumption of Lemma
\ref{thm: low level interactive FE} holds. Let $\delta_{NT}$, $H$,
$\Delta_{F}$ and $u_{1}$ be defined as in the proof of Lemma \ref{thm: low level interactive FE}.
Then (1) $I_{k}-H\hat{F}'F/T=O_{P}(\|\Delta_{\beta}\|_{2}+\delta_{NT}^{-2})$;
(2) $\Delta_{F}'\Delta_{F}=O_{P}(T\|\Delta_{\beta}\|_{2}^{2}+T\delta_{NT}^{-2})$;
(3) $\left\Vert \hat{F}'(u_{1}-X_{1}\Delta_{\beta})\right\Vert =O_{P}\left(\sqrt{T}+T\delta_{NT}^{-2}+T\|\Delta_{\beta}\|_{2}\right)$;
(4) $\|X_{1}\Delta_{\beta}\|_{2}=O_{P}(\sqrt{T}\|\Delta_{\beta}\|_{2})$;
(5) $\|\hat{F}\Delta_{F,t}\|_{2}=O_{P}(T\|\Delta_{\beta}\|_{2}+T\delta_{NT}^{-2})$;
(6) $\|\hat{F}_{t}\|_{2}=O_{P}(1)$ for $1\leq t\le T$. 
\end{lem}
\begin{proof}
\textbf{Proof of part (1).} Lemma A.7(i) in \citet{bai2009panel}
implies $HH'$ converges in probability to a nonsingular matrix. Hence,
\begin{equation}
H=O_{P}(1)\quad{\rm and}\quad H^{-1}=O_{P}(1).\label{eq: interactive FE eq 7}
\end{equation}

Notice that 
\begin{align}
I_{k}-H\hat{F}'F/T\overset{{\rm (i)}}{=}I_{k}-H(FH+\Delta_{F})'F/T & =I_{k}-(HH')(F'F/T)-H\Delta_{F}'F/T\nonumber \\
 & \overset{{\rm (ii)}}{=}O_{P}(\|\Delta_{\beta}\|_{2})+O_{P}(\delta_{NT}^{-2})-H\Delta_{F}'F/T\nonumber \\
 & \overset{{\rm (iii)}}{=}O_{P}(\|\Delta_{\beta}\|_{2})+O_{P}(\delta_{NT}^{-2}),\label{eq: interactive FE eq 8}
\end{align}
where (i) holds by $\hat{F}=FH+\Delta_{F}$, (ii) holds by $I_{k}-(HH')(F'F/T)=O_{P}(\|\Delta_{\beta}\|_{2})+O_{P}(\delta_{NT}^{-2})$
(due to Lemma A.7(i) in \citet{bai2009panel}) and (iii) holds
by (\ref{eq: interactive FE eq 7}) and $\Delta_{F}'F/T=O_{P}(\|\Delta_{\beta}\|_{2})+O_{P}(\delta_{NT}^{-2})$
(due to Lemma A.3(i) in \citet{bai2009panel}). This proves part (1).

\textbf{Proof of part (2).} Part (2) follows by Proposition A.1 of
\citet{bai2009panel}: 
\begin{equation}
T^{-1}\Delta_{F}'\Delta_{F}=O_{P}(\|\Delta_{\beta}\|_{2}^{2})+O_{P}(\delta_{NT}^{-2}).\label{eq: interactive FE eq 9}
\end{equation}

\textbf{Proof of part (3).} To see part (3), first observe that the
independence between $u_{1}$ and $F$ implies that 
\[
E(\|F'u_{1}\|^{2}\mid F)\leq\sum_{t=1}^{T}E(F_{t}'F_{t}u_{1t}^{2}\mid F)=\sum_{t=1}^{T}F_{t}'F_{t}E(u_{1t}^{2}).
\]
It follows that 
\[
E\left(\|F'u_{1}\|^{2}\right)\leq\sum_{t=1}^{T}E(F_{t}'F_{t})E(u_{1t}^{2})\overset{{\rm (i)}}{\lesssim}T\sum_{t=1}^{T}E(u_{1t}^{2})=O(T),
\]
where (i) holds by the uniform boundedness of $E(F_{t}'F_{t})$. This means that 
\begin{equation}
\|F'u_{1}\|_{2}=O_{P}(\sqrt{T}).\label{eq: interactive FE eq 10}
\end{equation}

Notice that 
\begin{align*}
\left\Vert \hat{F}'(u_{1}-X_{1}\Delta_{\beta})\right\Vert_{2}  & \leq\left\Vert H'F'u_{1}\right\Vert_{2} +\left\Vert \left(\hat{F}-FH\right)'u_{1}\right\Vert +\|\hat{F}\|\cdot\|X_{1}\|\cdot\|\Delta_{\beta}\|_{2}\\
 & \overset{{\rm (i)}}{=}\left\Vert H'F'u_{1}\right\Vert_{2} +\left(O_{P}(T^{1/2}\|\Delta_{\beta}\|_{2})+O_{P}(T\delta_{NT}^{-2})\right)+O_{P}(T\|\Delta_{\beta}\|_{2})\\
 & \overset{{\rm (ii)}}{=}O_{P}(\sqrt{T})+\left(O_{P}(T^{1/2}\|\Delta_{\beta}\|_{2})+O_{P}(T\delta_{NT}^{-2})\right)+O_{P}(T\|\Delta_{\beta}\|_{2}),
\end{align*}
where (i) follows by $\left(\hat{F}-FH\right)'u_{1}/T=O_{P}(T^{-1/2}\|\Delta_{\beta}\|_{2})+O_{P}(\delta_{NT}^{-2})$
(due to Lemma A.4 in \citet{bai2009panel}) and the fact that $\|\hat{F}\|=O(\sqrt{T})$
and $\|X_{1}\|=O_{P}(\sqrt{T})$ (see the beginning of Appendix A
in \citet{bai2009panel}), whereas (ii) follows by (\ref{eq: interactive FE eq 7})
and (\ref{eq: interactive FE eq 10}). We have proved part (3). 

\textbf{Proof of part (4).} We notice that $\|X_{1}\|=O_{P}(\sqrt{T})$;
see the beginning of Appendix A in \citet{bai2009panel}. Part (4)
follows by $\|X_{1}\Delta_{\beta}\|\leq\|X_{1}\|\cdot\|\Delta_{\beta}\|_{2}$.

\textbf{Proof of part (5).} Notice that 
\[
\|\hat{F}\Delta_{F,t}\|_{2}/T\leq\|\hat{F}\Delta_{F}\|/T\overset{{\rm (i)}}{\leq}O_{P}(\|\Delta_{\beta}\|_{2})+O_{P}(\delta_{NT}^{-2}),
\]
where (i) follows by Lemma A.3(ii) in \citet{bai2009panel}. We
have proved part (5). 

\textbf{Proof of part (6).} Notice that 
$$
T^{-1}\|\Delta_{F,t}\|_{2}^2\leq T^{-1}\Delta_{F}'\Delta_{F}  = T^{-1}\hat{F}'\Delta_{F}-T^{-1}H'F'\Delta_{F} \overset{{\rm (i)}}{=} O_{P}(\|\Delta_{\beta}\|_{2})+O_{P}(\delta_{NT}^{-2}),
$$
where (i) follows by Lemma A.3(i)-(ii) of \citet{bai2009panel}. By Theorem 1 of \citet{bai2009panel} and by the assumption of $N \asymp T $, we have that $ \|\Delta_{F,t}\|_{2}=O_P(1)$. By $\|\hat{F}_{t}\|_{2}\leq\|H'F_{t}\|_{2}+\|\Delta_{F,t}\|_{2}$, $F_{t}=O_{P}(1)$
and $H=O_{P}(1)$, we can see that $\|\hat{F}_{t}\|_{2}=O_{P}(1)$. The
proof is complete. 
\end{proof}

\subsection{Proof of Lemma \ref{lem: constrained nuclear}}

We start with the assumptions. Recall $N=J+1$.   Assume that (1) $\{u_{j}\}_{j=1}^{N}$
is independent across $j$ conditional on $M$, (2) $\max_{1\leq j\leq N}T^{-1}\sum_{t=1}^{T}E(|u_{jt}|^{2\kappa_{1}}\mid M)=O_{P}(1)$
for some constant $\kappa_{1}>1$, (3) $\|N^{-1}\sum_{j=1}^{N}E(u_{j}u_{j}'\mid M)\|=O_{P}(1)$ and  (4)  there exists a sequence $\ell_T >0$ such that\\
 $\ell_T (NT)^{-1}K\sqrt{N\vee(N^{1/\kappa_{1}}T\log N)}=o(1) $  and  with probability $1-o(1)$, $T^{-1}\sum_{t=1}^{T}(\hat{M}_{1t}-M_{1t})^{2} \leq   \ell_{T} (NT)^{-1}\sum_{t=1}^{T}\sum_{j=1}^{N}(\hat{M}_{jt}-M_{jt})^{2}$ and $(\hat{M}_{1t}-M_{1t})^{2} \leq   \ell_{T} (NT)^{-1}\sum_{t=1}^{T}\sum_{j=1}^{N}(\hat{M}_{jt}-M_{jt})^{2}$ for $T_0+1\leq t\leq T $.

We now prove Lemma \ref{lem: constrained nuclear}. Define
$\Delta=\hat{M}-M$. Let $Y\in\mathbb{R}^{T\times N}$ be the matrix
whose $(t,j)$ entry is $Y_{jt}^{N}$. For  $(j,t)$, define
the matrix $Q_{jt}\in\mathbb{R}^{N\times T}$ by $Q_{is}=\mathbf{1}\{(i,s)=(j,t)\}$,
i.e, a matrix of zeros except that the $(j,t)$ entry is one. Then
we can write the model as 
\begin{equation}
Y_{jt}^{N}={\rm trace}(Q_{jt}'M)+u_{jt}\qquad{\rm for}\quad 1\leq j\leq N,\ 1\leq t\leq T.\label{eq: constrained nuclear 1}
\end{equation}

Notice that the estimator $\hat{M}$ satisfies 
\[
\|\hat{M}\|_{*}\leq K
\]
and
\[
\sum_{t=1}^{T}\sum_{j=1}^{N}\left(Y_{jt}^{N}-{\rm trace}(Q_{jt}'\hat{M})\right)^{2}\leq \sum_{t=1}^{T}\sum_{j=1}^{N} \left(Y_{jt}^{N}-{\rm trace}(Q_{jt}'M)\right)^{2}.
\]

Plugging (\ref{eq: constrained nuclear 1}) into the above inequality
and rearranging terms, we obtain 
\begin{align}
\sum_{t=1}^{T}\sum_{j=1}^{N} \left({\rm trace}(Q_{jt}'\Delta)\right)^{2}\leq2\sum_{t=1}^{T}\sum_{j=1}^{N}u_{jt}{\rm trace}(Q_{jt}'\Delta) & =2{\rm trace}\left(\left[\sum_{t=1}^{T}\sum_{j=1}^{N}u_{jt}Q_{jt}\right]'\Delta\right)\nonumber \\
 & \overset{{\rm (i)}}{=}2{\rm trace}(u'\Delta)\nonumber \\
 & \overset{{\rm (ii)}}{\leq}2\|u\|\cdot\|\Delta\|_{*}\nonumber \\
 & \overset{{\rm (iii)}}{\leq}4K\|u\|,\label{eq: constrained nuclear 2}
\end{align}
where (i) follows by $\sum_{t=1}^{T}\sum_{j=1}^{N}u_{jt}Q_{jt}=u$, (ii) follows
by the trace duality property (see e.g., \citet{mccarthy1967c_},
\citet{rotfel1969singular} and \citet{rohde2011estimation}) and  (iii)
follows by the fact that $\|\hat{M}\|_{*}\leq K$ and $\|M\|_{*}\leq K$.

To bound $\|u\|$, we apply Lemma \ref{lem: maximal SV error matrix}. Note that the conditions of Lemma \ref{lem: maximal SV error matrix}
are satisfied by our assumption. Therefore, $E(\|u\|\mid M)=O_{P}\left(\sqrt{N\vee(N^{1/\kappa_{1}}T\log N)}\right)$.
This and (\ref{eq: constrained nuclear 2}) imply that 
\[
(NT)^{-1} \sum_{t=1}^{T}\sum_{j=1}^{N} \left({\rm trace}(Q_{jt}'\Delta)\right)^{2}=O_{P}\left((NT)^{-1}K\sqrt{N\vee(N^{1/\kappa_{1}}T\log N)}\right).
\]
The desired result follows by Assumption (4) listed at the beginning of the proof.

\begin{lem}
\label{lem: maximal SV error matrix}Suppose that the following conditions
hold:\\
(i) $\{u_{j}\}_{j=1}^{N}$ is independent across $j$ conditional
on $ M$. \\
(ii)  $\max_{1\leq j\leq N}T^{-1}\sum_{t=1}^{T}E(|u_{jt}|^{2\kappa_{1}}\mid M)=O_{P}(1)$
for some constant $\kappa_{1}>1$. \\
(iii)  $\|N^{-1}\sum_{j=1}^{N}E(u_{j}u_{j}'\mid M)\|=O_{P}(1)$. 

Then $E(\|u\|\mid M)=O_{P}\left(\sqrt{N\vee (N^{1/\kappa_{1}}T\log N)}\right).$
\end{lem}
\begin{proof}
Recall the elementary inequality $|T^{-1}\sum_{t=1}^{T}a_{t}|\leq[T^{-1}\sum_{t=1}^{T}|a_{t}|^{\kappa}]^{1/\kappa}$
for any $\kappa>1$ (due to Liapunov's inequality). It follows that
$T^{-1}\sum_{t=1}^{T}u_{jt}^{2}\leq[T^{-1}\sum_{t=1}^{T}|u_{jt}|^{2\kappa_{1}}]^{1/\kappa_{1}}$,
which means 
\begin{equation}
\left(\sum_{t=1}^{T}u_{jt}^{2}\right)^{\kappa_{1}}\leq T^{\kappa_{1}-1}\sum_{t=1}^{T}|u_{jt}|^{2\kappa_{1}}.\label{eq: factor single u eq 1}
\end{equation}

Hence, 
\begin{align*}
E\left(\left[\max_{1\leq j\leq N}\sum_{t=1}^{T}u_{jt}^{2}\right]\mid M\right) & \overset{{\rm (i)}}{\leq}\left\{ E\left[\max_{1\leq j\leq N}\left(\sum_{t=1}^{T}u_{jt}^{2}\right)^{\kappa_{1}}\mid M\right]\right\} ^{1/\kappa_{1}}\\
 & \leq\left\{ E\left[\sum_{i=1}^{N}\left(\sum_{t=1}^{T}u_{jt}^{2}\right)^{\kappa_{1}}\mid M\right]\right\} ^{1/\kappa_{1}}\\
 & \overset{{\rm (ii)}}{\leq}\left\{ E\left[T^{\kappa_{1}-1}\sum_{j=1}^{N}\sum_{t=1}^{T}|u_{jt}|^{2\kappa_{1}}\mid M\right]\right\} ^{1/\kappa_{1}}\\
 & \leq\left\{ \left[NT^{\kappa_{1}}\max_{1\leq i\leq N}T^{-1}\sum_{t=1}^{T}E\left(|u_{jt}|^{2\kappa_{1}}\mid M\right)\right]\right\} ^{1/\kappa_{1}}\\
 & \overset{{\rm (iii)}}{=}\left\{ \left[NT^{\kappa_{1}}O_{P}(1)\right]\right\} ^{1/\kappa_{1}}=O_{P}(N^{1/\kappa_{1}}T),
\end{align*}
where (i) follows by Liapunov's inequality, (ii) follows by (\ref{eq: factor single u eq 1})
and (iii) follows by the assumption that $\max_{1\leq j\leq N}T^{-1}\sum_{t=1}^{T}E(|u_{jt}|^{2\kappa_{1}}\mid M)=O_{P}(1)$.
Therefore, it follows, by Theorem 5.48 and Remark 5.49 of \citet{vershynin2010introduction},
that 
\begin{align*}
E\left(\|u\|\mid M\right)\leq\sqrt{E\left(\|u\|^{2}\mid M\right)} & \leq\|E(u'u\mid M)/N\|^{1/2}\sqrt{N}\\
&\qquad+O\left(\sqrt{O(N^{1/\kappa_{1}}T)\log\min\left(O(N^{1/\kappa_{1}}T),N\right)}\right)\\
 & \overset{{\rm (i)}}{\leq}O_{P}\left(\sqrt{N}\right)+O\left(\sqrt{N^{1/\kappa_{1}}T\log N}\right),
\end{align*}
where (i) holds by the assumption of $\|E(u'u\mid M)/N\|=\|N^{-1}\sum_{j=1}^{N}E(u_{j}u_{j}'\mid M)\|=O_{P}(1)$.
The proof is complete. 
\end{proof}

\subsection{Proof of Lemma  \ref{lem: low level AR}}

% In this proof, we use $\|\cdot\| $ to denote the Euclidean norm of vectors or the spectral norm of matrices. 
By the analysis on page 215-216 of \citet{Hamilton1994} (leading
to Equation (8.2.29) therein), we have that $\hat{\rho}-\rho=o_{P}(1)$.
Hence, 
\begin{align*}
\sum_{t=K+1}^{T}\left(\hat{u}_{t}-u_{t}\right)^{2} & =\sum_{t=K+1}^{T}\left(y_{t}'(\rho-\hat{\rho})\right)^{2} =(\hat{\rho}-\rho)'\left(\sum_{t=K+1}^{T}y_{t}y_{t}'\right)(\hat{\rho}-\rho) \leq\|\hat{\rho}-\rho\|_{2}^{2}\times \left\Vert\sum_{t=K+1}^{T} y_{t}y_{t}'\right\Vert.
 % & =\|\hat{\rho}-\rho\|^{2}\times\sum_{t=K+1}^{T}\left(\sum_{j=1}^{K}(Y^{N}_{t-j})^{2}+1\right) <\|\hat{\rho}-\rho\|^{2}\times\left(K\sum_{t=1}^{T}(Y^{N}_{t})^{2}+T\right).
\end{align*}

The analysis on page 215 of \citet{Hamilton1994} (leading to Equation
(8.2.26) therein) implies that $$T^{-1}\sum_{t=K+1}^{T} y_{t}y_{t}'=E(y_{t}y_{t}')+o_{P}(1),$$
which means $\left\Vert\sum_{t=K+1}^{T} y_{t}y_{t}'\right\Vert=O_{P}(T)$. Since $\hat{\rho}-\rho=o_{P}(1)$,
the above display implies that 
\[
\sum_{t=K+1}^{T}\left(\hat{u}_{t}-u_{t}\right)^{2}=o_{P}(T).
\]

Since $\hat{u}_{t}-u_{t}=y_{t}'(\rho-\hat{\rho})$, the pointwise
consistency follows by $\hat{\rho}-\rho=o_{P}(1)$ and the fact that
$y_{t}=O_{P}(1)$ for $T_{0}+1\leq t\leq T$ (due to the stationarity
property of $u_{t}$).

\subsection{Proof of Lemma \ref{lem: low level nonliear AR}}

By assumption, $\max_{K+1 \leq t \leq T }|\hat{P}_{t}^N-P_{t}^N|\leq \ell_T \| \hat \rho - \rho \|  $. Therefore, 
$$
\frac{1}{T} \sum_{t=K+1}^{T}(\hat{P}^N_{t}-P^N_{t})^{2}\leq \ell_T^2 \| \hat \rho - \rho \|
$$
and
$$
\max_{T_0+1 \leq t \leq T }|\hat{P}^N_{t}-P^N_{t}|\leq \ell_T \| \hat \rho - \rho \| .
$$
Since $\ell_T \| \hat \rho - \rho \|=o_P(1)$, the desired result follows.

\subsection{Proof of Lemma \ref{lem: pre-whitening u hat}}
%In this proof, we use $\|\cdot\| $ to denote the Euclidean norm of vectors or the spectral norm of matrices. 

We first derive the following result that is useful in proving Lemma \ref{lem: pre-whitening u hat}.

\begin{lem}
\label{lem: time series tool}Recall $\varepsilon_{t}=x_{t}'\rho+u_{t}$, where   $\rho=(\rho_{1},\rho_{2},...,\rho_{K})'\in\mathbb{R}^{K}$
and $x_{t}=(\varepsilon_{t-1},\varepsilon_{t-2},...,\varepsilon_{t-K})'\in\mathbb{R}^{K}$.
Suppose that  the following hold: 
(1) $\{u_{t}\}_{t=1}^{T}$ is an iid sequence with $E(u_{1}^{4})$
uniformly bounded. 
(2) the roots of $1-\sum_{j=1}^{K}\rho_{j}L^{j}=0$ are uniformly
bounded away from the unit circle.

Then we have
(i) $(T-K)^{-1}\sum_{t=K+1}^{T}u_{t}^{2}=O_{P}(1)$;
(ii) $(T-K)^{-1}\sum_{t=K+1}^{T}x_{t}u_{t}=o_{P}(1)$;
(iii) $(T-K)^{-1}\sum_{t=K+1}^{T}\|x_{t}\|^{2}=O_{P}(1)$.
(iv) There exists a constant $\lambda_{0}>0$ such that the smallest
eigenvalue of $(T-K)^{-1}\sum_{t=K+1}^{T}x_{t}x_{t}'$ is bounded
below by $\lambda_{0}$ with probability approaching one. 
\end{lem}
\begin{proof}
\textbf{Proof of part (i)}. Part (i) follows by the law of large numbers;
see e.g., Theorem 3.1 of \citet{white2014asymptotic}. 

\textbf{Proof of part (ii)}. Let $\mathcal{F}_{t}$ be the $\sigma$-algebra
generated by $\{ u_{s}: s\leq t \} $. First notice that $\{x_{t}u_{t}\}_{t=K+1}^{T}$ is a
martingale difference sequence with respect to the filtration $\{\mathcal{F}_t\} $. Since $\varepsilon_{t}$ is a stationary
process, we have that $E\|x_{t}u_{t}\|^{2}=  \sum_{j=1}^{K}E(\varepsilon_{t-j}^{2}u_{t}^{2})= \sum_{j=1}^{K}E(\varepsilon_{t-j}^{2})E(u_{t}^{2})$
is uniformly bounded. Hence, part (ii) follows by Exercise
3.77 of \citet{white2014asymptotic}. 

\textbf{Proof of part (iii)}. To see part (iii), notice that $\|x_{t}\|^{2} = x_{t}'x_{t}=\sum_{j=1}^{K}\varepsilon_{t-j}^{2}$.
By the analysis on page 215 of \citet{Hamilton1994}, for each $1\leq j\leq K$,
$(T-K)^{-1}\sum_{t=K+1}^{T}\varepsilon_{t-j}^{2}=E(\varepsilon_{t-j}^{2})+o_{P}(1)$.
Thus, part (iii) follows by
\[
(T-K)^{-1}\sum_{t=K+1}^{T}\|x_{t}\|^{2}=(T-K)^{-1}\sum_{j=1}^{K}\sum_{t=K+1}^{T}\varepsilon_{t-j}^{2}=K\left(E(\varepsilon_{t}^{2})+o_{P}(1)\right).
\]

\textbf{Proof of part (iv)}. Similarly, the analysis on page 215 of
\citet{Hamilton1994} implies that $$(T-K)^{-1}\sum_{t=K+1}^{T}x_{t}x_{t}'=o_{P}(1)+Ex_{t}x_{t}'.$$
By Proposition 5.1.1 of  \citet{brockwell2013time}, $E(x_{t}x_{t}')$
has eigenvalues bounded away from zero. Part (iv) follows. 
\end{proof}

Now we are ready to prove Lemma \ref{lem: pre-whitening u hat}.

\begin{proof}[Proof of Lemma \ref{lem: pre-whitening u hat}]
 Define $\delta_{t}=\hat{\varepsilon}_{t}-\varepsilon_{t}$, $\Delta_{t}=\hat{x}_{t}-x_{t}$,
$\tilde{u}_{t}=u_{t}+\delta_{t}-\Delta_{t}'\rho$ and $a_{t}=\tilde{u}_{t}-u_{t}$.
Notice that 
\begin{equation}
\hat{\varepsilon}_{t}=\delta_{t}+\varepsilon_{t}=\delta_{t}+x_{t}'\rho+u_{t}=\delta_{t}+\left(\hat{x}_{t}-\Delta_{t}\right)'\rho+u_{t}=\hat{x}_{t}'\rho+\tilde{u}_{t}.\label{eq: pre-white eq 1}
\end{equation}

Therefore, 
\begin{align}
\hat{\rho}=\left(\sum_{t=K+1}^{T}\hat{x}_{t}\hat{x}_{t}'\right)^{-1}\left(\sum_{t=K+1}^{T}\hat{x}_{t}\hat{\varepsilon}_{t}\right) & =\left(\sum_{t=K+1}^{T}\hat{x}_{t}\hat{x}_{t}'\right)^{-1}\left(\sum_{t=K+1}^{T}\hat{x}_{t}(\hat{x}_{t}'\rho+\tilde{u}_{t})\right)\nonumber \\
 & =\rho+\left(\sum_{t=K+1}^{T}\hat{x}_{t}\hat{x}_{t}'\right)^{-1}\left(\sum_{t=K+1}^{T}\hat{x}_{t}\tilde{u}_{t}\right).\label{eq: pre-white eq 2}
\end{align}

The rest of the proof proceeds in three steps. First two steps show
that $(T-K)^{-1}\sum_{t=K+1}^{T}\hat{x}_{t}\hat{x}_{t}'$
is well-behaved and $(T-K)^{-1}\sum_{t=K+1}^{T}\hat{x}_{t}\tilde{u}_{t}=o_{P}(1)$.
This would imply $\hat{\rho}=\rho+o_{P}(1)$. In the third step, we
derive the final result. 

\textbf{Step 1:} show that $\left[(T-K)^{-1}\sum_{t=K+1}^{T}\hat{x}_{t}\hat{x}_{t}'\right]^{-1}=O_{P}(1)$. 

It is not hard to see that $\|\Delta_{t}\|^{2}=\sum_{s=t-1}^{t-K}\delta_{s}^{2}$.
Therefore, 
\begin{equation}
\sum_{t=K+1}^{T}\|\Delta_{t}\|^{2}=\sum_{t=K+1}^{T}\sum_{s=t-1}^{t-K}\delta_{s}^{2}\leq K\sum_{t=1}^{T}\delta_{t}^{2}\overset{{\rm (i)}}{=}o_{P}(T),\label{eq: pre-white eq 5}
\end{equation}
where (i) follows by the assumption of $T^{-1}\sum_{t=1}^{T}\delta_{t}^{2}=o_{P}(1)$.
Notice that 
\begin{align}
\left\Vert \sum_{t=K+1}^{T}\left(\hat{x}_{t}\hat{x}_{t}'-x_{t}x_{t}'\right)\right\Vert  & =\left\Vert \sum_{t=K+1}^{T}\left(x_{t}\Delta_{t}'+\Delta_{t}x_{t}'+\Delta_{t}\Delta_{t}'\right)\right\Vert \nonumber \\
 & \leq2\sum_{t=K+1}^{T}\|x_{t}\|\cdot\|\Delta_{t}\|+\sum_{t=K+1}^{T}\|\Delta_{t}\|^{2}\nonumber \\
 & \leq2\sqrt{\left(\sum_{t=K+1}^{T}\|x_{t}\|^{2}\right)\left(\sum_{t=K+1}^{T}\|\Delta_{t}\|^{2}\right)}+\sum_{t=K+1}^{T}\|\Delta_{t}\|^{2}\overset{{\rm (i)}}{=}o_{P}(T),\label{eq: pre-white eq 5.5}
\end{align}
where (i) follows by (\ref{eq: pre-white eq 5}) and Lemma \ref{lem: time series tool}.
Thus, 
\[
\left\Vert \frac{1}{T-K}\sum_{t=K+1}^{T}\left(\hat{x}_{t}\hat{x}_{t}'-x_{t}x_{t}'\right)\right\Vert =o_{P}(1).
\]

By Lemma \ref{lem: time series tool}, the smallest eigenvalue of
$(T-K)^{-1}\sum_{t=K+1}^{T}x_{t}x_{t}'$ is bounded below by a positive constant with probability approaching one. It follows that 
\begin{equation}
\left[(T-K)^{-1}\sum_{t=K+1}^{T}\hat{x}_{t}\hat{x}_{t}'\right]^{-1}=O_{P}(1).\label{eq: pre-white eq 2.5}
\end{equation}

\textbf{Step 2:} show that $(T-K)^{-1}\sum_{t=K+1}^{T}\hat{x}_{t}\tilde{u}_{t}=o_{P}(1)$. 

By Lemma \ref{lem: time series tool}, we have 
\begin{equation}
(T-K)^{-1}\sum_{t=K+1}^{T}x_{t}u_{t}=o_{P}(1).\label{eq: pre-white eq 3}
\end{equation}

Notice that 
\begin{align}
\left\Vert \frac{1}{T-K}\sum_{t=K+1}^{T}\left(\hat{x}_{t}\tilde{u}_{t}-x_{t}u_{t}\right)\right\Vert  & =\left\Vert \frac{1}{T-K}\sum_{t=K+1}^{T}\left(\Delta_{t}u_{t}+x_{t}a_{t}+\Delta_{t}a_{t}\right)\right\Vert \nonumber \\
 & \leq\frac{1}{T-K}\sum_{t=K+1}^{T}\left(\left\Vert \Delta_{t}u_{t}\right\Vert +\left\Vert x_{t}a_{t}\right\Vert +\left\Vert \Delta_{t}a_{t}\right\Vert \right)\nonumber \\
 & \leq\sqrt{\left(\frac{1}{T-K}\sum_{t=K+1}^{T}\|\Delta_{t}\|^{2}\right)\left(\frac{1}{T-K}\sum_{t=K+1}^{T}u_{t}^{2}\right)}\nonumber \\
 & \qquad+\sqrt{\left(\frac{1}{T-K}\sum_{t=K+1}^{T}\|x_{t}\|^{2}\right)\left(\frac{1}{T-K}\sum_{t=K+1}^{T}a_{t}^{2}\right)}\nonumber \\
 & \qquad+\sqrt{\left(\frac{1}{T-K}\sum_{t=K+1}^{T}\|\Delta_{t}\|^{2}\right)\left(\frac{1}{T-K}\sum_{t=K+1}^{T}a_{t}^{2}\right)}.\label{eq: pre-white eq 4}
\end{align}

We observe that
\begin{align}
\sum_{t=K+1}^{T}a_{t}^{2}=\sum_{t=K+1}^{T}\left(\delta_{t}-\Delta_{t}'\rho\right)^{2} & \leq2\sum_{t=K+1}^{T}\delta_{t}^{2}+2\sum_{t=K+1}^{T}(\Delta_{t}'\rho)^{2}\nonumber \\
 & \leq2\sum_{t=1}^{T}\delta_{t}^{2}+2\|\rho\|^{2}\sum_{t=K+1}^{T}\|\Delta_{t}\|^{2}\overset{{\rm (i)}}{=}O_{P}(T),\label{eq: pre-white eq 6}
\end{align}
where (i) follows by (\ref{eq: pre-white eq 5}) and the assumption
of $T^{-1}\sum_{t=1}^{T}\delta_{t}^{2}=o_{P}(1)$. Combining (\ref{eq: pre-white eq 4}) with (\ref{eq: pre-white eq 5}) and (\ref{eq: pre-white eq 6}),
we obtain 
\begin{multline}
\left\Vert \frac{1}{T-K}\sum_{t=K+1}^{T}\left(\hat{x}_{t}\tilde{u}_{t}-x_{t}u_{t}\right)\right\Vert \\
\leq\sqrt{o_{P}(1)\left(\frac{1}{T-K}\sum_{t=K+1}^{T}u_{t}^{2}\right)}+\sqrt{\left(\frac{1}{T-K}\sum_{t=K+1}^{T}\|x_{t}\|^{2}\right)o_{P}(1)}+\sqrt{o_{P}(1)\times o_{P}(1)}\overset{{\rm (i)}}{=}o_{P}(1),\label{eq: pre-white eq 7}
\end{multline}
 where (i) follows by Lemma \ref{lem: time series tool}. Now we
combine (\ref{eq: pre-white eq 3}) and (\ref{eq: pre-white eq 7}),
obtaining 
\begin{equation}
(T-K)^{-1}\sum_{t=K+1}^{T}\hat{x}_{t}\tilde{u}_{t}=o_{P}(1).\label{eq: pre-white eq 8}
\end{equation}

By (\ref{eq: pre-white eq 2}) together with (\ref{eq: pre-white eq 2.5})
and (\ref{eq: pre-white eq 8}), it follows that 
\begin{equation}
\hat{\rho}-\rho=o_{P}(1).\label{eq: pre-white eq 9}
\end{equation}

\textbf{Step 3:} show the desired result.

Recall that $\hat{u}_{t}=\hat{\varepsilon}_{t}-\hat{x}_{t}'\hat{\rho}$.
Hence, 
\begin{equation}
\hat{u}_{t}-u_{t}=\left(\hat{\varepsilon}_{t}-\hat{x}_{t}'\hat{\rho}\right)-u_{t}\overset{{\rm (i)}}{=}\left(\hat{x}_{t}'(\rho-\hat{\rho})+\tilde{u}_{t}\right)-u_{t}=\hat{x}_{t}'(\rho-\hat{\rho})+a_{t},\label{eq: pre-white eq 10}
\end{equation}
where (i) follows by (\ref{eq: pre-white eq 1}). Therefore, we
have 
\begin{align*}
\sum_{t=K+1}^{T}\left(\hat{u}_{t}-u_{t}\right)^{2} & =\sum_{t=K+1}^{T}\left(\hat{x}_{t}'(\rho-\hat{\rho})+a_{t}\right)^{2}\\
 & \leq2\sum_{t=K+1}^{T}\left(\hat{x}_{t}'(\hat{\rho}-\rho)\right)^{2}+2\sum_{t=K+1}^{T}a_{t}^{2}\\
 & \leq2\|\hat{\rho}-\rho\|^{2}\sum_{t=K+1}^{T}\|\hat{x}_{t}\|^{2}+2\sum_{t=K+1}^{T}a_{t}^{2}\\
 & =2\|\hat{\rho}-\rho\|^{2}\left(\sum_{t=K+1}^{T}{\rm trace}(x_{t}x_{t}')+\sum_{t=K+1}^{T}{\rm trace}(\hat{x}_{t}\hat{x}_{t}'-x_{t}x_{t}')\right)+2\sum_{t=K+1}^{T}a_{t}^{2}\\
 & \overset{{\rm (i)}}{\leq}o_{P}(1)\times\left(O_{P}(T)+o_{P}(T)\right)+o_{P}(T)=o_{P}(T),
\end{align*}
where (i) follows by (\ref{eq: pre-white eq 5.5}), (\ref{eq: pre-white eq 9}),
(\ref{eq: pre-white eq 6}) and Lemma \ref{lem: time series tool}. 

To see the pointwise result, we notice that by (\ref{eq: pre-white eq 10})
and (\ref{eq: pre-white eq 9}), it suffices to verify that $a_{t}=o_{P}(1)$
and $\hat{x}_{t}=O_{P}(1)$ for $T_{0}+1\leq t\leq T$. 

Since $\hat{x}_{t}-x_{t}=(\delta_{t-1},\delta_{t-2},...,\delta_{t-K})'$,
the assumption of pointwise convergence of $\hat{\varepsilon}_{t}$
(i.e., $\delta_{t}=o_{P}(1)$ for $T_{0}+1-K\leq t\leq T$) implies
that $\hat{x}_{t}-x_{t}=o_{P}(1)$ for $T_{0}+1\leq t\leq T$. Since
$x_{t}=O_{P}(1)$ due to the stationarity condition, we have $\hat{x}_{t}=O_{P}(1)$
for $T_{0}+1\leq t\leq T$. 

Since both $\delta_{t}$ and $\Delta_{t}$ are both $o_{P}(1)$ for
$T_{0}+1\leq t\leq T$, so is $a_{t}=\delta_{t}-\Delta_{t}'\rho$.
Hence, we have proved the pointwise result. The proof is complete. 
\end{proof}

\subsection{Proof of Lemma  \ref{lem: low-dimensional stability}}

	Fix an arbitrary $\eta>0$. Define $a_{\eta}=\inf_{\|\beta-\beta_{*}\|_{2}\geq\eta}(L(\beta)-L(\beta_{*}))/3$.
	By the compactness of $\Bcal$ and the uniqueness of the minimum of
	$L(\cdot)$, we have $a_{\eta}>0$. 
	
	(Otherwise, one can find a sequence $\{\beta_{k}\}_{k=1}^{\infty}$
	with $\|\beta_{k}-\beta_{*}\|_{2}\geq\eta$ for all $k\geq1$ with
	$L(\beta_{k})\rightarrow L(\beta_{*})$. By compactness of $\Bcal$
	implies that some subsequence of $\beta_{k}$ converges to a point
	$\beta_{**}\in\Bcal$. Clearly, $\|\beta_{**}-\beta_{*}\|_{2}\geq\eta$.
	The continuity of $L(\cdot)$ implies $L(\beta_{**})=L(\beta_{*})$.
	This contradicts the uniqueness of the minimum of $L(\cdot)$.)
	
	Define the event $$\Mcal=\left\{ \sup_{\beta}|\hat{L}(\Zb;\beta)-L(\beta)|\leq a_{\eta}\right\} \bigcap\left\{ \max_{H\in\Hcal}\sup_{\beta}|\hat{L}(\Zb_{H};\beta)-L(\beta)|\leq a_{\eta}\right\}. $$
	By the assumption, we know $P(\Mcal)=1-o(1)$. 
	
	Notice that on the event $\Mcal$, 
	\[
	L(\hbeta(\Zb))-L(\beta_{*})\leq a_{\eta}+\hat{L}(\Zb;\hbeta(\Zb))-L(\beta_{*})\leq2a_{\eta}+\hat{L}(\Zb;\hbeta(\Zb))-\hat{L}(\Zb;\beta_{*})\leq2a_{\eta}.
	\]
	
	It follows by the definition of $a_{\eta}$ that $\|\hbeta(\Zb)-\beta_{*}\|_{2}\leq\eta$
	on the event $\Mcal$. Thus, $P(\|\hbeta-\beta_{*}\|_{2}\leq\eta)\geq P(\Mcal)=1-o(1)$.
	Since $\eta>0$ is arbitrary, we have $\|\hbeta(\Zb)-\beta_{*}\|_{2}=o_{P}(1)$. 
	
	By the same analysis, we have that on the event $\Mcal$, $\|\hbeta(\Zb_{H})-\beta_{*}\|_{2}\leq\eta$
	for all $H\in\Hcal$. Thus, on the event $\Mcal$, $\max_{H\in\Hcal}\|\hbeta(\Zb_{H})-\beta_{*}\|_{2}\leq\eta$.
	We have that $\max_{H\in\Hcal}\|\hbeta(\Zb_{H})-\beta_{*}\|_{2}=o_{P}(1)$.
	The desired result follows. 

\subsection{Proof of Lemma  \ref{lem: classo stability}}
For notational simplicity, we write $\hbeta=\hbeta(\Zb) $ and  $\hbeta_{H}=\hbeta(\Zb_{H}) $. Define the event $\Mcal=\Mcal_{1}\bigcap\Mcal_{2}$, where $\Mcal_{1}=\{\min_{\|v\|_{0}\leq m}v'\hSigma v/\|v\|_{2}^{2}\geq\kappa_{0}\}\bigcap\{\|\hbeta\|_{0}\leq s/2\}$
	and 
	\[
	\Mcal_{2}=\bigcap_{H\in\Hcal}\left(\left\{ \|\hSigma_{H}-\hSigma\|_{\infty}\leq c_{T},\ \|\hmu_{H}-\hmu\|_{\infty}\leq c_{T}\right\} \bigcap\left\{ \max_{H\in\Hcal}\|\hbeta_{H}\|_{0}\leq s/2\right\} \right).
	\]
	
	By assumption, $P(\Mcal)\geq1-\gamma_{1,T}-\gamma_{2,T}-\gamma_{3,T}$.
	The rest of the argument are statements on the event $\Mcal$. 
	
	Fix $H\in\Hcal$, let $\Delta=\hbeta_{H}-\hbeta$. Define $\xi=\hmu-\hSigma\hbeta$
	and $\xi_{H}=\hmu_{H}-\hSigma_{H}\hbeta$. Since $\|\hbeta\|_{1}\leq K$,
	we have 
	\begin{equation}
	\|\xi_{H}-\xi\|_{\infty}\leq\|\hmu_{H}-\hmu\|_{\infty}+\|\hSigma_{H}-\hSigma\|_{\infty}\|\hbeta\|_{1}\leq c_{T}(1+K).\label{eq: classo stability eq 7}
	\end{equation}
	
	When $\Delta=0$, the result clearly holds. Now we consider the case
	with $\Delta\neq0$. 
	
	\textbf{Step 1:} show that on the event $\Mcal$, $0\leq\lambda^{2}\Delta'\hSigma\Delta-\lambda\Delta'\hmu\leq c_{T}K^{2}+2c_{T}K$
	for any $\lambda\in[0,1]$. 
	
	Recall that $\hQ(\beta)=\beta'\hSigma\beta-2\hmu'\beta+T^{-1}\sum_{t=1}^{T}Y_{t}^{2}$. Since the term $T^{-1}\sum_{t=1}^{T}Y_{t}^{2}$ does not affect the minimizer, we modify $\hQ$ by dropping this term. With a slight abuse of notation, we still use the symbol  $\hQ(\beta)=\beta'\hSigma\beta-2\hmu'\beta$ and  $\hQ_{H}(\beta)=\beta'\hSigma_{H}\beta-2\hmu_{H}'\beta$. Therefore,  
	\[
	\hQ(\beta)-\hQ_{H}(\beta)=\beta'(\hSigma-\hSigma_{H})\beta-2(\hmu-\hmu_{H})'\beta.
	\]
	
	Since $\sup_{\beta\in\Wcal}\|\beta\|_{1}\leq K$, we have that on
	the event $\Mcal$, 
	\[
	\sup_{\beta\in\Wcal}\left|\hQ(\beta)-\hQ_{H}(\beta)\right|\leq c_{T}K^{2}+2c_{T}K.
	\]
	
	Let $\bar{\beta}=\hbeta+\lambda\Delta$, where $\lambda\in[0,1]$.
	Then clearly, $\bar{\beta}=\lambda\hbeta_{H}+(1-\lambda)\hbeta$.
	By definition of $\hbeta$, we have that $\hQ(\hbeta)\leq\hQ(\bar{\beta})$,
	which means that 
	\begin{equation}
	\lambda^{2}\Delta'\hSigma\Delta-2\lambda\xi'\Delta\geq0.\label{eq: classo stability eq 8}
	\end{equation}
	
	Clearly, $\hQ_{H}(\hbeta_{H})\leq\hQ_{H}(\bar{\beta})$ for any $\lambda\in[0,1]$. Hence, $\hQ_{H}(\hbeta+\Delta)\leq\hQ_{H}(\hbeta+\lambda \Delta)$ for any $\lambda\in(0,1)$.	By $\hQ_{H}(\beta)=\beta'\hSigma_{H}\beta-2\hmu_{H}'\beta$, this simplifies to $(1+\lambda)\Delta' \hSigma_{H} \Delta \leq 2\xi_{H}'\Delta $. Since  $\lambda$ can be arbitrarily close to one, this means	$\Delta' \hSigma_{H} \Delta \leq \xi_{H}'\Delta $. It follows that for any $\lambda \in[0,1]$, we have
	\[
	\lambda^{2}\Delta'\hSigma_{H}\Delta\leq2\lambda\xi_{H}'\Delta.
	\]
	
	Notice that 
	\begin{multline*}
	0\leq2\lambda\xi_{H}'\Delta-\lambda^{2}\Delta'\hSigma_{H}\Delta=2\lambda\xi'\Delta-\lambda^{2}\Delta'\hSigma\Delta+2\lambda(\xi_{H}-\xi)'\Delta+\lambda^{2}\Delta'(\hSigma-\hSigma_{H})\Delta\\
	\leq2\lambda\xi'\Delta-\lambda^{2}\Delta'\hSigma\Delta+2\lambda\|\xi_{H}-\xi\|_{\infty}\|\Delta\|_{1}+\lambda^{2}\|\hSigma-\hSigma_{H}\|_{\infty}\|\Delta\|_{1}^{2}.
	\end{multline*}
	
	It follows, by (\ref{eq: classo stability eq 7}) and $\|\Delta\|_{1}\leq\|\hbeta_{H}\|_{1}+\|\hbeta\|_{1}\leq2K$,
	that
	\begin{equation}
	\lambda^{2}\Delta'\Sigma\Delta-2\lambda\xi'\Delta\leq2\lambda\|\xi_{H}-\xi\|_{\infty}\|\Delta\|_{1}+\lambda^{2}\|\hSigma-\hSigma_{H}\|_{\infty}\|\Delta\|_{1}^{2}\leq4c_{T}(1+K)K+4c_{T}K^{2}.\label{eq: classo stability eq 9}
	\end{equation}
	
	Since (\ref{eq: classo stability eq 8}) and (\ref{eq: classo stability eq 9})
	hold for any $\lambda\in[0,1]$, we have that 
	\begin{equation}
	0\leq\lambda^{2}\Delta'\hSigma\Delta-2\lambda\xi'\Delta\leq4c_{T}K(2K+1)\qquad\forall\lambda\in[0,1].\label{eq: classo stability eq 10}
	\end{equation}
	
	\textbf{Step 2:} show the desired result.
	
	Since $0\leq\lambda^{2}\Delta'\hSigma\Delta-2\lambda\xi'\Delta$ for
	any $\lambda\in(0,1)$, we have that $\xi'\Delta\leq\lambda\Delta'\hSigma\Delta/2$
	for any $\lambda\in(0,1)$. Thus, 
	\[
	\xi'\Delta\leq0.
	\]
	
	Hence, by the second inequality in (\ref{eq: classo stability eq 10}),
	for any $\lambda\in[0,1]$,
	\[
	\lambda^{2}\Delta'\hSigma\Delta\leq\lambda^{2}\Delta'\hSigma\Delta-2\lambda\xi'\Delta\leq4c_{T}K(2K+1).
	\]
	
	Now we take $\lambda=1$, which implies 
	\[
	\Delta'\hSigma\Delta\leq4c_{T}K(2K+1).
	\]
	
	Since $\|\Delta\|_{0}\leq\|\hbeta_{H}\|_{0}+\|\hbeta\|_{0}\leq s$
	and $\|\Delta\|_{1}\leq\sqrt{\|\Delta\|_{0}}\|\Delta\|_{2}$, it follows
	that 
	\[
	4c_{T}K(2K+1)\geq\Delta'\hSigma\Delta\geq\kappa_{1}\|\Delta\|_{2}^{2}\geq\kappa_{1}s^{-1}\|\Delta\|_{1}^{2}.
	\]
	
	Hence, $\|\Delta\|_{1}\leq2\sqrt{\kappa_{1}sc_{T}K(2K+1)}$ and 
	\[
	\left|(Y_{t}-X_{t}'\hbeta)-(Y_{t}-X_{t}'\hbeta_{H})\right|=|X_{t}'\Delta|\leq\|X_{t}\|_{\infty}\|\Delta\|_{1}\leq2\kappa_{2}\sqrt{\kappa_{1}sc_{T}K(2K+1)}.
	\]
	
	On the event $\Mcal$, the above bound holds for all $H\in\Hcal$.
	The desired result follows by $P(\Mcal)\geq1-\gamma_{1,T}-\gamma_{2,T}-\gamma_{3,T}$.

\subsection{Proof of Lemma \ref{lem: ridge stability}}
	For notational simplicity, we write $\hbeta$ instead of $\hbeta_{\lambda}$
	and $c_{T}=|H|$. Let $\hOmega=X'X$, $\tOmega=\tilde{X}'\tilde{X}$,
	$\hmu=X'u$ and $\tmu=\tilde{X}'\tilde{u}.$ We work on the event
	on which $\kappa_{1}T\leq\lambda_{\min}(\hOmega)\leq\lambda_{\max}(\hOmega)\leq\kappa_{2}T$,
	$\|\hmu-\tilde{\mu}\|_{2}\leq\kappa_{3}\sqrt{Jc_{T}}$ and $\|\hOmega-\tilde{\Omega}\|\leq\kappa_{4}(c_{T}+J)$.
	Let $q_{1}\geq q_{2}\geq\cdots\geq q_{J}$ denote the eigenvalues
	of $\hOmega$. 
	
	\textbf{Step 1:} show inconsistency.
	
	Notice that 
	\begin{align}
	\hbeta & =(\hOmega+\lambda I)^{-1}X'(X\beta+u)\nonumber \\
	& =(\hOmega+\lambda I)^{-1}\hOmega\beta+(\hOmega+\lambda I)^{-1}X'u\nonumber \\
	& =\beta+\left[(\hOmega+\lambda I)^{-1}\hOmega-I\right]\beta+(\hOmega+\lambda I)^{-1}X'u\nonumber \\
	& =\beta+(\hOmega+\lambda I)^{-1}\left[\hOmega-(\hOmega+\lambda I)\right]\beta+(\hOmega+\lambda I)^{-1}X'u\nonumber \\
	& =\beta-\lambda(\hOmega+\lambda I)^{-1}\beta+(\hOmega+\lambda I)^{-1}X'u.\label{eq: ridge 4}
	\end{align}
	
	Therefore, 
	
	\begin{align*}
	& E\left(\|X(\hbeta-\beta)\|_{2}^{2}\mid X\right)\\
	& =\|\lambda X(\hOmega+\lambda I)^{-1}\beta\|_{2}^{2}+E\left(\|X(\hOmega+\lambda I)^{-1}X'u\|_{2}^{2}\mid X\right)\\
	& =\lambda^{2}\beta'(\hOmega+\lambda I)^{-1}\hOmega(\hOmega+\lambda I)^{-1}\beta+E\left(u'X(\hOmega+\lambda I)^{-1}\hOmega(\hOmega+\lambda I)^{-1}X'u\mid X\right)\\
	& =\lambda^{2}\beta'(\hOmega+\lambda I)^{-1}\hOmega(\hOmega+\lambda I)^{-1}\beta+E\left(\trace\left[(\hOmega+\lambda I)^{-1}\hOmega(\hOmega+\lambda I)^{-1}X'uu'X\right]\mid X\right)\\
	& =\lambda^{2}\beta'(\hOmega+\lambda I)^{-1}\hOmega(\hOmega+\lambda I)^{-1}\beta+\trace\left[(\hOmega+\lambda I)^{-1}\hOmega(\hOmega+\lambda I)^{-1}\hOmega\right].
	\end{align*}
	
	Notice that 
	\[
	\trace\left[(\hOmega+\lambda I)^{-1}\hOmega(\hOmega+\lambda I)^{-1}\hOmega\right]=\sum_{i=1}^{J}q_{i}^{2}(q_{i}+\lambda)^{-2}\gtrsim\frac{T^{2}J}{(T+\lambda)^{2}}
	\]
	and 
	\[
	\lambda^{2}\beta'(\hOmega+\lambda I)^{-1}\hOmega(\hOmega+\lambda I)^{-1}\beta\geq\lambda^{2}\|\beta\|_{2}^{2}\min_{1\leq i\leq J}q_{i}(q_{i}+\lambda)^{-2}\gtrsim\frac{\lambda^{2}T}{(T+\lambda)^{2}}.
	\]
	
	Then we have 
	\[
	E\left[T^{-1}\|X(\hbeta-\beta)\|_{2}^{2}\mid X\right]\gtrsim\frac{TJ+\lambda^{2}}{(T+\lambda)^{2}}\gtrsim\frac{TJ}{(T+\lambda)^{2}}+\frac{\lambda^{2}}{(T+\lambda)^{2}}.
	\]
	
	Since $\lambda\asymp T$, we have $E(\|X(\hbeta-\beta)\|_{2}^{2})\gtrsim T$. 
	
	\textbf{Step 2:} show stability.
	
	Similar to (\ref{eq: ridge 4}), we notice that the perturbed estimator
	would be
	\[
	\tbeta=\beta-\lambda(\tOmega+\lambda I)^{-1}\beta+(\tOmega+\lambda I)^{-1}\tmu.
	\]
	
	Then we bound 
	\begin{align*}
	X_{t}'(\hbeta-\tbeta) & =\lambda X_{t}'\left[(\tOmega+\lambda I)^{-1}-(\hOmega+\lambda I)^{-1}\right]\beta+X_{t}'\left[(\hOmega+\lambda I)^{-1}\hmu-(\tOmega+\lambda I)^{-1}\tmu\right]\\
	& =\lambda X_{t}'\left[(\tOmega+\lambda I)^{-1}-(\hOmega+\lambda I)^{-1}\right]\beta+X_{t}'\left[(\hOmega+\lambda I)^{-1}(\hmu-\tmu)+\left((\hOmega+\lambda I)^{-1}-(\tOmega+\lambda I)^{-1}\right)\tmu\right]\\
	& =X_{t}'\left[(\tOmega+\lambda I)^{-1}-(\hOmega+\lambda I)^{-1}\right]\left(\lambda\beta-\tmu\right)+X_{t}'(\hOmega+\lambda I)^{-1}(\hmu-\tmu).
	\end{align*}
	Now we notice that 
	\begin{align*}
	\|(\tOmega+\lambda I)^{-1}-(\hOmega+\lambda I)^{-1}\| & =\left\Vert (\tOmega+\lambda I)^{-1}\left[(\hOmega+\lambda I)-(\tOmega+\lambda I)\right](\hOmega+\lambda I)^{-1}\right\Vert \\
	& =\|(\tOmega+\lambda I)^{-1}(\hOmega-\tOmega)(\hOmega+\lambda I)^{-1}\|=O_{P}\left(\frac{(J+c_{T})}{(T+\lambda)^{2}}\right)
	\end{align*}
	and $\|\lambda\beta-\tmu\|_{2}\leq\lambda+\|\tmu\|_{2}=O_{P}(\lambda+\sqrt{JT})$.
	This means that 
	\[
	\left|X_{t}'\left[(\tOmega+\lambda I)^{-1}-(\hOmega+\lambda I)^{-1}\right]\left(\lambda\beta-\tmu\right)\right|=O_{P}\left(\frac{(J+c_{T})}{(T+\lambda)^{2}}\left(\lambda+\sqrt{JT}\right)\sqrt{J}\right).
	\]
	
	Moreover, we also have
	\[
	\left|X_{t}'(\hOmega+\lambda I)^{-1}(\hmu-\tmu)\right|\leq O_{P}\left(\sqrt{J}\frac{1}{\mu+\lambda}\sqrt{Jc_{T}}\right)=O_{P}\left(\frac{J\sqrt{c_{T}}}{T+\lambda}\right).
	\]
	
	It follows that 
	\[
	\max_{1\leq t\leq T}|X_{t}'(\hbeta-\tbeta)|=O_{P}\left(\frac{(J+c_{T})}{(T+\lambda)^{2}}\left(\lambda+\sqrt{JT}\right)\sqrt{J}+\frac{J\sqrt{c_{T}}}{T+\lambda}\right).
	\]
	
	Since $J\ll T$ and $T\gtrsim\lambda\gg J^{3/2}$, we have $\max_{1\leq t\leq T}|X_{t}'(\hbeta-\tbeta)|=o_{P}(1)$.

\newpage
\section{Tables and Figures Appendix}

\begin{table}[ht]
\centering

\caption{Size Properties: $F_{2t}\overset{iid}\sim N(0,1)$, $\rho_\epsilon=\rho_u=0$}
\scriptsize
\begin{tabular}{lccccccccc}
\\
\toprule
\midrule

&\multicolumn{9}{c}{DGP1} \\
\cmidrule(l{5pt}r{5pt}){2-10}   \ 
&\multicolumn{3}{c}{Diff-in-Diffs}&\multicolumn{3}{c}{Synthetic Control}&\multicolumn{3}{c}{Constrained Lasso}\\
\cmidrule(l{5pt}r{5pt}){2-4}    \cmidrule(l{5pt}r{5pt}){5-7}    \cmidrule(l{5pt}r{5pt}){8-10}  \ 
$T_0$&$J=20$&$J=50$& $J=100$&$J=20$&$J=50$& $J=100$&$J=20$&$J=50$& $J=100$\\
\midrule

20 & 0.10 & 0.10 & 0.09 & 0.10 & 0.10 & 0.10 & 0.10 & 0.09 & 0.10 \\ 
  50 & 0.10 & 0.09 & 0.10 & 0.10 & 0.10 & 0.10 & 0.09 & 0.10 & 0.10 \\ 
  100 & 0.10 & 0.10 & 0.10 & 0.10 & 0.10 & 0.10 & 0.10 & 0.10 & 0.10 \\

\midrule

&\multicolumn{9}{c}{DGP2} \\
\cmidrule(l{5pt}r{5pt}){2-10}   \ 
&\multicolumn{3}{c}{Diff-in-Diffs}&\multicolumn{3}{c}{Synthetic Control}&\multicolumn{3}{c}{Constrained Lasso}\\
\cmidrule(l{5pt}r{5pt}){2-4}    \cmidrule(l{5pt}r{5pt}){5-7}    \cmidrule(l{5pt}r{5pt}){8-10}  \ 
$T_0$&$J=20$&$J=50$& $J=100$&$J=20$&$J=50$& $J=100$&$J=20$&$J=50$& $J=100$\\
\midrule

20 & 0.10 & 0.09 & 0.09 & 0.10 & 0.09 & 0.09 & 0.10 & 0.09 & 0.10 \\ 
  50 & 0.10 & 0.10 & 0.10 & 0.09 & 0.10 & 0.10 & 0.09 & 0.10 & 0.10 \\ 
  100 & 0.11 & 0.10 & 0.10 & 0.11 & 0.10 & 0.10 & 0.11 & 0.10 & 0.10 \\ 

\midrule

&\multicolumn{9}{c}{DGP3} \\
\cmidrule(l{5pt}r{5pt}){2-10}   \ 
&\multicolumn{3}{c}{Diff-in-Diffs}&\multicolumn{3}{c}{Synthetic Control}&\multicolumn{3}{c}{Constrained Lasso}\\
\cmidrule(l{5pt}r{5pt}){2-4}    \cmidrule(l{5pt}r{5pt}){5-7}    \cmidrule(l{5pt}r{5pt}){8-10}  \ 
$T_0$&$J=20$&$J=50$& $J=100$&$J=20$&$J=50$& $J=100$&$J=20$&$J=50$& $J=100$\\
\midrule

20 & 0.10 & 0.09 & 0.09 & 0.09 & 0.10 & 0.09 & 0.11 & 0.09 & 0.09 \\ 
  50 & 0.10 & 0.10 & 0.09 & 0.10 & 0.10 & 0.10 & 0.10 & 0.10 & 0.10 \\ 
  100 & 0.10 & 0.09 & 0.10 & 0.10 & 0.09 & 0.10 & 0.10 & 0.10 & 0.09 \\ 

\midrule

&\multicolumn{9}{c}{DGP4} \\
\cmidrule(l{5pt}r{5pt}){2-10}   \ 
&\multicolumn{3}{c}{Diff-in-Diffs}&\multicolumn{3}{c}{Synthetic Control}&\multicolumn{3}{c}{Constrained Lasso}\\
\cmidrule(l{5pt}r{5pt}){2-4}    \cmidrule(l{5pt}r{5pt}){5-7}    \cmidrule(l{5pt}r{5pt}){8-10}  \ 
$T_0$&$J=20$&$J=50$& $J=100$&$J=20$&$J=50$& $J=100$&$J=20$&$J=50$& $J=100$\\
\midrule

20 & 0.10 & 0.10 & 0.10 & 0.10 & 0.09 & 0.10 & 0.10 & 0.10 & 0.10 \\ 
  50 & 0.11 & 0.10 & 0.10 & 0.10 & 0.10 & 0.10 & 0.10 & 0.10 & 0.10 \\ 
  100 & 0.10 & 0.10 & 0.10 & 0.10 & 0.10 & 0.10 & 0.10 & 0.10 & 0.11 \\

\midrule
\bottomrule
\multicolumn{10}{p{12cm}}{\scriptsize{\it Notes:} Simulation design as described in the main text. Nominal level $\alpha= 0.1$. Based on simulations with $5000$ repetitions.}
\end{tabular}
\label{tab:size_iid}
\end{table}
\normalsize

\begin{table}[H]
\centering
\caption{Size Properties: $F_{2t}\overset{iid}\sim N(0,1)$, $\rho_\epsilon=\rho_u=0.6$}
\scriptsize
\begin{tabular}{lccccccccc}
\\
\toprule
\midrule

&\multicolumn{9}{c}{DGP1} \\
\cmidrule(l{5pt}r{5pt}){2-10}   \ 
&\multicolumn{3}{c}{Diff-in-Diffs}&\multicolumn{3}{c}{Synthetic Control}&\multicolumn{3}{c}{Constrained Lasso}\\
\cmidrule(l{5pt}r{5pt}){2-4}    \cmidrule(l{5pt}r{5pt}){5-7}    \cmidrule(l{5pt}r{5pt}){8-10}  \ 
$T_0$&$J=20$&$J=50$& $J=100$&$J=20$&$J=50$& $J=100$&$J=20$&$J=50$& $J=100$\\
\midrule

20 & 0.13 & 0.13 & 0.12 & 0.12 & 0.11 & 0.10 & 0.12 & 0.11 & 0.11 \\ 
  50 & 0.11 & 0.11 & 0.10 & 0.12 & 0.12 & 0.12 & 0.12 & 0.12 & 0.12 \\ 
  100 & 0.11 & 0.10 & 0.10 & 0.11 & 0.11 & 0.12 & 0.12 & 0.12 & 0.11 \\ 

\midrule

&\multicolumn{9}{c}{DGP2} \\
\cmidrule(l{5pt}r{5pt}){2-10}   \ 
&\multicolumn{3}{c}{Diff-in-Diffs}&\multicolumn{3}{c}{Synthetic Control}&\multicolumn{3}{c}{Constrained Lasso}\\
\cmidrule(l{5pt}r{5pt}){2-4}    \cmidrule(l{5pt}r{5pt}){5-7}    \cmidrule(l{5pt}r{5pt}){8-10}  \ 
$T_0$&$J=20$&$J=50$& $J=100$&$J=20$&$J=50$& $J=100$&$J=20$&$J=50$& $J=100$\\
\midrule

20 & 0.12 & 0.12 & 0.11 & 0.12 & 0.11 & 0.11 & 0.12 & 0.11 & 0.11 \\ 
  50 & 0.11 & 0.11 & 0.11 & 0.11 & 0.12 & 0.12 & 0.12 & 0.11 & 0.12 \\ 
  100 & 0.10 & 0.10 & 0.11 & 0.11 & 0.11 & 0.12 & 0.11 & 0.11 & 0.12 \\ 
\midrule

&\multicolumn{9}{c}{DGP3} \\
\cmidrule(l{5pt}r{5pt}){2-10}   \ 
&\multicolumn{3}{c}{Diff-in-Diffs}&\multicolumn{3}{c}{Synthetic Control}&\multicolumn{3}{c}{Constrained Lasso}\\
\cmidrule(l{5pt}r{5pt}){2-4}    \cmidrule(l{5pt}r{5pt}){5-7}    \cmidrule(l{5pt}r{5pt}){8-10}  \ 
$T_0$&$J=20$&$J=50$& $J=100$&$J=20$&$J=50$& $J=100$&$J=20$&$J=50$& $J=100$\\
\midrule

20 & 0.10 & 0.10 & 0.10 & 0.10 & 0.09 & 0.09 & 0.11 & 0.10 & 0.10 \\ 
  50 & 0.10 & 0.10 & 0.10 & 0.10 & 0.10 & 0.10 & 0.12 & 0.12 & 0.12 \\ 
  100 & 0.10 & 0.10 & 0.10 & 0.10 & 0.10 & 0.10 & 0.12 & 0.13 & 0.12 \\ 

\midrule

&\multicolumn{9}{c}{DGP4} \\
\cmidrule(l{5pt}r{5pt}){2-10}   \ 
&\multicolumn{3}{c}{Diff-in-Diffs}&\multicolumn{3}{c}{Synthetic Control}&\multicolumn{3}{c}{Constrained Lasso}\\
\cmidrule(l{5pt}r{5pt}){2-4}    \cmidrule(l{5pt}r{5pt}){5-7}    \cmidrule(l{5pt}r{5pt}){8-10}  \ 
$T_0$&$J=20$&$J=50$& $J=100$&$J=20$&$J=50$& $J=100$&$J=20$&$J=50$& $J=100$\\
\midrule

20 & 0.11 & 0.11 & 0.10 & 0.11 & 0.10 & 0.10 & 0.12 & 0.11 & 0.12 \\ 
  50 & 0.11 & 0.11 & 0.11 & 0.11 & 0.11 & 0.11 & 0.12 & 0.12 & 0.12 \\ 
  100 & 0.10 & 0.10 & 0.11 & 0.11 & 0.11 & 0.11 & 0.11 & 0.11 & 0.11 \\ 
  
\midrule
\bottomrule
\multicolumn{10}{p{12cm}}{\scriptsize{\it Notes:} Simulation design as described in the main text. Nominal level $\alpha= 0.1$. Based on simulations with $5000$ repetitions.}
\end{tabular}
\label{tab:size_wd}
\end{table}
\normalsize

\begin{table}[H]
\centering

\caption{Size Properties: $F_{2t}\sim N(t,1)$, $\rho_\epsilon=\rho_u=0$}
\scriptsize
\begin{tabular}{lccccccccc}
\\
\toprule
\midrule

&\multicolumn{9}{c}{DGP1} \\
\cmidrule(l{5pt}r{5pt}){2-10}   \ 
&\multicolumn{3}{c}{Diff-in-Diffs}&\multicolumn{3}{c}{Synthetic Control}&\multicolumn{3}{c}{Constrained Lasso}\\
\cmidrule(l{5pt}r{5pt}){2-4}    \cmidrule(l{5pt}r{5pt}){5-7}    \cmidrule(l{5pt}r{5pt}){8-10}  \ 
$T_0$&$J=20$&$J=50$& $J=100$&$J=20$&$J=50$& $J=100$&$J=20$&$J=50$& $J=100$\\
\midrule

20 & 0.10 & 0.10 & 0.10 & 0.08 & 0.09 & 0.09 & 0.07 & 0.08 & 0.09 \\ 
  50 & 0.11 & 0.09 & 0.10 & 0.09 & 0.09 & 0.09 & 0.09 & 0.08 & 0.09 \\ 
  100 & 0.10 & 0.10 & 0.10 & 0.10 & 0.09 & 0.09 & 0.10 & 0.10 & 0.09 \\ 

\midrule

&\multicolumn{9}{c}{DGP2} \\
\cmidrule(l{5pt}r{5pt}){2-10}   \ 
&\multicolumn{3}{c}{Diff-in-Diffs}&\multicolumn{3}{c}{Synthetic Control}&\multicolumn{3}{c}{Constrained Lasso}\\
\cmidrule(l{5pt}r{5pt}){2-4}    \cmidrule(l{5pt}r{5pt}){5-7}    \cmidrule(l{5pt}r{5pt}){8-10}  \ 
$T_0$&$J=20$&$J=50$& $J=100$&$J=20$&$J=50$& $J=100$&$J=20$&$J=50$& $J=100$\\
\midrule

20 & 0.42 & 0.43 & 0.43 & 0.09 & 0.10 & 0.09 & 0.08 & 0.08 & 0.08 \\ 
  50 & 0.71 & 0.75 & 0.76 & 0.10 & 0.10 & 0.10 & 0.10 & 0.09 & 0.09 \\ 
  100 & 0.94 & 0.96 & 0.96 & 0.10 & 0.10 & 0.09 & 0.10 & 0.09 & 0.09 \\ 

\midrule

&\multicolumn{9}{c}{DGP3} \\
\cmidrule(l{5pt}r{5pt}){2-10}   \ 
&\multicolumn{3}{c}{Diff-in-Diffs}&\multicolumn{3}{c}{Synthetic Control}&\multicolumn{3}{c}{Constrained Lasso}\\
\cmidrule(l{5pt}r{5pt}){2-4}    \cmidrule(l{5pt}r{5pt}){5-7}    \cmidrule(l{5pt}r{5pt}){8-10}  \ 
$T_0$&$J=20$&$J=50$& $J=100$&$J=20$&$J=50$& $J=100$&$J=20$&$J=50$& $J=100$\\
\midrule

20 & 0.45 & 0.43 & 0.44 & 0.51 & 0.48 & 0.49 & 0.08 & 0.09 & 0.08 \\ 
  50 & 0.78 & 0.75 & 0.78 & 0.82 & 0.80 & 0.80 & 0.09 & 0.09 & 0.09 \\ 
  100 & 0.97 & 0.97 & 0.97 & 0.98 & 0.97 & 0.98 & 0.09 & 0.10 & 0.10 \\ 

\midrule

&\multicolumn{9}{c}{DGP4} \\
\cmidrule(l{5pt}r{5pt}){2-10}   \ 
&\multicolumn{3}{c}{Diff-in-Diffs}&\multicolumn{3}{c}{Synthetic Control}&\multicolumn{3}{c}{Constrained Lasso}\\
\cmidrule(l{5pt}r{5pt}){2-4}    \cmidrule(l{5pt}r{5pt}){5-7}    \cmidrule(l{5pt}r{5pt}){8-10}  \ 
$T_0$&$J=20$&$J=50$& $J=100$&$J=20$&$J=50$& $J=100$&$J=20$&$J=50$& $J=100$\\
\midrule

20 & 0.35 & 0.34 & 0.33 & 0.20 & 0.13 & 0.12 & 0.08 & 0.07 & 0.08 \\ 
  50 & 0.65 & 0.61 & 0.60 & 0.39 & 0.19 & 0.14 & 0.09 & 0.10 & 0.09 \\ 
  100 & 0.89 & 0.86 & 0.86 & 0.61 & 0.33 & 0.21 & 0.09 & 0.10 & 0.10 \\ 

\midrule
\bottomrule
\multicolumn{10}{p{12cm}}{\scriptsize{\it Notes:} Simulation design as described in the main text. Nominal level $\alpha= 0.1$. Based on simulations with $5000$ repetitions.}
\end{tabular}
\label{tab:size_iid_nonstat}
\end{table}
\normalsize

\begin{table}[H]
\centering

\caption{Size Properties: $F_{2t}\sim N(t,1)$, $\rho_\epsilon=\rho_u=0.6$}
\scriptsize
\begin{tabular}{lccccccccc}
\\
\toprule
\midrule

&\multicolumn{9}{c}{DGP1} \\
\cmidrule(l{5pt}r{5pt}){2-10}   \ 
&\multicolumn{3}{c}{Diff-in-Diffs}&\multicolumn{3}{c}{Synthetic Control}&\multicolumn{3}{c}{Constrained Lasso}\\
\cmidrule(l{5pt}r{5pt}){2-4}    \cmidrule(l{5pt}r{5pt}){5-7}    \cmidrule(l{5pt}r{5pt}){8-10}  \ 
$T_0$&$J=20$&$J=50$& $J=100$&$J=20$&$J=50$& $J=100$&$J=20$&$J=50$& $J=100$\\
\midrule
20 & 0.12 & 0.12 & 0.12 & 0.10 & 0.09 & 0.09 & 0.10 & 0.09 & 0.08 \\ 
  50 & 0.10 & 0.12 & 0.11 & 0.11 & 0.11 & 0.11 & 0.12 & 0.12 & 0.11 \\ 
  100 & 0.10 & 0.11 & 0.10 & 0.11 & 0.12 & 0.11 & 0.11 & 0.12 & 0.12 \\ 

\midrule

&\multicolumn{9}{c}{DGP2} \\
\cmidrule(l{5pt}r{5pt}){2-10}   \ 
&\multicolumn{3}{c}{Diff-in-Diffs}&\multicolumn{3}{c}{Synthetic Control}&\multicolumn{3}{c}{Constrained Lasso}\\
\cmidrule(l{5pt}r{5pt}){2-4}    \cmidrule(l{5pt}r{5pt}){5-7}    \cmidrule(l{5pt}r{5pt}){8-10}  \ 
$T_0$&$J=20$&$J=50$& $J=100$&$J=20$&$J=50$& $J=100$&$J=20$&$J=50$& $J=100$\\
\midrule
20 & 0.44 & 0.47 & 0.49 & 0.12 & 0.12 & 0.13 & 0.10 & 0.10 & 0.10 \\ 
  50 & 0.74 & 0.77 & 0.77 & 0.11 & 0.11 & 0.11 & 0.12 & 0.12 & 0.12 \\ 
  100 & 0.95 & 0.96 & 0.96 & 0.11 & 0.11 & 0.11 & 0.11 & 0.11 & 0.11 \\ 

\midrule

&\multicolumn{9}{c}{DGP3} \\
\cmidrule(l{5pt}r{5pt}){2-10}   \ 
&\multicolumn{3}{c}{Diff-in-Diffs}&\multicolumn{3}{c}{Synthetic Control}&\multicolumn{3}{c}{Constrained Lasso}\\
\cmidrule(l{5pt}r{5pt}){2-4}    \cmidrule(l{5pt}r{5pt}){5-7}    \cmidrule(l{5pt}r{5pt}){8-10}  \ 
$T_0$&$J=20$&$J=50$& $J=100$&$J=20$&$J=50$& $J=100$&$J=20$&$J=50$& $J=100$\\
\midrule
20 & 0.45 & 0.45 & 0.45 & 0.53 & 0.51 & 0.53 & 0.10 & 0.09 & 0.09 \\ 
  50 & 0.78 & 0.78 & 0.78 & 0.83 & 0.82 & 0.81 & 0.12 & 0.12 & 0.12 \\ 
  100 & 0.97 & 0.97 & 0.97 & 0.98 & 0.97 & 0.97 & 0.12 & 0.11 & 0.11 \\ 

\midrule

&\multicolumn{9}{c}{DGP4} \\
\cmidrule(l{5pt}r{5pt}){2-10}   \ 
&\multicolumn{3}{c}{Diff-in-Diffs}&\multicolumn{3}{c}{Synthetic Control}&\multicolumn{3}{c}{Constrained Lasso}\\
\cmidrule(l{5pt}r{5pt}){2-4}    \cmidrule(l{5pt}r{5pt}){5-7}    \cmidrule(l{5pt}r{5pt}){8-10}  \ 
$T_0$&$J=20$&$J=50$& $J=100$&$J=20$&$J=50$& $J=100$&$J=20$&$J=50$& $J=100$\\
\midrule

20 & 0.38 & 0.37 & 0.36 & 0.21 & 0.14 & 0.12 & 0.10 & 0.10 & 0.10 \\ 
  50 & 0.66 & 0.63 & 0.63 & 0.39 & 0.20 & 0.14 & 0.12 & 0.12 & 0.11 \\ 
  100 & 0.89 & 0.86 & 0.86 & 0.62 & 0.34 & 0.20 & 0.10 & 0.11 & 0.11 \\ 

\midrule
\bottomrule
\multicolumn{10}{p{12cm}}{\scriptsize{\it Notes:} Simulation design as described in the main text. Nominal level $\alpha= 0.1$. Based on simulations with $5000$ repetitions.}
\end{tabular}
\label{tab:size_wd_nonstat}
\end{table}
\normalsize

\begin{figure}[H]
\caption{Finite Sample Size Properties ($\alpha=0.1$)}
\begin{center}
\includegraphics[width=0.6\textwidth,trim={0 1.5cm 0 2.75cm}]{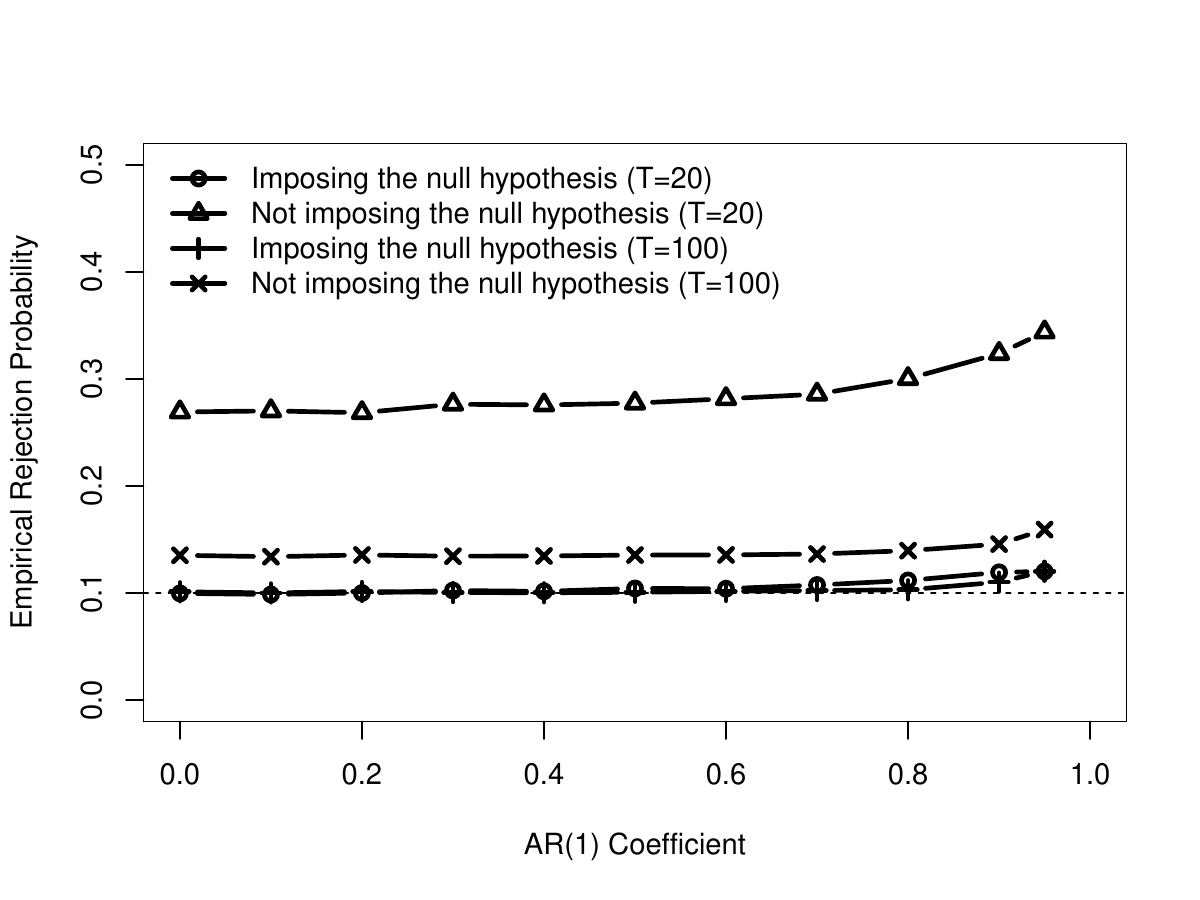}
\end{center}
    {\footnotesize \textit{Notes:} Empirical rejection probability from testing $H_0:\theta_{T}=0$. The data are generated as $Y^N_{1t}=\sum_{j=2}^{J+1}w_jY_{jt}^N+u_t$, where $Y_{jt}^N\sim N(0,1)$ is iid across $(j,t)$, $\{u_t\}$ is a Gaussian AR(1) process, $(w_2,\dots,w_{J+1})'=(1/3,1/3,1/3,0,\dots,0)'$, and $J=50$. The weights are estimated using the canonical SC method (cf. Section \ref{ex:sc}).}
\label{fig:imposing_H0}
\end{figure}

\begin{figure}[H]
\caption{Power Curves}
\begin{center}
\includegraphics[width=0.375\textwidth,trim={0 1cm 0 1.75cm}]{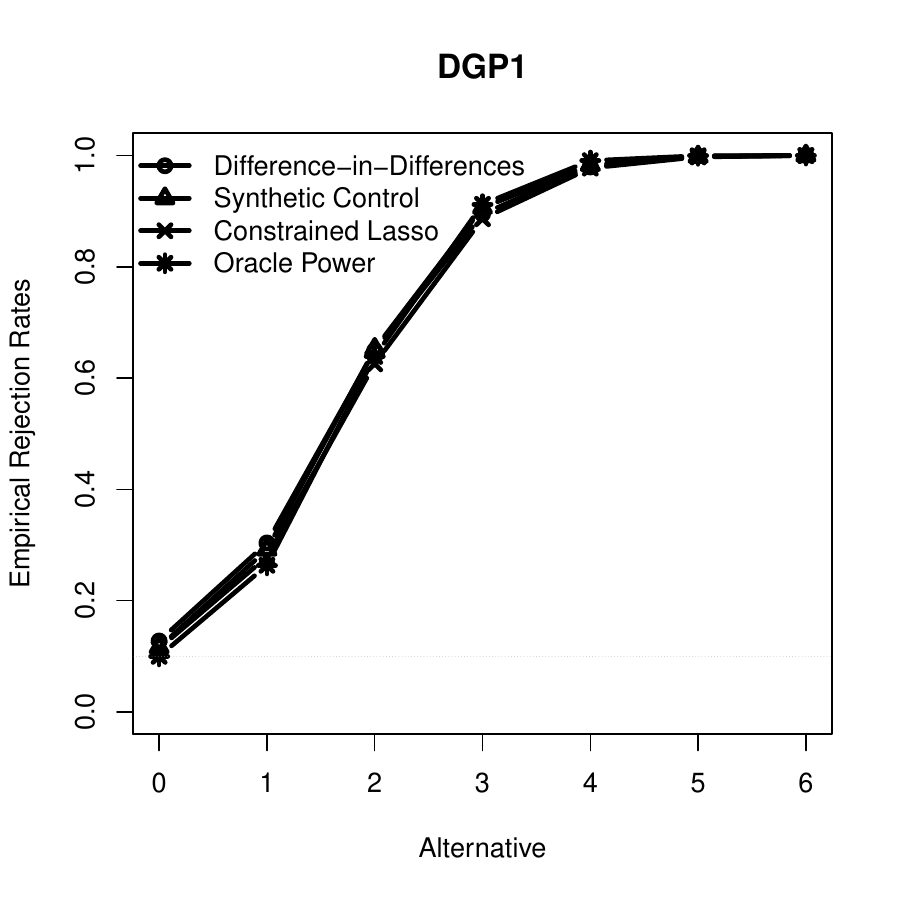}
\includegraphics[width=0.375\textwidth,trim={0 1cm 0 1.75cm}]{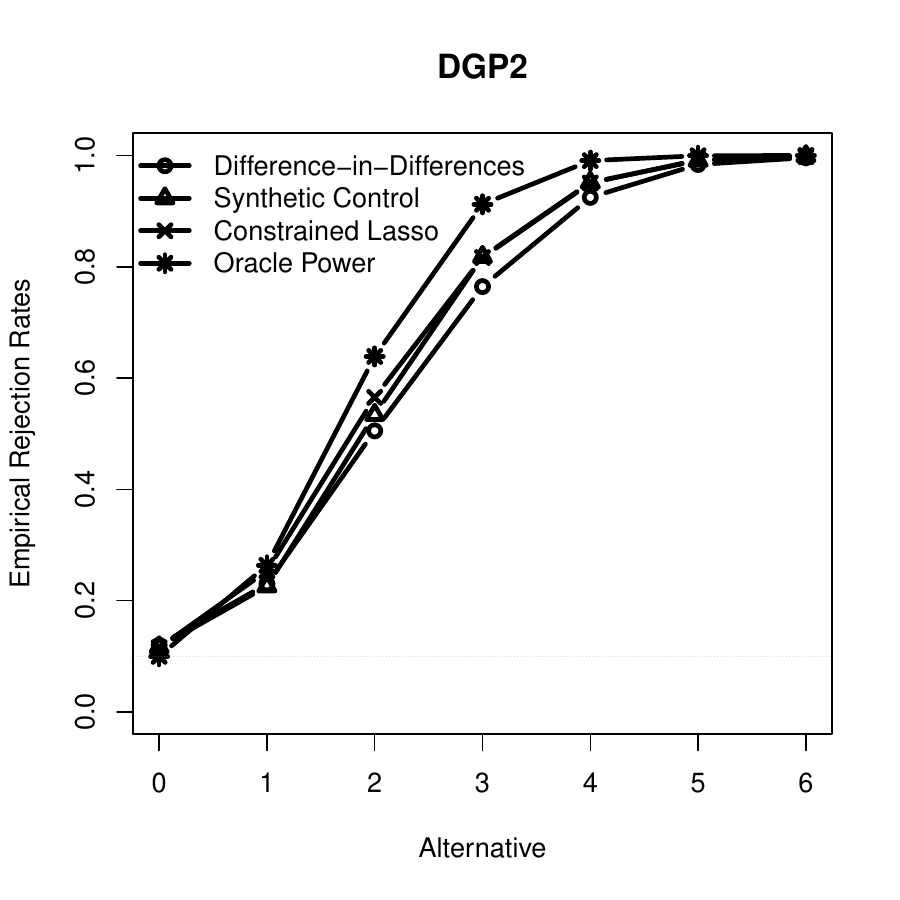}
\includegraphics[width=0.375\textwidth,trim={0 1.75cm 0 0}]{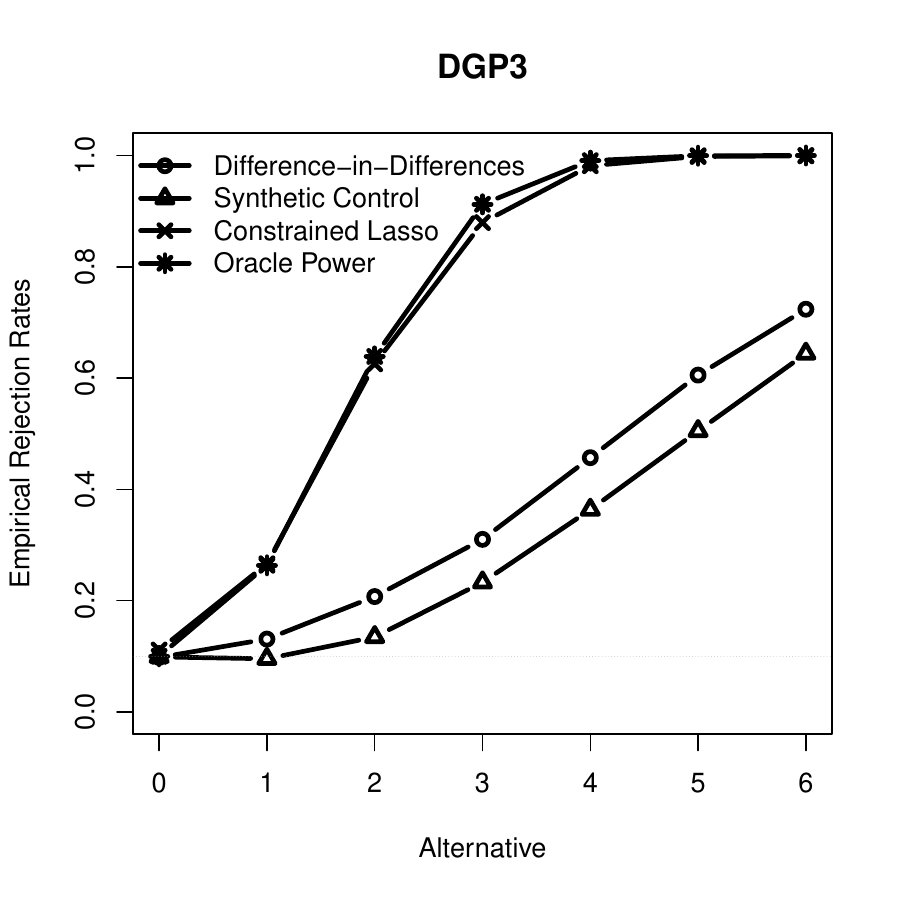}
\includegraphics[width=0.375\textwidth,trim={0 1.75cm 0 0}]{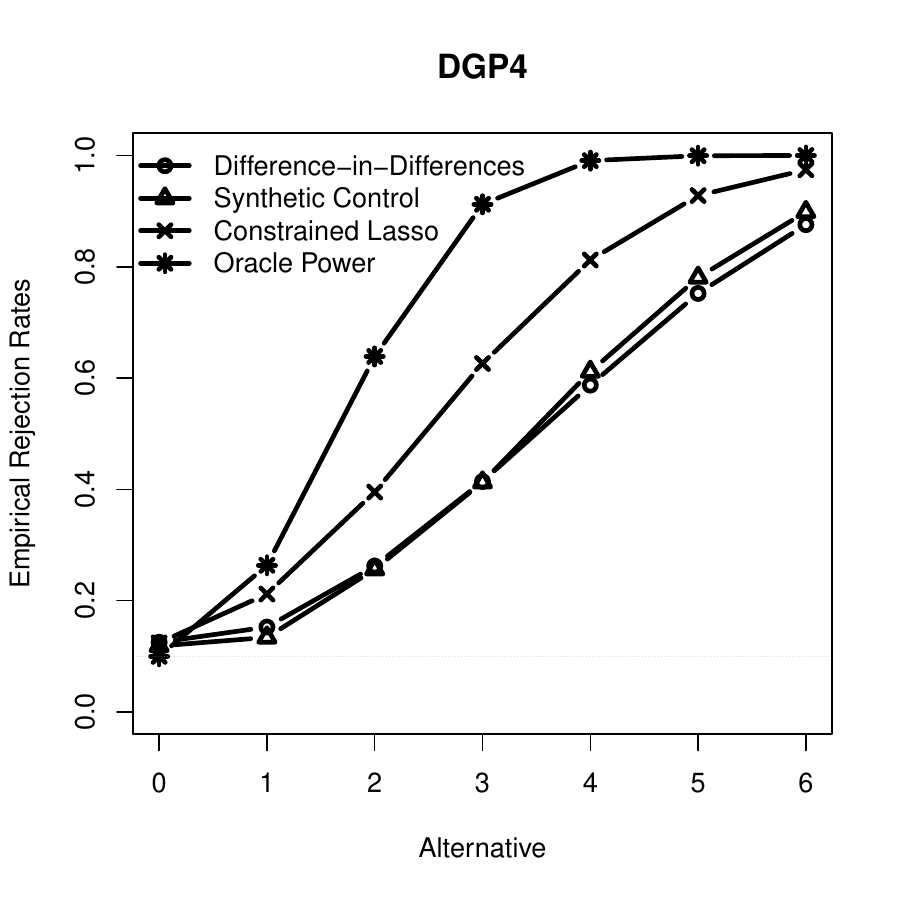}
\end{center}
    {\footnotesize \textit{Notes:} Simulation design as described in the main text. Nominal level $\alpha= 0.1$. Based on simulations with $5000$ repetitions.}
\label{fig:pcs}
\end{figure}

\newpage

\end{document}